\documentclass[journal,onecolumn]{IEEEtran}

\usepackage{array}

\usepackage{graphicx}
\usepackage{dcolumn}
\usepackage{bm}

\usepackage{hyperref}
\usepackage{textcomp}
\usepackage{csquotes}

\usepackage{amssymb}
\usepackage{amsthm}
\usepackage{amsfonts}
\usepackage{amssymb}
\usepackage{graphicx}
\usepackage{comment}

\usepackage{subfiles}

\usepackage{subcaption}
\usepackage{mathtools}
\usepackage{pgfplots}
\pgfplotsset{/pgf/number format/use comma,compat=newest}
\usepackage{url}
\usepackage{tikz}
\usepackage{bbm}
\usepackage{floatflt}

\usepackage{cite}

\usepackage[most]{tcolorbox}

\newtcolorbox{rsfunctionbox}[1]{
breakable,
title={#1},
fonttitle=\bfseries,
colback=white,
colframe=black!50,
colbacktitle=black!5,
coltitle=black,
boxrule=0.5pt,
arc=0pt,
left=5pt,
right=5pt,
top=5pt,
bottom=5pt,
before skip=10pt,
after skip=10pt
}

\newenvironment{talign*}
{\csname align*\endcsname}
{\endalign}

\newcommand{\Sf}{\mathsf}

\newcommand{\C}{\mathbb{C}}

\newcommand{\E}{\mathbb{E}}

\newcommand{\R}{\mathbb{R}}

\newcommand{\cC}{\mathcal{C}}
\newcommand{\cD}{\mathcal{D}}

\newcommand{\cG}{\mathcal{G}}
\newcommand{\cH}{\mathcal{H}}
\newcommand{\cI}{\mathcal{I}}
\newcommand{\cJ}{\mathcal{J}}

\newcommand{\cL}{\mathcal{L}}

\newcommand{\cQ}{\mathcal{Q}}

\newcommand{\cM}{\mathcal{M}}
\newcommand{\cN}{\mathcal{N}}

\newcommand{\cR}{\mathcal{R}}
\newcommand{\cS}{\mathcal{S}}

\newcommand{\Tr}{\Sf{Tr}}

\newcommand{\bu}{\bm{u}}

\newcommand{\bv}{\bm{v}}
\newcommand{\bs}{\bm{s}}

\newcommand{\bx}{\bm{x}}
\newcommand{\by}{\bm{y}}

\newcommand{\bA}{\bm{A}}
\newcommand{\bB}{\bm{B}}
\newcommand{\bC}{\bm{C}}

\newcommand{\bG}{\bm{G}}

\newcommand{\bI}{\bm{I}}

\newcommand{\bM}{\bm{M}}

\newcommand{\bP}{\bm{P}}
\newcommand{\bQ}{\bm{Q}}

\newcommand{\bU}{\bm{U}}
\newcommand{\bV}{\bm{V}}

\newcommand{\bX}{\bm{X}}
\newcommand{\bY}{\bm{Y}}
\newcommand{\bZ}{\bm{Z}}

\usepackage{algorithm}
\usepackage{algpseudocode}

\newcommand{\Diag}{\Sf{D}}

\newcommand{\dd}{\Sf{d}}

\newcommand{\sh}{\Sf{h}}

\newcommand{\su}{\Sf{u}}
\newcommand{\sv}{\Sf{v}}

\newcommand{\sx}{\Sf{x}}
\newcommand{\sy}{\Sf{y}}
\newcommand{\sz}{\Sf{z}}

\newcommand{\Extr}{\mathop{\operatorname{extr}}}
\newcommand{\imagUnit}{\mathrm{i}}
\newcommand{\Normal}{\mathcal{N}}

\newcommand{\MMSE}{\Sf{mmse}}
\newcommand{\Eye}{\bm{I}}

\newtheorem{lemma}{Lemma}
\newtheorem{result}{Result}

\begin{document}

\title{Information-Theoretic Limits for\\ Rectangular Spiked Matrix Models\\ with Rotationally Invariant Noise}

\author{
Qun Chen$^{1,\dagger}$, Rudy Skerk$^{2,\dagger}$, Jean Barbier$^{3,\ddagger}$, Francesco Camilli$^{4,\ddagger}$, Junjie Ma$^{5,\ddagger}$, Haochuan Zhang$^{1,\ddagger}$
}
\makeatletter
\def\blfootnote{\gdef\@thefnmark{}\@footnotetext}
\makeatother
\blfootnote{$^{1}$ School of Automation, Guangdong University of Technology, Guangzhou 510006, China.}
\blfootnote{$^{2}$ International School for Advanced Studies (SISSA), Via Bonomea 265, 34136 Trieste, Italy.}
\blfootnote{$^{3}$ The Abdus Salam International Centre for Theoretical Physics (ICTP), Strada Costiera 11, 34151 Trieste, Italy.}
\blfootnote{$^{4}$ Dipartimento di Matematica, Universit\`a di Bologna, Piazza di Porta San Donato 5, 40126 Bologna, Italy.}
\blfootnote{$^{5}$ State Key Laboratory of Mathematical Sciences, Academy of Mathematics and Systems Science, Chinese Academy of Sciences, Beijing 100190, China.}
\blfootnote{$^{\dagger}$ These authors contributed equally to this work (co-first authors). Emails: \texttt{c.prey.q@gmail.com}, \texttt{rskerk@sissa.it}.}
\blfootnote{$^{\ddagger}$ These authors are the corresponding authors, listed in alphabetical order. Emails: \texttt{jbarbier@ictp.it}, \texttt{francesco.camilli2@unibo.it}, \texttt{majunjie@lsec.cc.ac.cn}, \texttt{haochuan.zhang@gdut.edu.cn}}
\makeatletter
\def\blfootnote{\gdef\@thefnmark{}\@footnotetext}
\makeatother

\maketitle

\begin{abstract}

We study the information-theoretic limits of inference in rectangular rank-one spiked matrix models with additive rotationally invariant noise, thereby going beyond the standard assumption of independent noise entries. Using the replica method from statistical physics, we derive a variational formula for the mutual information between the signal and the data, together with a system of saddle point equations predicting the minimum mean squared error in estimating the rank-one spike. We further show that, in the Bayes-optimal setting, these saddle point equations coincide with the fixed point equations of the state evolution of a recently proposed orthogonal approximate message passing algorithm with Bayes-optimal denoisers. This equivalence provides an algorithmic interpretation of the replica prediction and supports the Bayes-optimality of orthogonal approximate message passing whenever the relevant fixed point is unique. Numerical evaluations of the saddle point equations for several rotationally invariant noise ensembles and signal priors, together with finite-size simulations of the algorithm, confirm the asymptotic predictions. They illustrate both regimes with a unique informative fixed point, where orthogonal approximate message passing is predicted to attain the minimum mean squared error, and regimes with multiple fixed points, where a statistical-to-computational gap exists.

\end{abstract}


\section{Introduction}

Modern machine learning relies on increasingly large high-dimensional datasets and on identifying the low-dimensional structures hidden within them. Principal Component Analysis (PCA) remains one of the simplest and most widely used spectral tools for this purpose, often serving as a preprocessing step. PCA identifies the leading eigenvectors of empirical covariance matrices (or the leading singular vectors of non-symmetric matrices), thereby producing a low-dimensional subspace that captures the dominant variance in the data. To analyse PCA and related spectral methods from an information-theoretic (IT) perspective, one often considers low-rank signals corrupted by full-rank noise. A canonical example is the rectangular rank-one spiked matrix model \cite{johnstone2001distribution,johnstone2009consistency}
\begin{align}
\bY&=\frac{\lambda}{N}\bu_{0}\bv_{0}^{\top}+\bZ_{0},\label{eq:model_def}
\end{align}
where $\bu_{0}\in\R^{N}$ and $\bv_{0}\in\R^{M}$ are a pair of unknown signals, $\lambda$ is a given signal-to-noise ratio (SNR), and $\bZ_{0}$ is an additive random noise matrix. The recoverability of the rank-one spike $\bu_{0}\bv_{0}^{\top}$ from the observation $\bY$ depends on the SNR, the signal priors, and the statistical structure of the noise. Beyond its role as a model for PCA, the rectangular spiked model arises in several high-dimensional inference problems. It provides an effective model for Gaussian mixture clustering \cite{lesieur2016phase} and captures spectral aspects of community detection in stochastic block models \cite{deshpande2015asymptotic,Lelarge2019Fundamental}. Related rectangular spiked formulations also appear in financial data analysis \cite{plerou2002random,bouchaud2009financial} and rectangular community-detection problems \cite{krzakala2013spectral,abbe2015exact}. Low-rank-plus-noise structures also arise in analyses of feature learning in two-layer neural networks \cite{ba2022highdimensionalasymptoticsfeaturelearning,dandi2023twolayerneuralnetworkslearn,cui2024asymptoticsfeaturelearningtwolayer}. These examples motivate a unified study of theoretical limits and algorithmic performance in the rectangular setting.


Model \eqref{eq:model_def} generalizes the classical setting in which both $\bZ_{0}$ and $\bv_{0}$ are independently Gaussian distributed. Specifically, it extends the well-studied Spiked-Wishart (or spiked covariance) model \cite{johnstone2001distribution} by interpreting both $\bu_{0}$ and $\bv_{0}$ as signals with general prior laws and allowing $\bZ_{0}$ to be rotationally invariant (RI). For model \eqref{eq:model_def}, the performance of the spectral PCA estimator as a function of the SNR is governed by the Baik, Ben Arous, P\'ech\'e (BBP) transition \cite{baik2005phase,feral2007largest,paul2007asymptotics,bai2010spectral,benaych2011eigenvalues,benaych2012singular}: above a critical SNR, an outlier singular value of $\bY$ separates from the bulk of singular values associated with the noise, and the corresponding singular vectors become correlated with the ground-truth signal pair; below this critical SNR, no outlier appears, and the correlation between any singular vector of the data and the signals vanishes asymptotically as $N\to\infty$. Beyond spectral methods, an Approximate Message Passing (AMP) algorithm for inference in model \eqref{eq:model_def} has been developed in \cite{rangan2012iterative,deshpande2014information,lesieur2017constrained,montanari2021estimation}. More broadly, AMP-type algorithms form a powerful class of Bayesian iterative methods that can exploit prior information in high-dimensional inference problems \cite{donoho2009message,dia2016mutual,krzakala2016mutual,Ma2017OAMP,maillard2019hightemp,fan2022approximate,wang2022universality,dudeja2023universality,lovig2025universality,Chen2026orthogonal,dudeja2026optimality}. A key feature of these algorithms is that their high-dimensional dynamics are described by low-dimensional state evolution (SE) recursions, which provide exact asymptotic characterizations of their algorithmic performance \cite{bayati2011dynamics}.

For spiked matrix models, the central IT questions are to characterize the fundamental limits of inference, such as the minimum mean square error (MMSE), and to determine whether efficient algorithms, such as AMP, can attain them. Although rigorous analyses of these limits are available in \cite{dia2016mutual,BarbierM17a,miolane2017fundamental,Lelarge2019Fundamental,macris2020all,mourrat2020hamilton}, they primarily concern Gaussian spiked settings. The replica method, originating in statistical physics \cite{mezard1987spin,mezard2009information}, has therefore been adopted by the information-theory community as a more general approach that provides non-rigorous but tractable predictions of these limits \cite{lesieur2015mmse}. When no algorithmic gap is present, these IT predictions are consistent with the SE predictions of the AMP algorithms \cite{rangan2012iterative,deshpande2014information,montanari2021estimation}, thus supporting the optimality of AMP in such settings.

Although the fundamental limits of spiked matrix models are now well understood, including in tensorial cases \cite{lesieur2017statistical,barbier2017layered,barbier2019mutual_even,reeves2020information}, these results rely strongly on the Gaussian, or more generally entrywise independent, structure of the noise. Some universality phenomena extend beyond the i.i.d.\ setting; for example, \cite{Guionnet_inhomogeneous_noise} shows that a broad class of spiked models with inhomogeneous noise can be mapped to equivalent Gaussian models. However, such reductions no longer apply to genuinely correlated noise, where non-trivial dependencies across matrix entries can substantially affect the spectral structure of the observation. RI random noise ensembles offer a tractable way to model such correlations. In the rectangular setting, rotational invariance means that the left and right singular vector matrices of the noise are Haar distributed, i.e., uniformly drawn from the orthogonal groups of the appropriate dimensions, while the singular values encode its spectral structure. The model thus departs from entrywise independence while remaining analytically tractable, with the final results depending on the noise distribution only through its limiting spectral distribution. Such ensembles provide a classical generalization of i.i.d.\ random matrices and a natural framework for studying correlated noise in high-dimensional inference \cite{Livan_2018,potters2020first}.

A subset of the present authors initiated the analysis of IT and algorithmic limits for spiked matrix models with non-independent RI noise in the simpler symmetric setting \cite{barbier2023fundamental,barbier2025information}. In particular, these works introduced an AMP algorithm that could strictly outperform the AMP algorithm of \cite{fan2022approximate}, which had previously been believed to be optimal. Subsequent work developed an orthogonal AMP (OAMP) algorithm with a rigorous SE description for the same RI noise setting \cite{dudeja2026optimality}. Its SE predictions agree with the replica-based IT characterization of \cite{barbier2023fundamental,barbier2025information}, yielding a coherent picture of the symmetric RI model. For the asymmetric/rectangular setting, recent work has addressed the algorithmic side by developing an OAMP framework for rank-one spiked models with general RI noise and establishing a rigorous SE characterization \cite{Chen2026orthogonal}. However, the corresponding IT characterization remains missing. This gap motivates the present work.\\

\textit{Contributions}. First, we derive a variational expression for the asymptotic mutual information between the signal pair and the observation of the rectangular spiked matrix model \eqref{eq:model_def}. Under the replica symmetric (RS) ansatz \cite{barbier_strong_replica_symmetry}, the variational parameters, known as order parameters in the replica method literature, are shown to satisfy the saddle point equations that predict the MMSE on the signal. Second, we show that in the Bayes-optimal setting we consider, these saddle point equations coincide with the state evolution fixed point equations for the OAMP algorithm \cite{Chen2026orthogonal} with a Bayes-optimal denoiser. This equivalence gives an algorithmic interpretation of the replica prediction and supports the Bayes-optimality of OAMP in regimes where the relevant fixed point is unique. Third, we develop a new approach to evaluate the rectangular inhomogeneous spherical integral. Unlike the symmetric case \cite{barbier2023fundamental,barbier2025information}, which involves a single Haar-distributed matrix, the asymmetric integral requires handling two Haar matrices. Fourth, we present an example with multiple solutions to the saddle point equations. In this example, a statistical-to-computational gap is observed between the MSE (theoretically predicted by the replica method) and the MSE (algorithmically achievable through OAMP).

\textit{Status of the Results}. We distinguish between proved statements and replica-level predictions. Lemmas are proved, whereas Results are obtained either through the replica method under the RS ansatz or through saddle point evaluations of high-dimensional integrals. These derivations involve standard non-rigorous steps: exchanging the replica $n\to0$ and large size $N\to\infty$ limits, analytically continuing all formulas from integer to real $n$, and interchanging integration and saddle point evaluation. These non-rigorous steps are conjectured to be exact.

\textit{Organization}. The remainder of this paper is organized as follows. Section \ref{Sec:Preliminary} introduces the notation and the necessary preliminaries on random matrix theory. Section \ref{Sec:Setup} formulates the rectangular rank-one spiked model and defines the associated free entropy, which is linked to the mutual information. Section \ref{Sec:Free_Entropy} states the main replica (i.e., information-theoretic) predictions for the free entropy and the saddle point equations. Section \ref{Sec:Equivalence} establishes the equivalence between the replica saddle point equations and the SE fixed point equations of OAMP, and discusses the resulting Bayes-optimal prediction. Numerical experiments compare IT and OAMP performance, including examples with multiple fixed points and a statistical-to-computational gap. Section \ref{Sec:RM} presents the replica derivation of the free entropy and the Bayes-optimal saddle point equations. Additional technical derivations are deferred to the Appendices.

\textit{Code Repository}. The source code and implementation details for all numerical evaluations are publicly available at \cite{git}.

\section{Preliminaries on Random Matrix Theory}\label{Sec:Preliminary}

Before introducing the random matrix preliminaries, we use, throughout this paper, plain letters to denote scalars, e.g., $u$ and $v$. Column vectors are denoted by boldface lowercase letters, e.g., $\bu$ and $\bv$, and row vectors by boldface lowercase letters with an arrow, e.g., $\vec{\bu}$ and $\vec{\bv}$. Matrices are denoted by boldface uppercase letters, e.g., $\bA$, and for $\bA$, $A_{i,j}$ denotes its $(i,j)$-th entry. $\bA^{\top}$ denotes its transpose, and $\Tr(\bA)$ denotes its trace. For matrices $\bA$ and $\bB$ with the same number of rows, $[\bA,\bB]$ denotes their horizontal concatenation. For a vector $\bv$, $\Diag(\bv)$ denotes the diagonal matrix with diagonal entries given by $\bv$. The sets of real and complex numbers are denoted by $\R$ and $\C$, respectively, and the imaginary unit is denoted by $\imagUnit$. We use $\E[\cdot]$ for expectation. The Gaussian density with mean $m$ and variance $v$ is denoted by $\Normal(x|m,v)\triangleq(2\pi v)^{-\frac{1}{2}}\exp[-\frac{1}{2v}(x-m)^{2}]$. The notation $\delta(\cdot)$ denotes the Dirac delta measure, $\Sf{Unif}(\mathbb{O}(N))$ denotes the Haar measure on the orthogonal group $\mathbb{O}(N)$, and $C$ denotes a generic constant whose value may change from line to line.
We next provide the necessary random matrix theory background. The readers are referred to \cite{potters2020first} if the notions are new.

\textit{Empirical Spectral Measure}. Consider a rectangular noise matrix $\bZ_{0}\in\R^{N\times M}$ ($N<M$) characterized by the singular value decomposition (SVD) $\bZ_{0}=\bP^{\top}[\Diag(\bs),\bm{0}]\bQ$, where the orthogonal matrices $\bP\sim \Sf{Unif}(\mathbb{O}(N))$ and $\bQ\sim \Sf{Unif}(\mathbb{O}(M))$ are independent and Haar distributed on their respective orthogonal groups (symbol $\sim$ means ``sampled from''). We assume throughout this paper that $\limsup_{N\to\infty}\|\bZ_{0}\|_{\Sf{op}}\le C$ almost surely for some deterministic constant $C$ independent of the dimensions $\{M,N\}$. In the ``thermodynamic limit'', i.e., the large size limit $M,N\to\infty$ with a fixed aspect ratio
\begin{align}
\alpha\triangleq\lim_{M,N\to\infty}\frac{M}{N}\in(1,\infty),\label{eq:alpha}
\end{align}
the empirical spectral distribution of the matrix $\bZ_{0}\bZ_{0}^{\top}$ is defined via the empirical measure
\begin{align}
\mu_{\bZ_{0}\bZ_{0}^{\top}}
& \triangleq \frac{1}{N}\sum_{i=1}^{N}\delta_{\gamma_{i}}
\to \mu(\gamma),
\end{align}
where $\gamma_{i}\triangleq s_{i}^{2}$ denotes the $i$-th eigenvalue of the matrix $\bZ_{0}\bZ_{0}^{\top}$, and $\mu(\cdot)$ is a deterministic probability measure to which the empirical spectral distribution converges weakly. This measure $\mu(\cdot)$ is compactly supported within $\R^{+}$ and admits a square-integrable spectral density, i.e., $\mu(\cdot)\in L^{2}(\R)$. The limiting spectral measure $\tilde{\mu}(\cdot)$ of the companion matrix $\bZ_{0}^{\top}\bZ_{0}$ is coupled to $\mu(\cdot)$ through the identity:
\begin{align}
\tilde{\mu}(\gamma)&=\frac{1}{\alpha}\mu(\gamma)+\left(1-\frac{1}{\alpha}\right)\delta(\gamma).
\end{align}

\textit{Spectral Transforms}. The Stieltjes transform $\cS_{\mu}(\cdot)$ of the limiting spectral density $\mu(\cdot)$ is defined as
\begin{align}
\cS_{\mu}(z)&\triangleq\int{\dd\gamma}\,\frac{\mu(\gamma)}{z-\gamma},\quad z\in\C\setminus\operatorname{supp}(\mu).\label{Eq:S_Tran}
\end{align}
For a function $J(\cdot)$, we use $\cS_{J;\mu}(\cdot)$ to denote the Stieltjes transform of the law of $J(\cdot)$ with $\gamma\sim\mu(\cdot)$:
\begin{align}
\cS_{J;\mu}(z)&\triangleq\E_{\gamma\sim\mu}\frac{1}{z-J(\gamma)}.
\end{align}

The $R$-transform of the limiting density $\mu(\cdot)$, denoted by $\cR_{\mu}(\cdot)$, is defined through the local functional inverse of its Stieltjes transform near the origin:
\begin{align}
\cR_{\mu}(z)&\triangleq\cS^{-1}_{\mu}(z)-\frac{1}{z}.\label{Eq:R_Tran}
\end{align}
Similarly, $\cR_{J;\mu}(\cdot)$ denotes the $R$-transform of the law of $J(\cdot)$ with $\gamma\sim\mu(\cdot)$. With a slight abuse of notation we denote $\cR_{\bA}(\cdot)$ the R-transform associated with the asymptotic spectral law of the random symmetric matrix $\bA$. It linearizes the so-called free additive convolution and satisfies $\cR_{a\bA}(z)=a\cR_{\bA}(az)$ and $\cR_{a\Eye}(z)=a$ \cite{tulino2004random}.

The Hilbert transform of the limiting density $\mu(\cdot)$, denoted by $\cH_{\mu}(\cdot)$, is defined by the Cauchy principal value integral
\begin{align}
\cH_{\mu}(x)&=\frac{1}{\pi}\Sf{P.V.}\int{\dd\gamma}\,\frac{\mu(\gamma)}{x-\gamma},\quad x\in\R.
\end{align}
As the complex argument $z$ approaches the real axis from the upper half-plane, the limiting values of the Stieltjes transform are given by the Sokhotski-Plemelj formula:
\begin{align}
\lim_{\varepsilon\to0^{+}}\cS_{\mu}(x+\imagUnit\varepsilon)&=\pi\cH_{\mu}(x)-\imagUnit\pi\mu(x),\quad x\in\R.
\end{align}
Consequently, the limiting spectral density $\mu(\cdot)$ and its Hilbert transform $\cH_{\mu}(\cdot)$ are recovered from these limiting values via the relations
\begin{align}
\cH_{\mu}(x)=\frac{1}{\pi}\lim_{\varepsilon\to0^{+}}\Re(\cS_{\mu}(x+\imagUnit\varepsilon)),\quad\mu(x)=-\frac{1}{\pi}\lim_{\varepsilon\to0^{+}}\Im(\cS_{\mu}(x+\imagUnit\varepsilon)).\label{Eq:Stieltjes_Tran}
\end{align}

\textit{Contour Integral Representation of a Resolvent}. The resolvent of the matrix $\bZ_{0}\bZ_{0}^{\top}$ is defined as
\begin{align}
\cG_{\bZ_{0}\bZ_{0}^{\top}}(z)\triangleq(z\bI-\bZ_{0}\bZ_{0}^{\top})^{-1},\quad z\in\C\setminus\{\gamma_{i}(\bZ_{0}\bZ_{0}^{\top})\}_{i\leq N}.
\end{align}
Let $\Gamma$ be a closed contour enclosing the spectrum of $\bZ_{0}\bZ_{0}^{\top}$. For any function $g(\cdot)$ that is holomorphic on and inside $\Gamma$, the Cauchy integral formula yields
\begin{align}
\Tr\,g(\bZ_{0}\bZ_{0}^{\top})&=\frac{1}{2\pi\imagUnit}\oint_{\Gamma}{\dd z}\,g(z)\Tr\,\cG_{\bZ_{0}\bZ_{0}^{\top}}(z).\label{Eq:Contour_Integral}
\end{align}

\textit{Woodbury Identity}. For $\bA\in\C^{d\times d}$, $\bB\in\C^{r\times r}$,
$\bU\in\C^{d\times r}$, and $\bV\in\C^{r\times d}$, the Woodbury identity reads
\begin{align}
(\bA+\bU\bB\bV)^{-1}&=\bA^{-1}-\bA^{-1}\bU(\bB^{-1}+\bV\bA^{-1}\bU)^{-1}\bV\bA^{-1},\label{Eq:Woodbury_Identity}
\end{align}
whenever the required inverses exist.

\section{The Rectangular Rank-One Spiked Matrix Model}\label{Sec:Setup}

We consider the rectangular rank-one spiked matrix model defined in \eqref{eq:model_def} under the thermodynamic limit $M,N\to\infty$ with fixed ratio $\alpha$ in \eqref{eq:alpha}. The ground-truth signal vectors $\bu_{0}\in\R^{N}$ and $\bv_{0}\in\R^{M}$ have i.i.d. entries drawn from the prior distributions $P_{\su}(\cdot)$ and $P_{\sv}(\cdot)$, respectively. Without loss of generality, the priors are normalized such that $\E[u_{0}^{2}]=1$ and $\E[v_{0}^{2}]=1$. The additive noise matrix $\bZ_{0}\in\R^{N\times M}$ is independent of $\{\bu_{0},\bv_{0}\}$ and bi-rotationally invariant, i.e., $\bZ_{0}$ and $\bU\bZ_{0}\bV$ are equal in law for arbitrary orthogonal matrices $\bU\in\mathbb{O}(N)$ and $\bV\in\mathbb{O}(M)$. We further specialize to the rectangular trace ensemble \cite{di2003rectangular,barbier2023fundamental,barbier2025information}, in which the distribution of $\bZ_{0}$ is specified by the spectral potential $V(\cdot)$ through
\begin{align}
\bZ_{0}&=\bP^{\top}[\Diag(\bs),\bm{0}]\bQ\sim C_{\sv}\exp\left(-\frac{N}{2}\Tr\,V(\bZ_{0}\bZ_{0}^{\top})\right),\label{eq:noise_density}
\end{align}
where $\bs=(s_{1},\ldots,s_{N})^{\top}$ denotes the singular values of $\bZ_{0}$ and $\bP$, $\bQ$ are the corresponding singular vector matrices. The analytic potential $V(\cdot)$ acts spectrally on $\bZ_{0}\bZ_{0}^{\top}$ as
\begin{align}
V(\bZ_{0}\bZ_{0}^{\top})&=\bP^{\top}\Diag(V(\bs^{2}))\bP=\bP^{\top}\begin{bmatrix}V(s_{1}^{2})&&\\&\ddots&\\&&V(s_{N}^{2})\end{bmatrix}\bP.\label{Eq:Potential_V}
\end{align}

Let $\{\bu_{a},\bv_{a}\}$ denote a candidate signal pair, where the subscript $a$ distinguishes it from the ground-truth pair $\{\bu_{0},\bv_{0}\}$. The corresponding factorized priors are $P_{\su}(\bu_{a})\triangleq\prod_{i=1}^{N}P_{\su}(u_{a,i})$ and $P_{\sv}(\bv_{a})\triangleq\prod_{i=1}^{M}P_{\sv}(v_{a,i})$. Given the observation matrix $\bY$, the posterior distribution of the candidate signal pair is, up to a normalization constant,
\begin{align}
\label{eq:posterior}
P(\bu_{a},\bv_{a}|\bY)&\propto P_{\su}(\bu_{a})P_{\sv}(\bv_{a})\exp\left(-\frac{N}{2}\Tr\,V(\bZ_{a}\bZ_{a}^{\top})\right),
\end{align}
where
\begin{align}
\bZ_{a}&\triangleq\bY-\frac{\lambda}{N}\bu_{a}\bv_{a}^{\top}.\label{Eq:Def_A}
\end{align}
The associated partition function is
\begin{align}
\Omega(\bZ_{0},\bY)&\triangleq\int{\dd P_{\su}(\bu_{a})\dd P_{\sv}(\bv_{a})}\,\exp\left(-\frac{N}{2}\left(\Tr\,V(\bZ_{a}\bZ_{a}^{\top})-\Tr\,V(\bZ_{0}\bZ_{0}^{\top})\right)\right).\label{Eq:PartitionFunction}
\end{align}

Then, the IT limits of inference (mutual information and MMSE) are deduced from the asymptotic free entropy $f$ below, and our main goal is therefore to compute it, recalling that $M$ grows with $N$ such that \eqref{eq:alpha} holds:
\begin{align}
f &\triangleq\lim_{N\to\infty}f_{N},
\quad 
f_{N} \triangleq \frac{1}{N}\E_{\bu_{0},\bv_{0},\bZ_{0}}\log\Omega(\bZ_{0},\bY)
.
\label{Eq:Free_Entropy}
\end{align}
It serves as the central quantity throughout the remainder of the paper. The factor $\exp(\frac{N}{2}\Tr\,V(\bZ_{0}\bZ_{0}^{\top}))$ in \eqref{Eq:PartitionFunction}, which does not depend on the integration variables, is needed for $f_{N}$ to have a finite limit (see Appendix \ref{Sec:C_y} and an analogous argument in \cite[Eq. (8)]{barbier2025information}). Without this subtraction, the exponent scales as $N^{2}$ and therefore the free entropy would not admit a finite limit.

Another central IT quantity is the mutual information between the signal pair and the observation,
\begin{align}
I&\triangleq\lim_{N\to\infty}I_{N},\quad I_{N}\triangleq\frac{1}{N}I\big(\bu_{0},\bv_{0};\bY\big).\label{Eq:MI_def}
\end{align} 
These two quantities are related by an exact identity, valid for every finite $N$:
\begin{align}
I&=-f,\quad I_{N}=-f_{N}.\label{Eq:MI_identity}
\end{align}
Indeed, by \eqref{eq:noise_density} the noise admits the density $C_{\sv}\exp(-\frac{N}{2}\Tr V(\bZ_{0}\bZ_{0}^{\top}))$ w.r.t. the Lebesgue measure on $\R^{N\times M}$, where $C_{\sv}$ depends only on $V(\cdot)$ and the dimensions $M$ and $N$. Hence the conditional density of the data given the signal pair, and the evidence, read
\begin{align}
P(\bY|\bu_{0},\bv_{0})&=C_{\sv}\exp\left(-\frac{N}{2}\Tr\,V(\bZ_{0}\bZ_{0}^{\top})\right),\\
P(\bY)&=C_{\sv}\int{\dd P_{\su}(\bu_{a})\dd P_{\sv}(\bv_{a})}\,\exp\left(-\frac{N}{2}\Tr\,V(\bZ_{a}\bZ_{a}^{\top})\right),
\end{align}
where in the first expression $\bZ_{0}=\bY-\frac{\lambda}{N}\bu_{0}\bv_{0}^{\top}$ is a deterministic function of $\{\bu_{0},\bv_{0},\bY\}$. Therefore
\begin{align}
NI_{N}&=\E\log\frac{P(\bY|\bu_{0},\bv_{0})}{P(\bY)}=-\E\log\Omega(\bZ_{0},\bY)=-Nf_{N},
\end{align} 
where the expectation is taken over the joint law of $(\bu_{0},\bv_{0},\bZ_{0})$. Note that the constant $C_{\sv}$ cancels, and the last equality follows from \eqref{Eq:Free_Entropy}. Identity \eqref{Eq:MI_identity} gives two immediate sanity checks. Since the mutual information is non-negative, $f_{N}\le 0$ for every $N$. At $\lambda=0$ one has $\bZ_{a}=\bY=\bZ_{0}$, so the integrand in \eqref{Eq:PartitionFunction} equals one, and $I_{N}=0$, as it must be since the data are then independent of the signal.

\section{Information-Theoretic Results: Predictions from the Replica Method}\label{Sec:Free_Entropy}

We now state the replica predictions for the asymptotic free entropy (and thus the mutual information) and the associated saddle point equations that yield the MMSE. Their complete derivation is given in Section \ref{Sec:RM}.

\begin{result}[Replica Symmetric Free Entropy]\label{Result:Free_Entropy}

The asymptotic free entropy $f$ of the rectangular rank-one spiked model in \eqref{eq:model_def} admits the replica symmetric variational characterization
\begin{align}
f&=\log\varphi_{\alpha,\lambda}+\psi_{\alpha,\lambda}+\Extr_{\mathcal{M}_{\Sf{RS}}}f_{\Sf{RS}}\label{eq:RSfreeEnt_extr}
\end{align}
Here, $f_{\Sf{RS}}$ is the replica symmetric variational functional
\begin{align}
f_{\Sf{RS}}&\triangleq\frac{1}{4\pi\imagUnit}\oint_{\Gamma_{\sx}}{\dd x}\,\Tr\hat{\bM}^{\top}(x)\bM(x)-\frac{1}{2}\log(1-m_{\su})-\frac{\alpha}{2}\log(1-m_{\sv})-\frac{m_{\su}}{2(1-m_{\su})}-\frac{\alpha\, m_{\sv}}{2(1-m_{\sv})}\nonumber\\
&\quad+\E_{u_{0},z_{\su}}\log\E_{u_{a}}\exp\left(\hat{m}_{\su}u_{0}u_{a}-\frac{1}{2}\hat{m}_{\su}u_{a}^{2}+\sqrt{\hat{m}_{\su}}u_{a}z_{\su}\right)+\alpha\E_{v_{0},z_{\sv}}\log\E_{v_{a}}\exp\left(\hat{m}_{\sv}v_{0}v_{a}-\frac{1}{2}\hat{m}_{\sv}v_{a}^{2}+\sqrt{\hat{m}_{\sv}}v_{a}z_{\sv}\right)\nonumber\\
&\quad-\frac{\alpha-1}{2}\left(\log\tilde{\cG}_{\sv}(0)+\frac{\tilde{m}_{\sv}-\tilde{\cM}_{\sv}^{2}(0)}{\tilde{\cG}_{\sv}(0)}\right)-\frac{1}{2}\E_{s^{2}\sim\mu}\left[\log\tilde{\cG}_{\su}(s)+\frac{\tilde{m}_{\su}-\tilde{\cM}_{\su}^{2}(s)}{\tilde{\cG}_{\su}(s)}+\frac{\bar{\cN}(s)}{\bar{\cD}(s)}-\frac{\tilde{M}_{\Sf{uv}}^{2}(s)}{\tilde{\cG}_{\su}(s)}+\log\bar{\cD}(s)\right].\label{eq:RSfreeEnt}
\end{align}
The notation $\Extr_{\mathcal{M}_{\Sf{RS}}}(\cdot)$ denotes extremization w.r.t. the scalar parameters $\mathcal{M}_{\Sf{RS}}\triangleq\{m_{\su},m_{\sv},\hat{m}_{\su},\hat{m}_{\sv}\}$. For the $\su$-channel, $z_{\su}\sim\cN(0,1)$ and $u_{0},u_{a}\overset{\text{i.i.d.}}{\sim}P_{\su}(\cdot)$; the same convention is used for the $\sv$-channel. The measure $\mu(\cdot)$ denotes the limiting spectral measure of $\bZ_{0}\bZ_{0}^{\top}$. Let $\Gamma_{\sx}$ be a fixed closed contour in $\C$, independent of $M$ and $N$, chosen as described in Section \ref{Sec:Block_Symmetrization}. For ease of reference, the final expressions for the parameters and functions needed to evaluate \eqref{eq:RSfreeEnt_extr}-\eqref{eq:RSfreeEnt} are collected in Appendix \ref{App:RS_Coupling_Functions}.

\end{result}

The asymptotic mutual information follows from \eqref{Eq:MI_identity}, with $f$ given by the variational formula in Result \ref{Result:Free_Entropy}. In contrast to the saddle point equations in Result \ref{Result:Saddle_Point}, which are unaffected by the additive constants $\log\varphi_{\alpha,\lambda}$ and $\psi_{\alpha,\lambda}$, the mutual information requires these constants to be retained. Indeed, they are not free parameters but explicit quantities computed in Appendix \ref{Sec:C_y} and \eqref{eq:free_entropy_additive_constant}. They depend only on $\alpha$, $\lambda$, and the limiting spectral measure $\mu(\cdot)$. Note that, for standard Gaussian noise, the mutual information (i.e., negative free entropy) determines the MMSE through the I-MMSE relation \cite{guo2005mutual}. For general RI noise as in the present paper, no such relation is available. Nevertheless, the stationary values of the order parameters yield the MMSE directly.

The extremization in \eqref{eq:RSfreeEnt_extr} is associated with a system of saddle point equations. Under the Bayes-optimal setting, the Nishimori identities in Lemma \ref{rmk:Nishimori} reduce the subset of these equations needed to determine the MMSE to the following scalar system.

Let us observe that at $\lambda=0$ one has $\cJ_{\su}=\cJ_{\sv}=0$, so that the saddle point equations of Result~\ref{Result:Saddle_Point} give $\hat{m}_{\su}=\hat{m}_{\sv}=0$ and hence $m_{\su}=\E[u_{0}]^{2}$, $m_{\sv}=\E[v_{0}]^{2}$. Substituting into \eqref{eq:RSfreeEnt}, the contour term vanishes since $\hat{\bM}\propto\lambda^{2}$, the scalar-channel terms vanish since $\hat{m}_{\su}=\hat{m}_{\sv}=0$, and the remaining terms cancel pairwise, giving $f_{\Sf{RS}}=0$. Moreover, $\log\varphi_{\alpha,0}=\psi_{\alpha,0}=0$, because at $\lambda=0$ both integrands in \eqref{eq:free_entropy_additive_constant} and in $\log\varphi_{\alpha,\lambda}$ reduce to $xV'(x^{2})$ times a constant, which is analytic inside $\Gamma_{\sx}$. Hence $f=0$, in agreement with the identity $I=-f$ and with the fact that the data carry no information on the signal at $\lambda=0$.

\begin{result}[Saddle Point Equations and Minimum Mean Square Error]\label{Result:Saddle_Point}

Under the Bayes-optimal setting, the scalar saddle point equations for $\mathcal{M}_{\Sf{RS}}$ take the form
\begin{align}
m_{\su}&=\E_{u_{0},z_{\su}}[u_{0}\langle u_{a}\rangle],\label{Eq:RM_m_u}\\
m_{\sv}&=\E_{v_{0},z_{\sv}}[v_{0}\langle v_{a}\rangle],\label{Eq:RM_m_v}\\
\hat{m}_{\su}&=-\cR_{\cJ_{\su};\mu}(1-m_{\su}),\label{Eq:RM_hm_u}\\
\hat{m}_{\sv}&=-\cR_{\cJ_{\sv};\tilde{\mu}}(1-m_{\sv}).\label{Eq:RM_hm_v}
\end{align}
Here, $\cR_{\cJ;\nu}(\cdot)$ denotes the $R$-transform of the law of $\cJ(\cdot)$ for $\gamma\sim\nu(\cdot)$, with the functions $\cJ_{\su}(\cdot)$ and $\cJ_{\sv}(\cdot)$ given in \eqref{Eq:Ju_x}-\eqref{Eq:Jv_x} below. The random variables follow the same conventions as in Result \ref{Result:Free_Entropy}. In particular, the scalar posterior expectation for the $\su$-channel is
\begin{align}
\langle u_{a}\rangle&\triangleq\frac{\E_{u_{a}}\left[u_{a}\exp\left(\hat{m}_{\su}u_{0}u_{a}-\frac{1}{2}\hat{m}_{\su}u_{a}^{2}+\sqrt{\hat{m}_{\su}}u_{a}z_{\su}\right)\right]}{\E_{u_{a}}\exp\left(\hat{m}_{\su}u_{0}u_{a}-\frac{1}{2}\hat{m}_{\su}u_{a}^{2}+\sqrt{\hat{m}_{\su}}u_{a}z_{\su}\right)}.\label{eq:scalPosterior}
\end{align}
The $\sv$-channel is defined analogously. The functions entering \eqref{Eq:RM_hm_u}-\eqref{Eq:RM_hm_v} are given by
\begin{align}
\cJ_{\su}(x)&=-\cJ_{\Sf{uu}}(x)+\frac{\cJ_{\Sf{uv}}(x)}{(1-m_{\sv})^{-1}-\hat{m}_{\sv}+\cJ_{\Sf{vv}}(x)},\label{Eq:Ju_x}\\
\cJ_{\sv}(x)&=-\cJ_{\Sf{vv}}(x)+\frac{\cJ_{\Sf{uv}}(x)}{(1-m_{\su})^{-1}-\hat{m}_{\su}+\cJ_{\Sf{uu}}(x)},\label{Eq:Jv_x}\\
\cJ_{\Sf{uu}}(x)&=\lambda^{2}\left[x\E_{\gamma\sim\mu}\frac{V^{\prime}(x)-V^{\prime}(\gamma)}{x-\gamma}+\E_{\gamma\sim\mu}V^{\prime}(\gamma)+(\alpha-1)V^{\prime}(x)\right],\label{Eq:Juu_x}\\
\cJ_{\Sf{vv}}(x)&=\lambda^{2}\left(x\E_{\gamma\sim\mu}\frac{V^{\prime}(x)-V^{\prime}(\gamma)}{x-\gamma}+\E_{\gamma\sim\mu}V^{\prime}(\gamma)\right),\label{Eq:Jvv_x}\\
\cJ_{\Sf{uv}}(x)&=\lambda^{2}x(V^{\prime}(x))^{2}.\label{Eq:Juv_x}
\end{align}

Let $\operatorname{Sol}_{\Sf{RS}}\triangleq\{(m_{\su,\ell},m_{\sv,\ell},\hat{m}_{\su,\ell},\hat{m}_{\sv,\ell})\}_{\ell=1}^{L}$ denote the scalar stationary solutions of \eqref{Eq:RM_m_u}-\eqref{Eq:RM_hm_v}, where $L$ denotes the number of solutions. Among these solutions, the one with the largest value of $f_{\Sf{RS}}$ is selected as the IT solution. Let $\ell_{\Sf{IT}}$ denote its index and set $m_{\su}^{*}\triangleq m_{\su,\ell_{\Sf{IT}}}$ and $m_{\sv}^{*}\triangleq m_{\sv,\ell_{\Sf{IT}}}$. The selected values $m_{\su}^{*}$ and $m_{\sv}^{*}$ are the asymptotic overlaps reported in the numerical experiments:
\begin{align}
m_{\su}^{*}&=\lim_{N\to\infty}\frac{1}{N}\E|\bu_{0}^{\top}\bu_{a}|,\quad m_{\sv}^{*}=\lim_{N\to\infty}\frac{1}{M}\E|\bv_{0}^{\top}\bv_{a}|,\label{Eq:overlaps}
\end{align}
where $\{\bu_{a},\bv_{a}\}$ is a sample from the posterior \eqref{eq:posterior} given $\bY$, and the expectation is over the joint law of $\{\bu_{0},\bv_{0},\bZ_{0}\}$ and of this sample. The absolute values account for the invariance of \eqref{eq:model_def} under the joint sign flip $(\bu_{0},\bv_{0})\mapsto(-\bu_{0},-\bv_{0})$: for priors symmetric under this flip, $\E[\bu_{a}|\bY]=\E[\bv_{a}|\bY]=\boldsymbol{0}$, so that the overlap with the posterior mean vanishes even when the signal is partially recovered. Equivalently, $m_{\su}^{*}$ and $m_{\sv}^{*}$ are the overlaps obtained in the presence of an infinitesimal amount of side information breaking the symmetry. To avoid this ambiguity, as in \cite{barbier2023fundamental}, we quantify the Bayes-optimal performance through the error on the rank-one spike, which is invariant under the flip. Under the normalization $\E[u_{0}^{2}]=\E[v_{0}^{2}]=1$, the Nishimori identities give
\begin{align}
\lim_{N\to\infty}\frac{1}{NM}\E\left\|\bu_{0}\bv_{0}^{\top}-\E[\bu_{a}\bv_{a}^{\top}|\bY]\right\|_{F}^{2}&=1-m_{\su}^{*}m_{\sv}^{*}.\label{Eq:MMSEs}
\end{align}
The errors on the individual vectors, defined up to the global sign, are correspondingly $1-m_{\su}^{*}$ and $1-m_{\sv}^{*}$.

\end{result}

The equations \eqref{Eq:RM_m_u}-\eqref{Eq:RM_m_v} admit a mean-field interpretation in terms of coupled scalar Gaussian channels. For $\hat{m}_{\sv}>0$, the $\sv$-channel is represented by
\begin{align}
\mu_{\sv}&=v_{0}+\hat{m}_{\sv}^{-\frac{1}{2}}z_{\sv},\label{Eq:Mu_v}
\end{align}
where $\mu_{\sv}$ denotes the effective scalar observation used to estimate $v_{0}$, and $\hat{m}_{\sv}$ plays the role of the effective SNR. An analogous characterization holds for the $\su$-channel.

The grid-based residual minimization procedure for obtaining the stationary solutions in $\operatorname{Sol}_{\Sf{RS}}$ is described in Appendix \ref{Sec:Numerical_Solution}. For each scalar stationary solution in $\operatorname{Sol}_{\Sf{RS}}$, the scalar parameters alone are not sufficient to evaluate $f_{\Sf{RS}}$ in \eqref{eq:RSfreeEnt}. The matrix-valued functions $\{\tilde{\bM},\bM,\hat{\bM}\}$ collected in Appendix \ref{App:RS_Coupling_Functions} are also required. These functions satisfy the closed self-consistency map $\tilde{\bM}\to\bM\to\hat{\bM}\to\tilde{\bM}$ and must therefore be determined jointly from a fixed point of this map before the free entropy can be evaluated. The numerical solution of this map and the subsequent evaluation of the free entropy are detailed in Appendix \ref{Sec:Branch_Selection}.

When the above fixed point equations admit multiple solutions, the IT prediction is obtained by evaluating $f_{\Sf{RS}}$ on every solution and selecting the maximizer. We denote by $\lambda_{\Sf{IT}}$ the smallest SNR at which the global maximizer of $f_{\Sf{RS}}$ is a solution with non-zero overlaps $\{m_{\su}^{*}, m_{\sv}^{*}\}$. Below $\lambda_{\Sf{IT}}$ a solution with non-zero overlaps may already exist while having a smaller $f_{\Sf{RS}}$; it is then not the IT prediction, and we call it metastable.

Above $\lambda_{\Sf{IT}}$ the saddle point equations may still admit several solutions. In the examples of Section V-D, three stationary solutions are detected: $\{m_{\su}^{*},m_{\sv}^{*}\}$ and two additional solutions $\{m_{\su,\ell}^{\Sf{meta}},m_{\sv,\ell}^{\Sf{meta}}\}_{\ell=1,2}$. All three are stationary solutions of $f_{\Sf{RS}}$, but the IT solution attains the largest free entropy among them. Selecting between them is therefore not a matter of solving the saddle point equations, which all of them satisfy, but of comparing the values of $f_{\Sf{RS}}$ they attain.

The fixed-point iteration described in Appendix \ref{Sec:Numerical_Solution} performs no such comparison. It starts from $\hat{m}_{\su}=\hat{m}_{\sv}=\hat{m}_{\rm init}$, where $\hat{m}_{\rm init}$ is a very small positive number, and converges to whichever solution the recursion reaches, without evaluating $f_{\Sf{RS}}$. We therefore denote by $\lambda_{\Sf{alg}}$ the smallest SNR at which this iteration converges to $\{m_{\su}^{*},m_{\sv}^{*}\}$. For $\lambda>\lambda_{\Sf{alg}}$ the two selections agree. For $\lambda\in(\lambda_{\Sf{IT}},\lambda_{\Sf{alg}})$ they disagree: the IT prediction is $\{m_{\su}^{*},m_{\sv}^{*}\}$, while the iteration converges to $\{m_\su^{\Sf{meta}}, m_\sv^{\Sf{meta}}\}$. The numerical examples in Section \ref{Sec:Bayes_Optimality_Prediction} show that on this interval the state evolution of OAMP selects the same branch as the iteration, so that the error predicted for OAMP remains strictly above the replica-predicted MMSE; the interval is then the statistical-to-computational gap.

\textit{Consistency Check in the Gaussian Noise Setting}. As a consistency check, we verify that our generalized saddle point equations recover the classical spiked-Wishart equations under i.i.d. Gaussian noise \cite{miolane2017fundamental}. We first recall the known fixed point equations in this case and then show that our generalized equations reduce to the same form. The equations in \cite{miolane2017fundamental} take the following simple linear form:
\begin{align}
\hat{m}_{\su}=\alpha\lambda^{2}m_{\sv},\quad\hat{m}_{\sv}=\lambda^{2}m_{\su},\label{eq:S_W_SP}
\end{align}
while \eqref{Eq:RM_m_u}-\eqref{Eq:RM_m_v} remain unchanged. In our setting, Gaussian noise corresponds to the potential $V(x)=x$, and hence $V'(x)=1$. Therefore, the functions $\{\cJ_{\su},\cJ_{\sv}\}$ defined in \eqref{Eq:Ju_x}-\eqref{Eq:Jv_x} reduce to
\begin{align}
\cJ_{\su}(x)&=-\alpha\lambda^{2}+\frac{\lambda^{2}}{(1-m_{\sv})^{-1}-\hat{m}_{\sv}+\lambda^{2}}x,\quad\cJ_{\sv}(x)=-\lambda^{2}+\frac{\lambda^{2}}{(1-m_{\su})^{-1}-\hat{m}_{\su}+\alpha\lambda^{2}}x.
\end{align}

The evaluation of the conjugate parameters uses the additivity of the $R$-transform for freely independent matrices, together with its scaling and constant shift properties, namely $\cR_{\bA+\bB}(z)=\cR_{\bA}(z)+\cR_{\bB}(z)$, $\cR_{a\bA}(z)=a\cR_{\bA}(az)$, and $\cR_{a\Eye}(z)=a$ \cite{tulino2004random}. Together with the standard $R$-transforms for i.i.d. Gaussian noise, $\cR_{\mu}(z)=\frac{\alpha}{1-z}$ and $\cR_{\tilde{\mu}}(z)=\frac{1}{1-\alpha z}$ \cite{tulino2004random,potters2020first}, the conjugate parameter $\hat{m}_{\su}$ in \eqref{Eq:RM_hm_u} becomes
\begin{align}
\hat{m}_{\su}&=-\cR_{\cJ_{\su};\mu}(1-m_{\su})\nonumber\\
&=-\left(-\alpha\lambda^{2}+\frac{\lambda^{2}}{(1-m_{\sv})^{-1}-\hat{m}_{\sv}+\lambda^{2}}\cR_{\mu}\left(\frac{\lambda^{2}}{(1-m_{\sv})^{-1}-\hat{m}_{\sv}+\lambda^{2}}(1-m_{\su})\right)\right)\nonumber\\
&=\alpha\lambda^{2}-\frac{\alpha\lambda^{2}}{(1-m_{\sv})^{-1}-\hat{m}_{\sv}+\lambda^{2}m_{\su}}.\label{Eq:hm_u}
\end{align}
Similarly, the conjugate parameter $\hat{m}_{\sv}$ in \eqref{Eq:RM_hm_v} becomes
\begin{align}
\hat{m}_{\sv}&=-\cR_{\cJ_{\sv};\tilde{\mu}}(1-m_{\sv})\nonumber\\
&=-\left(-\lambda^{2}+\frac{\lambda^{2}}{(1-m_{\su})^{-1}-\hat{m}_{\su}+\alpha\lambda^{2}}\cR_{\tilde{\mu}}\left(\frac{\lambda^{2}}{(1-m_{\su})^{-1}-\hat{m}_{\su}+\alpha\lambda^{2}}(1-m_{\sv})\right)\right)\nonumber\\
&=\lambda^{2}-\frac{\lambda^{2}}{(1-m_{\su})^{-1}-\hat{m}_{\su}+\alpha\lambda^{2}m_{\sv}}.\label{Eq:hm_v}
\end{align}
We can verify that any solution to \eqref{eq:S_W_SP} also satisfies \eqref{Eq:hm_u}-\eqref{Eq:hm_v}, confirming that the classical spiked-Wishart fixed points are recovered by our generalized equations. This consistency can also be verified through the equivalence established in Section \ref{Sec:Equivalence} between the proposed saddle point equations and the SE fixed point equations of OAMP. Proposition 3 of \cite{Chen2026orthogonal} proves that these SE fixed point equations of OAMP reduce to the classical result of \cite{miolane2017fundamental} under i.i.d. Gaussian noise. Therefore, the same reduction of our saddle point equations follows from this equivalence.

\section{Algorithmic Results:  Bayes-Optimality of OAMP}\label{Sec:Equivalence}

This section establishes an exact equivalence between the replica saddle point equations in Result \ref{Result:Saddle_Point} and the fixed point equations associated with the rigorous SE of OAMP \cite{Chen2026orthogonal}. We first review the OAMP recursion and the SE fixed point equations, and then derive a nontrivial spectral reformulation of the replica saddle point equations that allows a direct algebraic comparison between the two sets of equations. The resulting equivalence provides an algorithmic interpretation of the replica saddle point equations. In particular, when the SE admits a unique fixed point, the asymptotic performance of OAMP, as characterized by the SE, coincides with the IT limit predicted by the replica analysis. Consequently, \emph{OAMP is conjectured to be optimal in regimes where no statistical-to-computational gap exists}.

\subsection{Review of OAMP and Its State Evolution}

To make the comparison self-contained, we briefly review the OAMP framework and its SE fixed point equations. The OAMP algorithm introduced in \cite{Chen2026orthogonal} addresses the rank-one rectangular matrix estimation problem with RI noise specified in \eqref{eq:model_def}. The algorithm generates sequences of iterates $\{\bu_{t}\}$ and $\{\bv_{t}\}$ according to
\begin{align}
\bu_{t}&=F_{t}(\bY\bY^{\top})f_{t}(\bu_{\le t-1})+\tilde{F}_{t}(\bY\bY^{\top})\bY g_{t}(\bv_{\le t-1}),\\
\bv_{t}&=G_{t}(\bY^{\top}\bY)g_{t}(\bv_{\le t-1})+\tilde{G}_{t}(\bY^{\top}\bY)\bY^{\top}f_{t}(\bu_{\le t-1}),
\end{align}
where $f_{t}(\cdot)$ and $g_{t}(\cdot)$ are entrywise denoisers. The matrix denoisers $F_{t}(\cdot)$, $\tilde{F}_{t}(\cdot)$, $G_{t}(\cdot)$, and $\tilde{G}_{t}(\cdot)$ are spectral functions that act on the eigenvalues of their matrix arguments while preserving the corresponding eigenvectors. In particular, $F_t(\cdot)$ and $G_t(\cdot)$ are trace-free, whereas $f_t(\cdot)$ and $g_t(\cdot)$ are divergence-free. We refer the reader to \cite{Chen2026orthogonal} for the precise definitions and the optimal choices of these functions.

In the thermodynamic limit, the performance of OAMP is characterized by its SE. Since our goal is to compare the SE fixed point equations with the saddle point equations in Result \ref{Result:Saddle_Point}, we work directly at the fixed point level and suppress the iteration indices. Following \cite{Chen2026orthogonal}, an SE fixed point is characterized by the four parameters $\{\rho_{\su},\rho_{\sv},w_{\su},w_{\sv}\}$, which satisfy
\begin{align}
\rho_{\su}&=\frac{1}{\MMSE_{\su}(w_{\su})}-\frac{1}{1-w_{\su}},\label{Eq:OAMP_ru}\\
\rho_{\sv}&=\frac{1}{\MMSE_{\sv}(w_{\sv})}-\frac{1}{1-w_{\sv}},\label{Eq:OAMP_rv}\\
w_{\su}&=1-\frac{1-\E_{\gamma\sim\mu} P(\gamma;\rho_{\su},\rho_{\sv})}{\E_{\gamma\sim\mu} P(\gamma;\rho_{\su},\rho_{\sv})}\frac{1}{\rho_{\su}},\label{Eq:OAMP_wu}\\
w_{\sv}&=1-\frac{1-\E_{\gamma\sim\tilde{\mu}} Q(\gamma;\rho_{\su},\rho_{\sv})}{\E_{\gamma\sim\tilde{\mu}} Q(\gamma;\rho_{\su},\rho_{\sv})}\frac{1}{\rho_{\sv}}.\label{Eq:OAMP_wv}
\end{align}
Here, the functions $P(\cdot)$ and $Q(\cdot)$ are defined by
\begin{align}
P(x;\rho_{\su},\rho_{\sv})&\triangleq\frac{x(\rho_{\sv}\varphi_{2}(x)+\delta)}{D(x)},\label{Eq:P_Func}\\
Q(x;\rho_{\su},\rho_{\sv})&\triangleq\begin{cases}\frac{\delta x(\rho_{\su}\varphi_{1}(x)+1)}{D(x)}&x>0\\\frac{\delta}{\rho_{\sv}\varphi_{2}(0)+\delta}&x=0\end{cases},\label{Eq:Q_Func}\\
D(x)&\triangleq x(\rho_{\su}\varphi_{1}(x)+1)(\rho_{\sv}\varphi_{2}(x)+\delta)-\rho_{\su}\rho_{\sv}\varphi_{3}^{2}(x),\\
\varphi_{1}(x)&\triangleq\frac{1+\delta\theta^{2}\pi^{2}x(\cH_{\mu}^{2}(x)+\mu^{2}(x))}{C(x)},\label{Eq:Rho_1}\\
\varphi_{2}(x)&\triangleq\begin{cases}\delta\varphi_{1}(x)+\theta(1-\delta)\frac{1}{x}\varphi_{3}(x)&x>0\\\frac{\delta}{1+(1-\delta)\theta^{2}\E_{\gamma\sim\mu}\frac{1}{\gamma}}&x=0\end{cases},\label{Eq:Rho_2}\\
\varphi_{3}(x)&\triangleq\frac{\theta(1-\delta+2\delta\pi x\cH_{\mu}(x))}{C(x)},\label{Eq:Rho_3}\\
C(x)&\triangleq\Big\{\theta^{2}\pi\mu(x)\Big[(1-\delta)+2\delta\pi x\cH_{\mu}(x)\Big]\Big\}^{2}+\Big\{1-\theta^{2}\Big[(1-\delta)\pi\cH_{\mu}(x)-\delta\pi^{2}x\left(\mu^{2}(x)-\cH_{\mu}^{2}(x)\right)\Big]\Big\}^{2}.\label{Eq:C_Func}
\end{align}

Throughout the comparison below, the parameters $\delta$ and $\theta$ in \eqref{Eq:P_Func}-\eqref{Eq:C_Func} are given by
\begin{align}
\delta\triangleq\frac{1}{\alpha},\quad\theta\triangleq\sqrt{\alpha}\lambda.\label{Eq:Mapping}
\end{align}

\subsection{Equivalence with the SE Fixed Point Equations of OAMP}

We first rewrite the replica saddle point equations \eqref{Eq:RM_m_u}-\eqref{Eq:RM_hm_v} in terms of $\mu(\cdot)$ and $\cH_{\mu}(\cdot)$, and then directly compare the resulting equations with the SE fixed point equations in \eqref{Eq:OAMP_ru}-\eqref{Eq:OAMP_wv}.\\

\subsubsection{Spectral Reformulation of the Replica Saddle Point Equations}

The SE fixed point equations of OAMP are expressed in terms of the limiting spectral density $\mu(\cdot)$ and its Hilbert transform $\cH_{\mu}(\cdot)$ \cite{Chen2026orthogonal}. In contrast, the saddle point equations in Result \ref{Result:Saddle_Point} are written through the matrix potential $V(\cdot)$. To compare the two formulations, we use Lemmas \ref{Lemma:Hilbert_Potential_Identities} and \ref{Lemma:Spectral_Mapping} below to rewrite the functions in Result \ref{Result:Saddle_Point} in terms of $\mu(\cdot)$ and $\cH_{\mu}(\cdot)$. We consider the last two saddle point equations \eqref{Eq:RM_hm_u}-\eqref{Eq:RM_hm_v}. Since the calculations for the two channels are analogous, we give the details for \eqref{Eq:RM_hm_v}. By the definition of the R-transform in \eqref{Eq:R_Tran}, \eqref{Eq:RM_hm_v} is equivalent to
\begin{align}
\hat{m}_{\sv}&=-\cS^{-1}_{\cJ_{\sv};\tilde{\mu}}(1-m_{\sv})+(1-m_{\sv})^{-1}.
\end{align}
Here, $\cS^{-1}_{\cJ_{\sv};\tilde{\mu}}(\cdot)$ denotes the inverse of $\cS_{\cJ_{\sv};\tilde{\mu}}(\cdot)$. Hence,
\begin{align}
1-m_{\sv}&=\cS_{\cJ_{\sv};\tilde{\mu}}\left[(1-m_{\sv})^{-1}-\hat{m}_{\sv}\right]\nonumber\\
&=\E_{\gamma\sim\tilde{\mu}}\left[(1-m_{\sv})^{-1}-\hat{m}_{\sv}-\cJ_{\sv}(\gamma)\right]^{-1}\nonumber\\
&\overset{\text{(A)}}{=}\frac{1}{\alpha}\E_{\gamma\sim\mu}\left[(1-m_{\sv})^{-1}-\hat{m}_{\sv}-\cJ_{\sv}(\gamma)\right]^{-1}+\left(1-\frac{1}{\alpha}\right)\left[(1-m_{\sv})^{-1}-\hat{m}_{\sv}-\cJ_{\sv}(0)\right]^{-1}\nonumber\\
&\overset{\text{(B)}}{=}\frac{1}{\alpha}\E_{\gamma\sim\mu}\left[(1-m_{\sv})^{-1}-\hat{m}_{\sv}+\cJ_{\Sf{vv}}(\gamma)-\frac{\cJ_{\Sf{uv}}(\gamma)}{(1-m_{\su})^{-1}-\hat{m}_{\su}+\cJ_{\Sf{uu}}(\gamma)}\right]^{-1}\nonumber\\
&\qquad\left(1-\frac{1}{\alpha}\right)\left[(1-m_{\sv})^{-1}-\hat{m}_{\sv}+\cJ_{\Sf{vv}}(0)\right]^{-1}\nonumber\\
&\overset{\text{(C)}}{=}\frac{1}{\alpha}\E_{\gamma\sim\mu}\frac{\rho_{\su}+1+\cJ_{\Sf{uu}}(\gamma)}{\rho_{\su}\rho_{\sv}+\rho_{\su}(1+\cJ_{\Sf{vv}}(\gamma))+\rho_{\sv}(1+\cJ_{\Sf{uu}}(\gamma))+(1+\cJ_{\Sf{uu}}(\gamma))(1+\cJ_{\Sf{vv}}(\gamma))-\cJ_{\Sf{uv}}(\gamma)}\\
&\qquad+\left(1-\frac{1}{\alpha}\right)(\rho_{\sv}+1+\cJ_{\Sf{vv}}(0))^{-1}.\label{Eq:1_m_v}
\end{align}
Here, step (A) uses $\tilde{\mu}=\frac{1}{\alpha}\mu+(1-\frac{1}{\alpha})\delta_{0}$, while step (B) follows from \eqref{Eq:Jv_x}. Step (C) follows from the definitions
\begin{align}
\rho_{\su}&\triangleq(1-m_{\su})^{-1}-\hat{m}_{\su}-1,\label{Eq:rho_u}\\
\rho_{\sv}&\triangleq(1-m_{\sv})^{-1}-\hat{m}_{\sv}-1.\label{Eq:rho_v}
\end{align}
Similarly, \eqref{Eq:RM_hm_u} gives
\begin{align}
1-m_{\su}&=\E_{\gamma\sim\mu}\frac{\rho_{\sv}+1+\cJ_{\Sf{vv}}(\gamma)}{\rho_{\su}\rho_{\sv}+\rho_{\su}(1+\cJ_{\Sf{vv}}(\gamma))+\rho_{\sv}(1+\cJ_{\Sf{uu}}(\gamma))+(1+\cJ_{\Sf{uu}}(\gamma))(1+\cJ_{\Sf{vv}}(\gamma))-\cJ_{\Sf{uv}}(\gamma)}.\label{Eq:1_m_u}
\end{align}

The functions $\{\cJ_{\Sf{uu}},\cJ_{\Sf{vv}},\cJ_{\Sf{uv}}\}$ are still defined in terms of $V(\cdot)$. The following two lemmas allow us to express them in terms of $\mu(\cdot)$ and $\cH_{\mu}(\cdot)$. Detailed proofs of both lemmas are provided in Appendix \ref{Appendix:Proof_Spectral_Mapping}.

\begin{lemma}\label{Lemma:Hilbert_Potential_Identities}

Let $\mu(\cdot)$ be a square-integrable limiting spectral density, i.e., $\mu(\cdot)\in L^{2}(\R)$. The mean of its associated Hilbert transform $\cH_{\mu}(\cdot)$ w.r.t. $\mu(\cdot)$ vanishes identically:
\begin{align}
\E_{\gamma\sim\mu}\cH_{\mu}(\gamma)&=\int{\dd\gamma}\,\mu(\gamma)\cH_{\mu}(\gamma)=0.
\end{align}

\end{lemma}

\begin{lemma}\label{Lemma:Spectral_Mapping}

Consider the RI noise matrix $\bZ_{0}\in\R^{N\times M}$ governed by the potential $V(\bZ_{0}\bZ_{0}^{\top})$ in \eqref{eq:noise_density}. Assume $V(\cdot)$ is analytic in the neighbourhood of $\text{supp}({\mu})$, that $\text{supp}({\mu})$ is a single compact interval bounded away from the origin, and that $\mu(\cdot)$ admits a density. Then for every $x$ in the interior of $\text{supp}({\mu})$, the following identities hold:
\begin{align}
V'(x)&=2\pi\cH_{\mu}(x)+(\alpha-1)\frac{1}{x},\\
\E_{\gamma\sim\mu}\frac{V'(x)-V'(\gamma)}{x-\gamma}&=\pi^{2}(\mu^{2}(x)+\cH^{2}_{\mu}(x))-(\alpha-1)\frac{1}{x}\E_{\gamma\sim\mu}\frac{1}{\gamma}.
\end{align}

\end{lemma}

For $\cJ_{\Sf{uu}}(\cdot)$ in \eqref{Eq:Juu_x}, we have
\begin{align}
\cJ_{\Sf{uu}}(x)&=\lambda^{2}\left[x\E_{\gamma\sim\mu}\frac{V'(x)-V'(\gamma)}{x-\gamma}+\E_{\gamma\sim\mu}V'(\gamma)+(\alpha-1)V'(x)\right]\nonumber\\
&\overset{\text{(D)}}{=}\lambda^{2}\left[\pi^{2}x\left(\mu^{2}(x)+\cH_{\mu}^{2}(x)\right)+2\pi\E_{\gamma\sim\mu}\cH_{\mu}(\gamma)+2(\alpha-1)\pi\cH_{\mu}(x)+(\alpha-1)^{2}\frac{1}{x}\right]\nonumber\\
&\overset{\text{(E)}}{=}\lambda^{2}\pi^{2}x\left(\mu^{2}(x)+\cH_{\mu}^{2}(x)\right)+2(\alpha-1)\lambda^{2}\pi\cH_{\mu}(x)+(\alpha-1)^{2}\lambda^{2}\frac{1}{x}.\label{Eq:Juu_X}
\end{align}
Here, step (D) uses Lemma \ref{Lemma:Spectral_Mapping}, while step (E) uses Lemma \ref{Lemma:Hilbert_Potential_Identities}. Similarly, \eqref{Eq:Jvv_x}-\eqref{Eq:Juv_x} give
\begin{align}
\cJ_{\Sf{vv}}(0)&=(\alpha-1)\lambda^{2}\E_{\gamma\sim\mu}\frac{1}{\gamma},\\
\cJ_{\Sf{vv}}(x)&=\lambda^{2}\pi^{2}x(\mu^{2}(x)+\cH^{2}_{\mu}(x)),\label{Eq:Jvv_X}\\
\cJ_{\Sf{uv}}(x)&=4\lambda^{2}\pi^{2}x\cH_{\mu}^{2}(x)+4(\alpha-1)\lambda^{2}\pi\cH_{\mu}(x)+(\alpha-1)^{2}\lambda^{2}\frac{1}{x}.\label{Eq:Juv_X}
\end{align}

After substituting \eqref{Eq:Juu_X}-\eqref{Eq:Juv_X}, equations \eqref{Eq:1_m_v} and \eqref{Eq:1_m_u} are expressed entirely in terms of $\mu(\cdot)$ and $\cH_{\mu}(\cdot)$. The replica saddle point equations are therefore in the spectral form required for a direct comparison with the SE fixed point equations.\\

\subsubsection{Matching with the SE Fixed Point Equations}

With both formulations written in compatible spectral forms, we now establish their algebraic equivalence. The comparison proceeds in two steps. First, we match the effective scalar Gaussian channels in the SE and the replica formulation, which yields the relations in \eqref{Eq:rho_u}-\eqref{Eq:rho_v}. Second, we derive the MMSE expressions from the SE fixed point equations \eqref{Eq:OAMP_ru}-\eqref{Eq:OAMP_wv} and show that they coincide with the saddle point equations in \eqref{Eq:1_m_v} and \eqref{Eq:1_m_u}. This completes the algebraic comparison.

In the SE, the scalar Gaussian channel for $\sv$ is
\begin{align}
\mu_{\sv}&=\sqrt{w_{\sv}}v_{0}+\sqrt{1-w_{\sv}}z_{\sv},\label{Eq:O_Mu_v}
\end{align}
where $v_{0}\sim P_{\sv}(\cdot)$ is the ground-truth signal component and $z_{\sv}\sim\cN(0,1)$ is standard Gaussian noise. Under this channel,
\begin{align}
\MMSE_{\sv}(w_{\sv})&=\E[(v_{0}-\langle v_{a}\rangle)^{2}]=1-m_{\sv},\label{Eq:SE_MMSE_SV}
\end{align}
where $m_{\sv}$ is the overlap between $v_{0}$ and the posterior mean $\langle v_{a}\rangle$. Comparing this channel with the replica scalar channel in \eqref{Eq:Mu_v} gives
\begin{align}
\hat{m}_{\sv}=\frac{w_{\sv}}{1-w_{\sv}}\quad\Longleftrightarrow\quad w_{\sv}=\frac{\hat{m}_{\sv}}{1+\hat{m}_{\sv}}.
\end{align}
Substituting this relation into \eqref{Eq:OAMP_rv} yields
\begin{align}
\rho_{\sv}&=\frac{1}{\MMSE_{\sv}(w_{\sv})}-\left(1-\frac{\hat{m}_{\sv}}{1+\hat{m}_{\sv}}\right)^{-1}=(1-m_{\sv})^{-1}-\hat{m}_{\sv}-1,
\end{align}
which recovers \eqref{Eq:rho_v}. The same argument for the $\su$-channel gives
\begin{align}
\MMSE_{\su}(w_{\su})&=1-m_{\su},\label{Eq:SE_MMSE_SU}\\
\rho_{\su}&=(1-m_{\su})^{-1}-\hat{m}_{\su}-1,
\end{align}
which recovers \eqref{Eq:rho_u}.

We next derive the MMSE expressions from the SE fixed point equations. Using \eqref{Eq:SE_MMSE_SV}, the MMSE for the $\sv$-channel can be written as
\begin{align}
1-m_{\sv}&=\MMSE_{\sv}(w_{\sv})\nonumber\\
&\overset{\text{\eqref{Eq:OAMP_rv}}}{=}[\rho_{\sv}+(1-w_{\sv})^{-1}]^{-1}\nonumber\\
&\overset{\text{\eqref{Eq:OAMP_wv}}}{=}\frac{1}{\rho_{\sv}}\E_{\gamma\sim\tilde{\mu}}[1-Q(\gamma;\rho_{\su},\rho_{\sv})]\nonumber\\
&\overset{\text{(A)}}{=}\frac{1}{\rho_{\sv}}\left[\frac{1}{\alpha}\E_{\gamma\sim\mu}[1-Q(\gamma;\rho_{\su},\rho_{\sv})]+\left(1-\frac{1}{\alpha}\right)(1-Q(0;\rho_{\su},\rho_{\sv}))\right]\nonumber\\
&=\frac{1}{\alpha}\E_{\gamma\sim\mu}\frac{\gamma(\rho_{\su}\varphi_{1}(\gamma)+1)\varphi_{2}(\gamma)-\rho_{\su}\varphi_{3}^{2}(\gamma)}{\gamma(\rho_{\su}\varphi_{1}(\gamma)+1)(\rho_{\sv}\varphi_{2}(\gamma)+\delta)-\rho_{\su}\rho_{\sv}\varphi_{3}^{2}(\gamma)}+\left(1-\frac{1}{\alpha}\right)\frac{\varphi_{2}(0)}{\rho_{\sv}\varphi_{2}(0)+\delta}\nonumber\\
&=\frac{1}{\alpha}\E_{\gamma\sim\mu}\frac{\rho_{\su}+\frac{\varphi_{2}(\gamma)}{\varphi_{1}(\gamma)\varphi_{2}(\gamma)-\frac{1}{\gamma}\varphi_{3}^{2}(\gamma)}}{\rho_{\su}\rho_{\sv}+\frac{\rho_{\su}\delta\varphi_{1}(\gamma)+\rho_{\sv}\varphi_{2}(\gamma)+\delta}{\varphi_{1}(\gamma)\varphi_{2}(\gamma)-\frac{1}{\gamma}\varphi_{3}^{2}(\gamma)}}+\left(1-\frac{1}{\alpha}\right)\left(\rho_{\sv}+\frac{\delta}{\varphi_{2}(0)}\right)^{-1}\nonumber\\
&\overset{\text{(B)}}{=}\frac{1}{\alpha}\E_{\gamma\sim\mu}\frac{\rho_{\su}+1+\cJ_{\Sf{uu}}(\gamma)}{\rho_{\su}\rho_{\sv}+\rho_{\su}(1+\cJ_{\Sf{vv}}(\gamma))+\rho_{\sv}(1+\cJ_{\Sf{uu}}(\gamma))+(1+\cJ_{\Sf{uu}}(\gamma))(1+\cJ_{\Sf{vv}}(\gamma))-\cJ_{\Sf{uv}}(\gamma)}\nonumber\\
&\qquad+\left(1-\frac{1}{\alpha}\right)(\rho_{\sv}+1+\cJ_{\Sf{vv}}(0))^{-1}.\label{Eq:SE_MMSE_V}
\end{align}
Here, step (A) follows from $\tilde{\mu}=\frac{1}{\alpha}\mu+(1-\frac{1}{\alpha})\delta_{0}$. In step (B), we substitute \eqref{Eq:Rho_1}-\eqref{Eq:Rho_3} and the parameter mapping \eqref{Eq:Mapping}; after simplification, the fractions involving $\varphi_{1}(\cdot)$, $\varphi_{2}(\cdot)$, and $\varphi_{3}(\cdot)$ reduce to the expressions involving $\cJ_{\Sf{uu}}(\cdot)$, $\cJ_{\Sf{vv}}(0)$, $\cJ_{\Sf{vv}}(\cdot)$, and $\cJ_{\Sf{uv}}(\cdot)$ given in \eqref{Eq:Juu_X}-\eqref{Eq:Juv_X}. Similarly, using \eqref{Eq:OAMP_ru}, \eqref{Eq:OAMP_wu}, and \eqref{Eq:SE_MMSE_SU}, we obtain
\begin{align}
1-m_{\su}&=\E_{\gamma\sim\mu}\frac{\rho_{\sv}+1+\cJ_{\Sf{vv}}(\gamma)}{\rho_{\su}\rho_{\sv}+\rho_{\su}(1+\cJ_{\Sf{vv}}(\gamma))+\rho_{\sv}(1+\cJ_{\Sf{uu}}(\gamma))+(1+\cJ_{\Sf{uu}}(\gamma))(1+\cJ_{\Sf{vv}}(\gamma))-\cJ_{\Sf{uv}}(\gamma)}.\label{Eq:SE_MMSE_U}
\end{align}
Thus, \eqref{Eq:SE_MMSE_V}-\eqref{Eq:SE_MMSE_U}, derived from the SE fixed point equations, coincide with \eqref{Eq:1_m_v} and \eqref{Eq:1_m_u}, respectively. Together with the scalar channel matching that yields \eqref{Eq:rho_u}-\eqref{Eq:rho_v}, this establishes the algebraic equivalence between the replica saddle point equations in Result \ref{Result:Saddle_Point} and the SE fixed point equations \eqref{Eq:OAMP_ru}-\eqref{Eq:OAMP_wv}.

\subsection{Implications of the Equivalence}\label{Sec:Bayes_Optimality_Prediction}

We now discuss the algorithmic implications of the equivalence established above. The replica method used to derive the IT limits is a heuristic but powerful framework originating from the statistical physics of disordered systems \cite{parisi1980sequence,mezard1987spin,mezard2009information} and widely applied in high-dimensional Bayesian inference \cite{takahashi2022macroscopic,barbier2023fundamental,barbier2025information}. Although the replica method does not provide a rigorous proof in general, many of its predictions have been rigorously established in spin glass and inference models; see, for example, \cite{guerra2003broken,talagrand2006parisi,panchenko2013parisi,reeves2016replica,miolane2017fundamental,barbier2018mutual,barbier2019adaptive,barbier2019optimal,barbier2020mutual,mourrat2020hamilton}. We therefore conjecture that the IT limits characterized by Result \ref{Result:Saddle_Point} are asymptotically exact.

Under this conjecture, the equivalence provides an algorithmic interpretation of the replica saddle point equations. In particular, the fixed point reached by the SE recursion corresponds to a stationary solution of the replica saddle point equations. In the unique fixed point regime, this correspondence implies that the SE recursion has no competing stationary branch to select, and its asymptotic performance coincides with the IT limit predicted by the replica analysis. Hence, OAMP is conjectured to achieve Bayes-optimal performance in this regime. When multiple fixed points coexist, however, the equivalence alone does not guarantee that the SE recursion reaches the IT optimal branch, namely the stationary solution maximizing $f_{\Sf{RS}}$ in \eqref{eq:RSfreeEnt_extr}. In such algorithmically hard regimes \cite{lesieur2017statistical,antenucci2019glassy,barbier2019optimal,barbier2020mutual,angelini2023limits}, the SE recursion with an uninformative initialization may converge to a metastable branch with lower free entropy rather than the IT branch, resulting in a statistical-to-computational gap for OAMP \cite{pourkamali2024matrix,defilippis2026noise,dudeja2026optimality}. This behavior is consistent with the conjecture that no polynomial-time algorithm can achieve the IT limit from an uninformative initialization in these regimes \cite{dudeja2026optimality}. Thus, the equivalence reveals the SE trajectory as an algorithmic selection mechanism among stationary solutions of the replica saddle point equations, and OAMP is expected to be Bayes-optimal precisely when the selected branch coincides with the IT optimal branch.

\subsection{Numerical Examples}\label{Sec:Numerical_Validation}

This section numerically solves the saddle point equations in Result \ref{Result:Saddle_Point}, studies the resulting stationary branches, and compares them with the SE fixed points of OAMP. All details concerning how the various solutions of the saddle point equations are numerically found are described in Appendix \ref{App:NumericalSP}.

The numerical examples are organized according to the fixed point structure of the saddle point equations. We first consider a regime with multiple fixed points, where the distinction between the IT prediction and the algorithmic prediction becomes nontrivial. Among the coexisting stationary solutions, the IT prediction corresponds to the solution maximizing the RS free entropy $f_{\Sf{RS}}$, whereas the algorithmic prediction is determined by the fixed point selected by the SE dynamics. These two selections need not coincide, leading to a possible statistical-to-computational gap. We then consider regimes with a unique detected informative fixed point, where no competing informative stationary branch is detected in the prescribed numerical domain and OAMP is conjectured to achieve the IT limit.

In the numerical experiments, empirical performance is measured by the squared cosine similarity between the estimated and ground-truth signals, while the replica prediction is expressed through the overlap with the posterior mean estimate. For the $\su$-channel, this quantity is
\begin{align}
m_{\su}&=\E[u_{0}\langle u_{a}\rangle].
\end{align}
The $\sv$-channel is evaluated analogously. At a stationary point, these quantities coincide with the overlaps $m_{\su}^{*}$ and $m_{\sv}^{*}$ defined in \eqref{Eq:overlaps}.

In the figures, ``IT Limit'' denotes the stationary solution with the largest replica free entropy among all detected stationary solutions. When multiple stationary branches coexist, ``Metastable State'' denotes a branch with smaller free entropy. The label ``PCA'' refers to the spectral initialization proposed in \cite{Chen2026orthogonal}, rather than the standard PCA based only on the leading singular vector. In rectangular RI models, a rank-one signal can generate multiple informative singular outliers, and the PCA baseline considered here combines these informative outliers to construct the initialization. See Section 4 of \cite{Chen2026orthogonal} for the detailed construction. The label ``OAMP'' denotes the empirical performance of the OAMP algorithm \cite{Chen2026orthogonal} at the finite dimensions $N=4000$ and $M=8000$, whereas ``SE'' denotes the corresponding asymptotic performance predicted by its SE. Since the SE of OAMP is rigorously characterized, the agreement between its fixed points and the replica saddle point solutions provides numerical evidence supporting the replica prediction.\\

\begin{figure}[!t]
\centering
\begin{subfigure}{1.0\linewidth}
\centering
\stepcounter{subfigure}
\begin{subfigure}{0.49\linewidth}
\centering
\includegraphics[width=\linewidth]{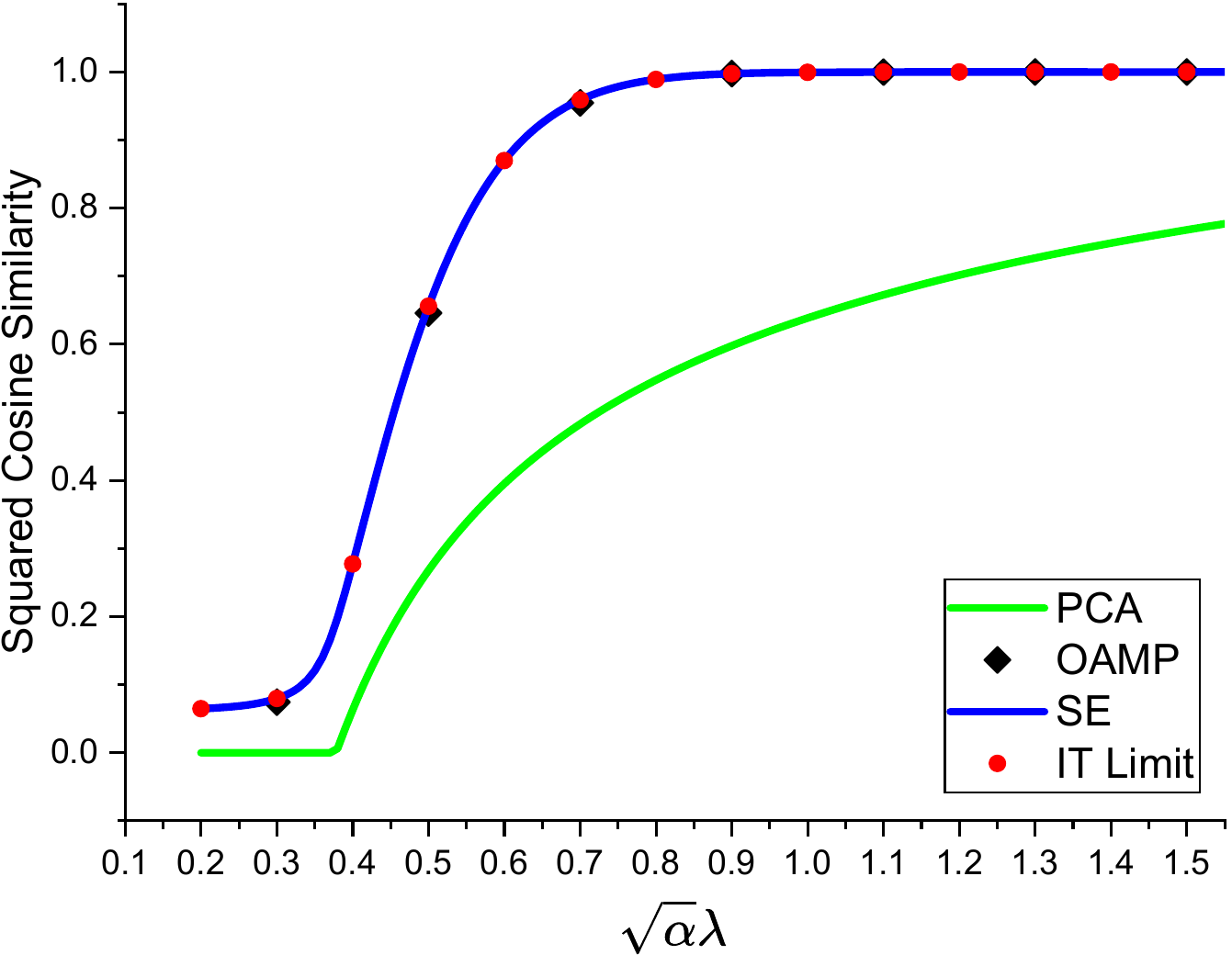}
\caption*{(\thesubfigure.1) $\su$-channel.}
\end{subfigure}
\hfill
\begin{subfigure}{0.49\linewidth}
\centering
\includegraphics[width=\linewidth]{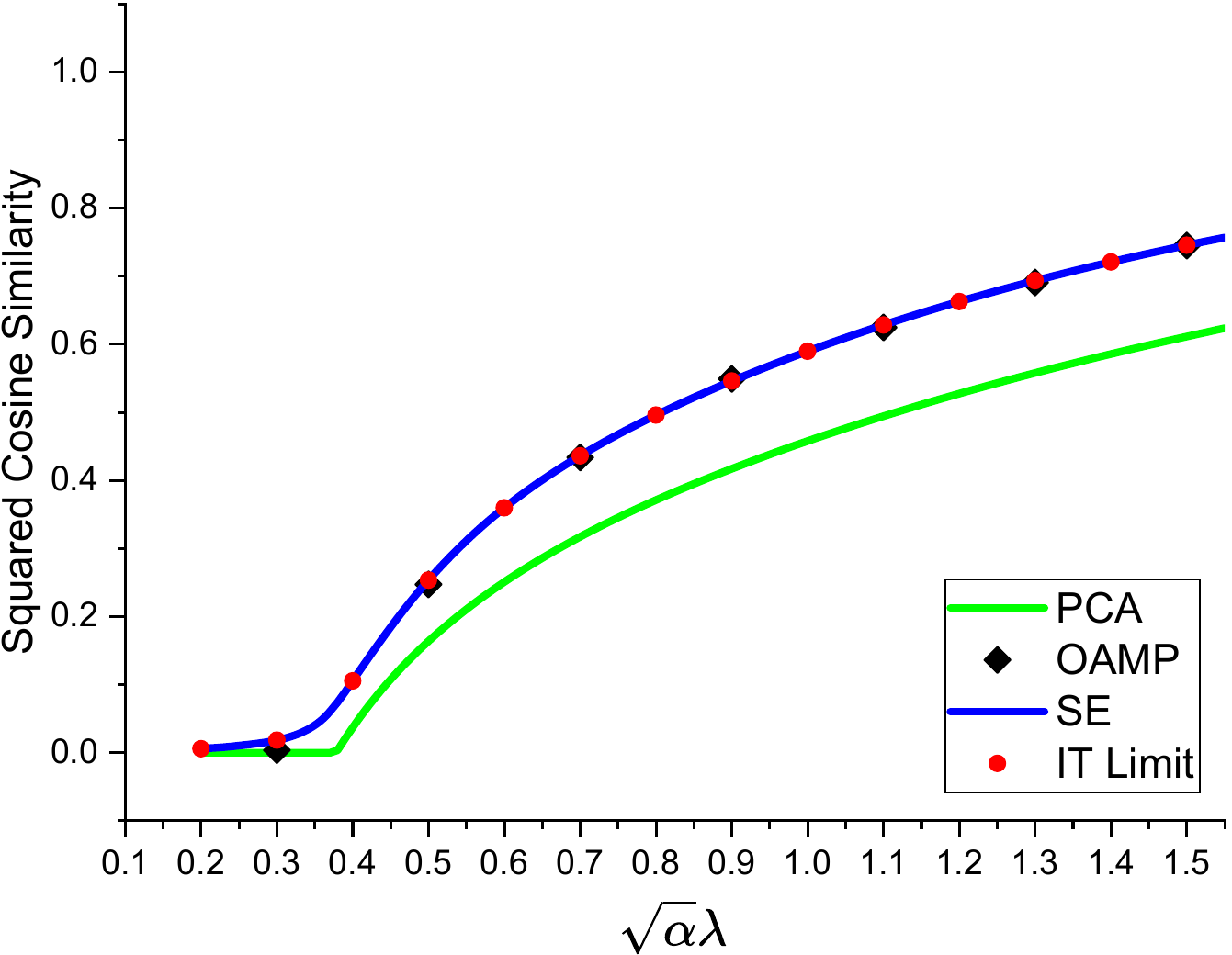}
\caption*{(\thesubfigure.2) $\sv$-channel.}
\end{subfigure}
\addtocounter{subfigure}{-1}
\caption{Single-branch case with $\epsilon=\frac{1}{4}$.}
\label{Fig:Main_SE_RM_Rad_1_3_2Points_0_25_Gauss_Ber_0_4_A_2_T_0_2_0_01_2}
\end{subfigure}
\begin{subfigure}{1.0\linewidth}
\centering
\stepcounter{subfigure}
\begin{subfigure}{0.49\linewidth}
\centering
\includegraphics[width=\linewidth]{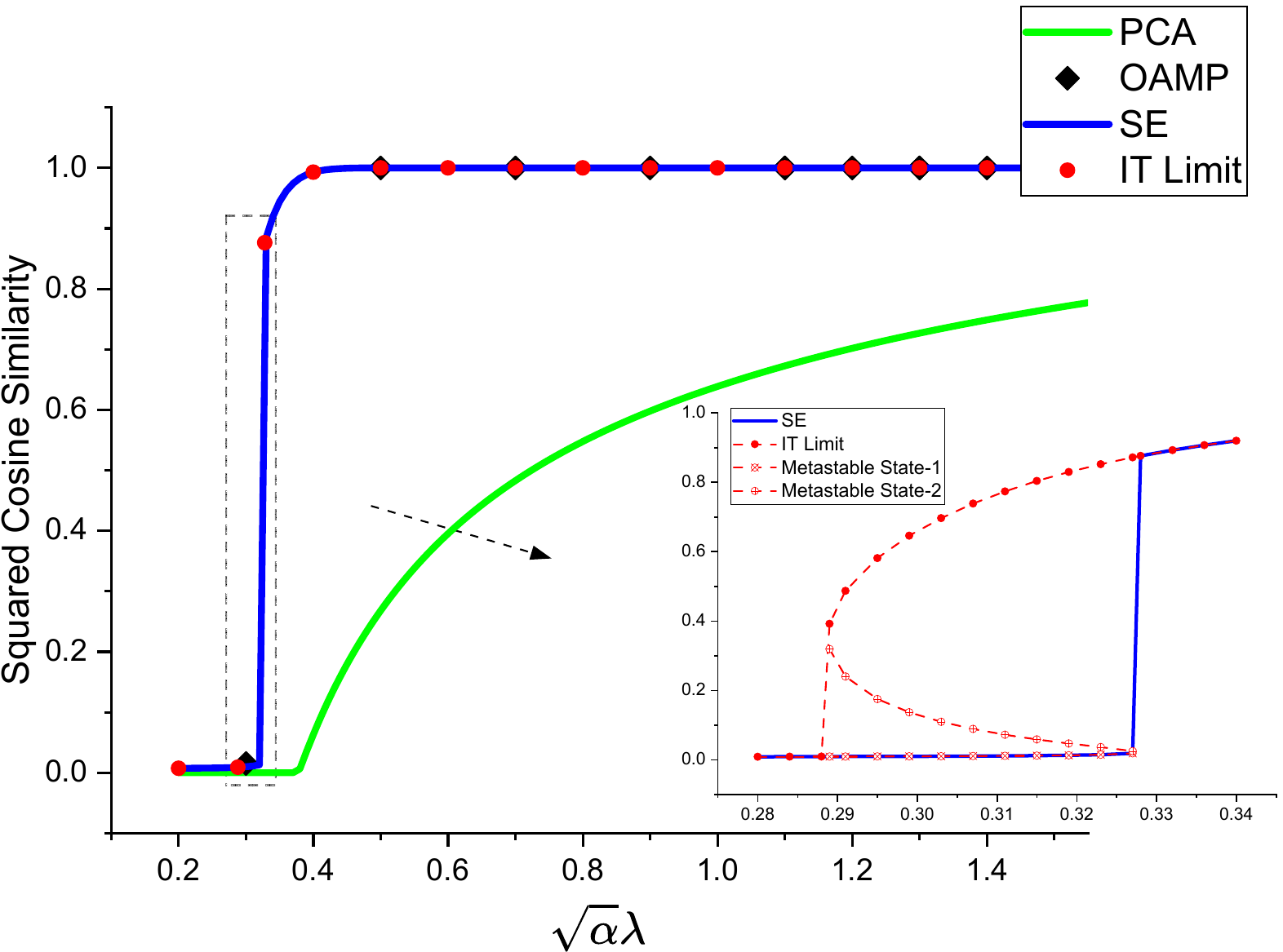}
\caption*{(\thesubfigure.1) $\su$-channel.}
\end{subfigure}
\hfill
\begin{subfigure}{0.49\linewidth}
\centering
\includegraphics[width=\linewidth]{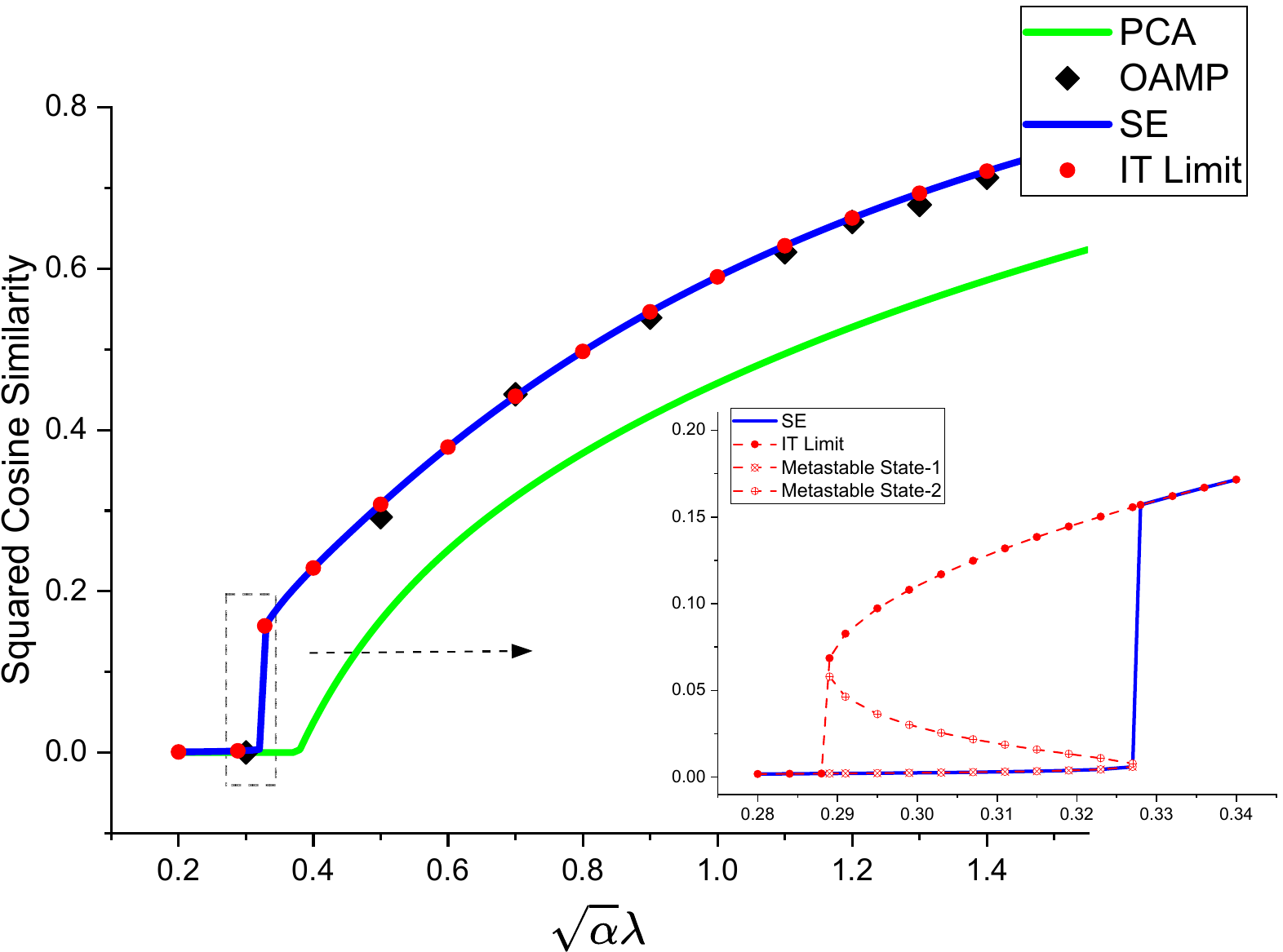}
\caption*{(\thesubfigure.2) $\sv$-channel.}
\end{subfigure}
\addtocounter{subfigure}{-1}
\caption{Multiple-branch case with $\epsilon=\frac1{12}$. The insets magnify the transition region where multiple stationary branches coexist.}
\label{Fig:Main_SE_RM_Rad_1_3_2Points_0_083333_Gauss_Ber_0_4_A_2_T_0_2_0_01_2}
\end{subfigure}
\caption{Sparse-prior examples with $\rho = 0.4$ illustrating the transition from a single stationary branch to multiple coexisting stationary branches. The RI noise is defined in \eqref{eq: RI_noise_ex_1} with $\alpha=2$, $a=1$ and $b=3$. For the Bayes-optimal estimator, the squared cosine similarity on the vertical axis coincides asymptotically with the overlaps $m_{\su}^{*}$ and $m_{\sv}^{*}$ of \eqref{Eq:overlaps}.}
\label{Fig:Main_SE_RM_Rad}
\end{figure}

\begin{figure}[!t]
\centering
\begin{subfigure}{0.49\linewidth}
\centering
\includegraphics[width=\linewidth]{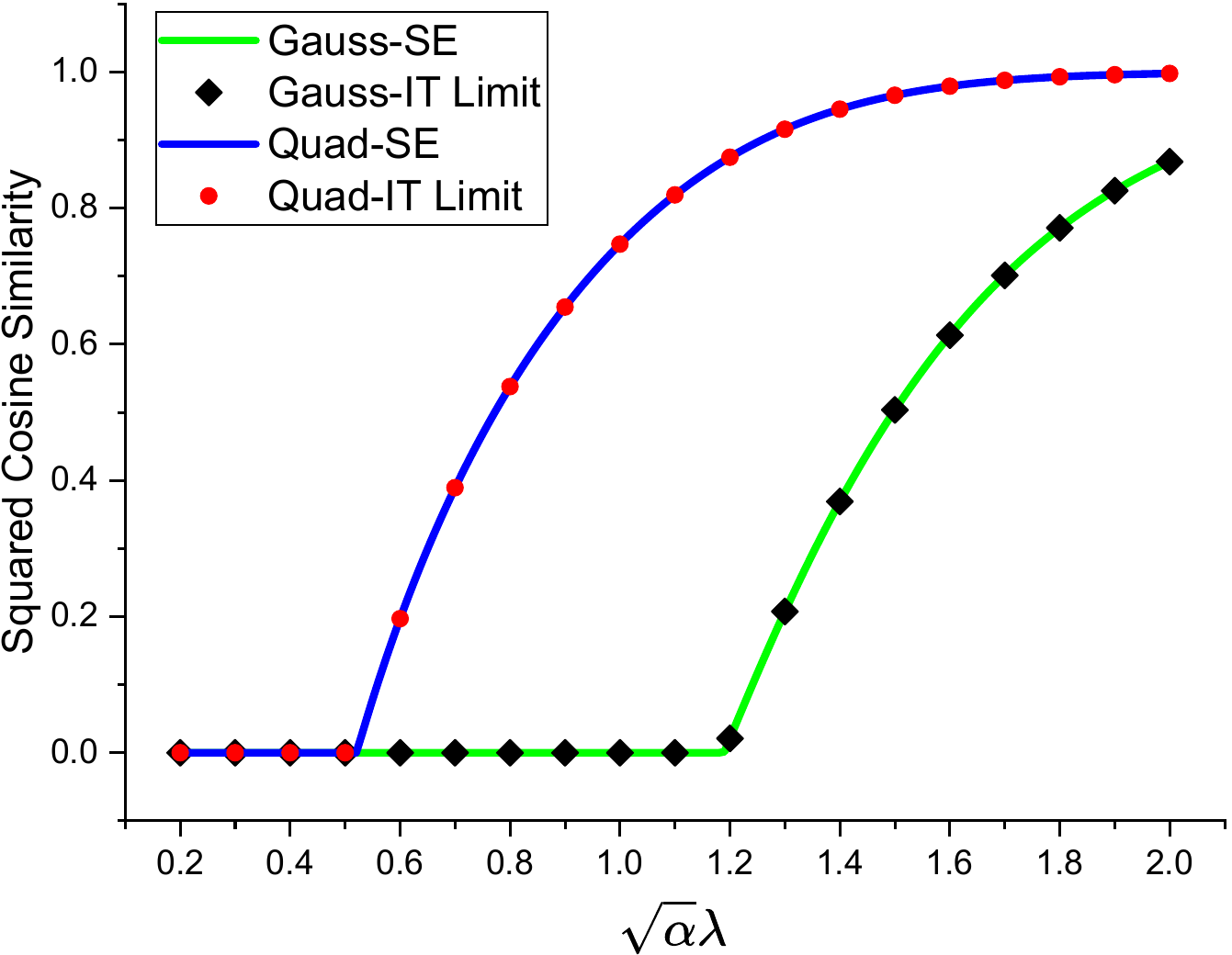}
\caption{$\su$-channel}
\end{subfigure}
\hfill
\begin{subfigure}{0.49\linewidth}
\centering
\includegraphics[width=\linewidth]{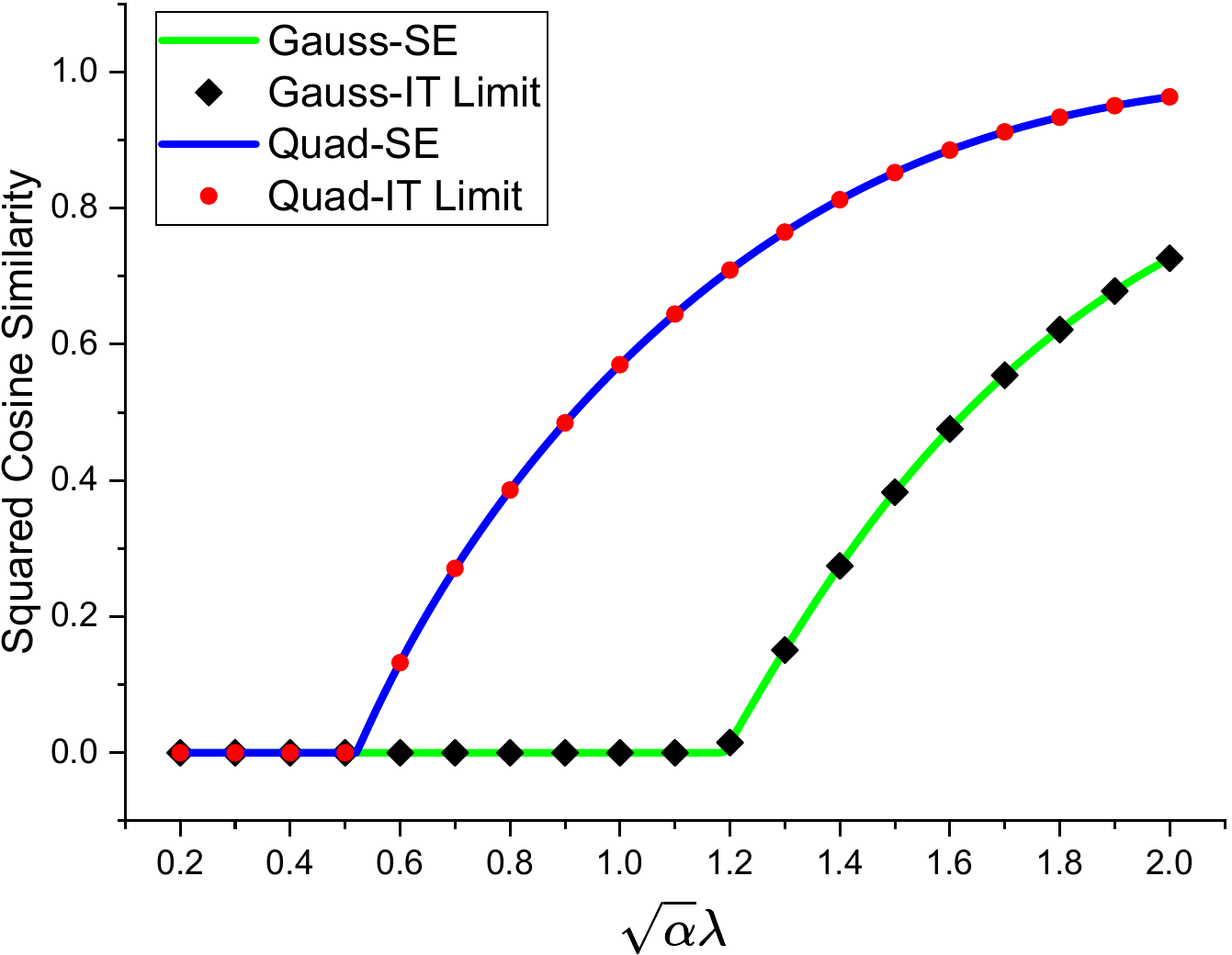}
\caption{$\sv$-channel}
\end{subfigure}
\caption{Numerical comparison in the unique informative fixed point regime for quadratic potential RI noise and standard i.i.d. Gaussian noise at the same aspect ratio $\alpha=2$. For each noise model, the fixed-point iteration and the SE fixed points of OAMP identify the same performance curve. The ground-truth signal follows the symmetric Rademacher prior defined in \eqref{eq:rademacher_prior}. As in Fig.~\ref{Fig:Main_SE_RM_Rad}, the vertical axis reports the overlaps $m_{\su}^{*}$ and $m_{\sv}^{*}$.}
\label{Fig:Main_SE_RM}
\end{figure}

\subsubsection{Regime with Multiple Fixed Points}

This subsection studies a sparse-prior setting in which the stationary branch structure changes with the sparsity level. This regime is particularly relevant because the replica free entropy comparison is required to distinguish the IT branch from metastable branches. We first consider a baseline case with a single detected stationary branch and then decrease the sparsity level to reveal a regime with multiple coexisting branches.

The ground-truth signal pair $\{\bu_{0},\bv_{0}\}$ is drawn from a two-point prior and a Gauss-Bernoulli prior, respectively:
\begin{align}
P_{\su}(u_{0})\sim\epsilon^{2}\delta_{\frac{1}{\epsilon}}+(1-\epsilon^{2})\delta_{0},\quad P_{\sv}(v_{0})\sim\rho\Normal\left(0,\frac{1}{\rho}\right)+(1-\rho)\delta_{0}.
\end{align}
The system is corrupted by an RI noise model following the example considered in \cite{Chen2026orthogonal}, with spectral density and potential given by
\begin{align}
\label{eq: RI_noise_ex_1}
\mu(\gamma)&=\frac{2}{\pi}\sqrt{(\gamma-a)(b-\gamma)},\quad V(\gamma)=2[\gamma^{2}-(a+b)\gamma]+(\alpha-1)\log\gamma,\quad\gamma\in[a,b].
\end{align}
This model has an explicit limiting spectral density and provides a convenient setting to study the replica-SE correspondence beyond the Marchenko-Pastur case.

As shown in Fig. \ref{Fig:Main_SE_RM_Rad_1_3_2Points_0_25_Gauss_Ber_0_4_A_2_T_0_2_0_01_2}, when $\epsilon=\frac{1}{4}$, the solution obtained by fixed-point iteration coincides with the only stationary solution detected by the grid-based residual minimization described in Appendix \ref{App:NumericalSP}. The SE prediction of OAMP agrees with the replica prediction, and no competing stationary solution is detected in the prescribed numerical domain. The finite-dimensional OAMP results at $N=4000$ and $M=8000$ also closely follow the SE prediction. Since no free-entropy-based selection among the detected solutions is required, the unique detected solution is taken as the IT prediction. Hence, OAMP is conjectured to achieve the IT limit in this regime.

We next decrease the sparsity parameter to $\epsilon=\frac{1}{12}$. As shown in Fig. \ref{Fig:Main_SE_RM_Rad_1_3_2Points_0_083333_Gauss_Ber_0_4_A_2_T_0_2_0_01_2}, the fixed-point iteration and the finite-dimensional OAMP results continue to follow the SE prediction over a broad parameter range. However, the additional stationary solutions occur within a narrow transition window that is not visible at the scale of the main panels. This region is therefore magnified in the corresponding insets.

To examine the transition, we apply the grid-based residual minimization described in Appendix \ref{App:NumericalSP} over the interval $\sqrt{\alpha}\lambda\in[0.289,0.327]$. Within this interval, the search reveals a region in which the saddle point equations \eqref{Eq:Unit_Saddle_Point} admit three stationary solutions. At each SNR, we evaluate $f_{\Sf{RS}}$ at every detected solution. The solution attaining the largest value is identified as the IT solution, while the other two have smaller free entropies and are classified as metastable. The IT selection changes at $\lambda_{\Sf{IT}}$, whereas the fixed-point iteration, initialized close to zero, and the SE of OAMP continue to select one of the metastable solutions until $\lambda_{\Sf{alg}}$. The finite-dimensional OAMP results follow the same SE-selected solution. Consequently, for $\lambda\in(\lambda_{\Sf{IT}},\lambda_{\Sf{alg}})$, the OAMP prediction differs from the IT prediction, giving rise to the statistical-to-computational gap.

This example demonstrates that the equivalence between the replica saddle point equations and the SE fixed point equations of OAMP does not imply that the SE dynamics always select the IT solution. The equivalence identifies the same set of stationary solutions, whereas the SE dynamics select a particular solution through their initialization and basin of attraction. In the transition region, the change between the selected solutions occurs simultaneously in the $\su$- and $\sv$-channels because their fixed point equations are coupled. Nevertheless, the two channels can exhibit substantially different reconstruction performance: the $\su$-channel may approach nearly perfect recovery, whereas the $\sv$-channel can remain only partially recovered.\\

\subsubsection{Regime with a Unique Informative Fixed Point}

We next consider a regime where only one informative stationary branch is detected. In this setting, no competing informative branch is detected, and the comparison between the replica prediction and the SE prediction of OAMP is direct. Specifically, we consider the quadratic potential RI noise model characterized by
\begin{align}
\mu(\gamma)&=\frac{\kappa}{2\pi}\sqrt{(\gamma-a)(b-\gamma)}\left(1+\frac{a+b}{2\gamma}\right),\quad V(\gamma)=\frac{\kappa}{2}\gamma^{2},\quad\gamma\in[a,b],
\end{align}
where $\kappa>0$ and $[a,b]$ is the associated spectral support. The explicit expressions for the support endpoints are standard and are omitted here for brevity; see Chapter 4 of \cite{potters2020first} and Chapter 1, particularly Eqs. (1.67)-(1.68), of \cite{log_gases}. In this numerical experiment, the ground-truth signal pair $\{\bu_{0},\bv_{0}\}$ follows a symmetric Rademacher prior:
\begin{align}
P_{\su}(\cdot)&=P_{\sv}(\cdot)=\frac12\delta_{+1}+\frac12\delta_{-1}.\label{eq:rademacher_prior}
\end{align}
For these symmetric priors, we assume infinitesimal side information to resolve the global sign ambiguity, with its strength taken to zero after the thermodynamic limit.

In the numerical experiment, we set $\kappa=2$ for the quadratic potential model and compare it with standard i.i.d. Gaussian noise at the same aspect ratio $\alpha=2$. As shown in Fig. \ref{Fig:Main_SE_RM}, for each noise model, the replica fixed-point iteration and the OAMP SE yield the same performance curve. The grid-based search detects no competing informative stationary branch in the prescribed numerical domain for either noise model, so the unique detected informative branch is taken as the corresponding IT prediction. For each noise model, the overlaps in both channels vanish below the transition, become nonzero at the same threshold, and then increase with the SNR. Compared with the standard i.i.d. Gaussian case, the quadratic potential RI noise model exhibits an earlier transition, after which the overlaps increase more rapidly and remain higher over the range shown. Unlike the sparse-prior setting considered above, no coexistence window or switching between competing informative stationary solutions is observed.

Although the two channels undergo the transition simultaneously within each noise model because their fixed point equations are coupled, their reconstruction performance can differ owing to the rectangular asymmetry of the model. In particular, under quadratic potential RI noise, the $\su$-channel can approach essentially perfect recovery, whereas the $\sv$-channel remains slightly below perfect recovery over the range shown. Under standard i.i.d. Gaussian noise, the overlaps in both channels remain noticeably lower, particularly in the $\sv$-channel. This example provides numerical evidence consistent with the equivalence established above: with no competing informative stationary solutions detected in the prescribed numerical domain, the replica fixed-point iteration and the OAMP SE yield the same unique detected informative solution for each noise model.

\section{Derivation of the Information-Theoretic Results by the Replica Method}\label{Sec:RM}

This section derives the replica predictions stated in Result \ref{Result:Free_Entropy} and Result \ref{Result:Saddle_Point}. We will start by reformulating the partition function $\Omega(\bZ_{0},\bY)$ in \eqref{Eq:PartitionFunction} in a more convenient form. The derivation of the free entropy is then presented, followed by the simplification of the saddle point equations under the Bayes-optimal setting. Technical algebraic steps are deferred to Appendix~\ref{Sec:Proof_RM}.

\subsection{Reformulating the Free Entropy}

\subsubsection{Block-Symmetrization of the Trace Potential}\label{Sec:Block_Symmetrization}

We begin by reformulating the trace potential $\Tr\,V(\bZ_{a}\bZ_{a}^{\top})$ appearing in \eqref{Eq:PartitionFunction}. This potential is naturally expressed through the resolvent of $\bZ_{a}\bZ_{a}^{\top}$, which remains tied to the rectangular formulation and cannot be handled directly by the replica framework developed for symmetric RI models \cite{barbier2023fundamental,barbier2025information,dudeja2026optimality}. We will first introduce a block-symmetrization of the trace potential, and then factorize the partition function into a noise-dependent coefficient and a candidate-signal-dependent contribution. The latter is further rewritten through matrix-valued functions and functional Dirac constraints, providing the functional representation to which the replica analysis is subsequently applied.

Since the analytic potential $V(\cdot)$ acts on the spectrum of $\bZ_{a}\bZ_{a}^{\top}$, Cauchy's integral formula provides a convenient contour representation that separates the matrix dependence from the analytic potential \cite{bodin2024random,barbier2025information}. Let $\{\sigma_{i}^{2}\}_{i=1}^{N}$ denote the positive eigenvalues of $\bZ_{a}\bZ_{a}^{\top}$. Under our assumptions on the noise and the signal priors, we choose a fixed closed contour $\Gamma_{\sy}$ that is common to all indices $a$ and, with probability one, encloses $\{\sigma_{i}^{2}\}_{i=1}^{N}$ for all sufficiently large dimensions. Then,
\begin{align}
\Tr\,V(\bZ_{a}\bZ_{a}^{\top})&=\frac{1}{2\pi\imagUnit}\oint_{\Gamma_{\sy}}{\dd y}\,V(y)\Tr\,\cG_{\bZ_{a}\bZ_{a}^{\top}}(y)\nonumber\\
&=\frac{1}{\pi\imagUnit}\oint_{\Gamma_{\sx}}{\dd x}\,xV(x^{2})\Tr\,\cG_{\bZ_{a}\bZ_{a}^{\top}}(x^{2}),\label{Eq:Cauchy_{i}ntegral}
\end{align}
where the second equality follows from the change of variable $y=x^{2}$. The positive square-root branch maps $\Gamma_{\sy}$ onto a fixed closed contour $\Gamma_{\sx}$, common to all indices $a$ and enclosing the positive singular values $\{\sigma_{i}\}_{i=1}^{N}$.

Equation \eqref{Eq:Cauchy_{i}ntegral} expresses the trace potential in terms of the resolvent of $\bZ_{a}\bZ_{a}^{\top}$. However, this resolvent is still associated with the rectangular formulation, preventing the direct application of the replica framework developed for symmetric RI models \cite{barbier2023fundamental,barbier2025information}. To recover the required spectral symmetry, we introduce the auxiliary block-symmetric matrix \cite{pourkamali2023rectangularrotationalinvariantestimator} (below matrices of zeroes are taken with proper dimensions):
\begin{align}
\tilde{\bZ}_{a}&\triangleq\begin{bmatrix}\bm{0}&\bZ_{a}\\\bZ_{a}^{\top}&\bm{0}\end{bmatrix}=\tilde{\bY}-\frac{\lambda}{N}\tilde{\bX}_{a}\tilde{\bI}\tilde{\bX}_{a}^{\top}\in\R^{(N+M)\times(N+M)},
\end{align}
where 
\begin{align}
\tilde{\bY}\triangleq\begin{bmatrix}\bm{0}&\bY\\\bY^{\top}&\bm{0}\end{bmatrix}\in\R^{(N+M)\times(N+M)},\quad\tilde{\bX}_{a}\triangleq\begin{bmatrix}\bu_{a}&\bm{0}\\\bm{0}&\bv_{a}\end{bmatrix}\in\R^{(N+M)\times2},\quad\tilde{\bI}\triangleq\begin{bmatrix}0&1\\1&0\end{bmatrix}\in\R^{2\times2}.
\end{align}
The spectrum of $\tilde{\bZ}_{a}$ consists of the eigenvalues $\{\pm\sigma_{i}\}_{i=1}^{N}$ together with $M-N$ zero eigenvalues. Consequently, the corresponding resolvent traces satisfy
\begin{align}
\Tr\,\cG_{\tilde{\bZ}_{a}}(x)=\sum_{i=1}^{N}\left(\frac{1}{x-\sigma_{i}}+\frac{1}{x+\sigma_{i}}\right)+\frac{M-N}{x},\quad\Tr\,\cG_{\bZ_{a}\bZ_{a}^{\top}}(x^{2})=\sum_{i=1}^{N}\frac{1}{x^{2}-\sigma^{2}_{i}},
\end{align}
which implies
\begin{align}
\Tr\,\cG_{\bZ_{a}\bZ_{a}^{\top}}(x^{2})&=\frac{1}{2x}\left(\Tr\,\cG_{\tilde{\bZ}_{a}}(x)-\frac{M-N}{x}\right).
\end{align}

Substituting this identity into \eqref{Eq:Cauchy_{i}ntegral} yields
\begin{align}
\Tr\,V(\bZ_{a}\bZ_{a}^{\top})&=\frac{1}{\pi\imagUnit}\oint_{\Gamma_{\sx}}{\dd x}\,xV(x^{2})\Tr\,\cG_{\bZ_{a}\bZ_{a}^{\top}}(x^{2})\nonumber\\
&=\frac{1}{2\pi\imagUnit}\oint_{\Gamma_{\sx}}{\dd x}\,V(x^{2})\Tr\,\cG_{\tilde{\bZ}_{a}}(x)-\frac{1}{2\pi\imagUnit}\oint_{\Gamma_{\sx}}{\dd x}\,V(x^{2})\frac{M-N}{x}.\label{Eq:Trace_V}
\end{align}

Equation \eqref{Eq:Trace_V} reformulates the trace potential in terms of the resolvent of the auxiliary symmetric matrix while preserving the probabilistic structure of the original rectangular model. In particular, the two signal vectors remain structurally independent throughout the transformation, distinguishing the present setting from symmetric models \cite{barbier2023fundamental,barbier2025information,dudeja2026optimality}. This equivalent symmetric formulation provides the analytical foundation for extending the replica analysis developed for symmetric models \cite{barbier2023fundamental,barbier2025information} to the rectangular setting considered in this work.\\

\subsubsection{Factorization of the Partition Function}\label{sec:matrix_decoupling}

We reformulate the partition function $\Omega(\bZ_{0},\bY)$ in \eqref{Eq:PartitionFunction} as follows:
\begin{align}
\Omega(\bZ_{0},\bY)&\overset{\text{\eqref{Eq:Trace_V}}}{=}\int{\dd P_{\su}(\bu_{a})\dd P_{\sv}(\bv_{a})}\,\exp\Bigg(-\frac{N}{4\pi\imagUnit}\oint_{\Gamma_{\sx}}{\dd x}\,V(x^{2})\Tr\,\cG_{\tilde{\bZ}_{a}}(x)+\frac{N}{4\pi\imagUnit}\oint_{\Gamma_{\sx}}{\dd x}\,V(x^{2})\frac{M-N}{x}\Bigg)\exp\left(\frac{N}{2}\Tr\,V(\bZ_{0}\bZ_{0}^{\top})\right)\nonumber\\
&\overset{\text{(A)}}{=}\int{\dd P_{\su}(\bu_{a})\dd P_{\sv}(\bv_{a})}\,\exp\Bigg(-\frac{N}{4\pi\imagUnit}\oint_{\Gamma_{\sx}}{\dd x}\,V(x^{2})\Tr\,\left(\cG_{\tilde{\bY}}(x)-\frac{\lambda}{N}\cG_{\tilde{\bY}}(x)\tilde{\bX}_{a}\left(\tilde{\bI}+\frac{\lambda}{N}\tilde{\bX}_{a}^{\top}\cG_{\tilde{\bY}}(x)\tilde{\bX}_{a}\right)^{-1}\tilde{\bX}_{a}^{\top}\cG_{\tilde{\bY}}(x)\right)\nonumber\\
&\quad+\frac{N}{4\pi\imagUnit}\oint_{\Gamma_{\sx}}{\dd x}\,V(x^{2})\frac{M-N}{x}\Bigg)\exp\left(\frac{N}{2}\Tr\,V(\bZ_{0}\bZ_{0}^{\top})\right)\nonumber\\
&\overset{\text{(B)}}{=}\int{\dd P_{\su}(\bu_{a})\dd P_{\sv}(\bv_{a})}\,\exp\Bigg(-\frac{N}{4\pi\imagUnit}\oint_{\Gamma_{\sx}}{\dd x}\,V(x^{2})\left(\Tr\,\cG_{\tilde{\bY}}(x)+\frac{\dd}{\dd x}\log\left|\tilde{\bI}+\frac{\lambda}{N}\tilde{\bX}_{a}^{\top}\cG_{\tilde{\bY}}(x)\tilde{\bX}_{a}\right|\right)\nonumber\\
&\quad+\frac{N}{4\pi\imagUnit}\oint_{\Gamma_{\sx}}{\dd x}\,V(x^{2})\frac{M-N}{x}\Bigg)\exp\left(\frac{N}{2}\Tr\,V(\bZ_{0}\bZ_{0}^{\top})\right)\nonumber\\
&\overset{\text{(C)}}{=}\int{\dd P_{\su}(\bu_{a})\dd P_{\sv}(\bv_{a})}\,\exp\Bigg(\frac{N}{2\pi\imagUnit}\oint_{\Gamma_{\sx}}{\dd x}\,xV^{\prime}(x^{2})\log\left|\tilde{\bI}+\frac{\lambda}{N}\tilde{\bX}_{a}^{\top}\cG_{\tilde{\bY}}(x)\tilde{\bX}_{a}\right|\nonumber\\
&\quad-\frac{N}{4\pi\imagUnit}\oint_{\Gamma_{\sx}}{\dd x}\,V(x^{2})\Tr\,\cG_{\tilde{\bY}}(x)+\frac{N}{4\pi\imagUnit}\oint_{\Gamma_{\sx}}{\dd x}\,V(x^{2})\frac{M-N}{x}\Bigg)\exp\left(\frac{N}{2}\Tr\,V(\bZ_{0}\bZ_{0}^{\top})\right)\nonumber\\
&\overset{\text{(D)}}{=}\underbrace{\exp\left(-\frac{N}{2}\left(\Tr\,V(\bY\bY^{\top})-\Tr\,V(\bZ_{0}\bZ_{0}^{\top})\right)\right)}_{\triangleq\varphi(\bZ_{0},\bY)}\nonumber\\
&\quad\times\underbrace{\int{\dd P_{\su}(\bu_{a})\dd P_{\sv}(\bv_{a})}\,\exp\left(\frac{N}{2\pi\imagUnit}\oint_{\Gamma_{\sx}}{\dd x}\,xV^{\prime}(x^{2})\log\left|\tilde{\bI}+\frac{\lambda}{N}\tilde{\bX}_{a}^{\top}\cG_{\tilde{\bY}}(x)\tilde{\bX}_{a}\right|\right)}_{\triangleq\phi(\bY)}.\label{Eq:Partition_Func}
\end{align}
Here, step (A) follows from the Woodbury matrix identity. Step (B) uses $\dd\log|\bX|=\Tr\,(\bX^{-1}\dd\bX)$ and $\dd\cG_{\tilde{\bY}}(x)=-\cG_{\tilde{\bY}}^{2}(x)\dd x$ to rewrite the second trace term as a logarithmic determinant derivative. Step (C) follows from integration by parts in the contour integral. Step (D) follows from applying the contour representation in \eqref{Eq:Trace_V} to the observation matrix $\bY$. This factorization separates the term $\varphi(\bZ_{0},\bY)$, which depends only on the observation and noise matrices and is analyzed separately in Appendix \ref{Sec:C_y}, from the remaining term $\phi(\bY)$ used in the subsequent free entropy reformulation.\\

\subsubsection{Final Steps}

Using the factorization of the partition function in \eqref{Eq:Partition_Func}, the asymptotic free entropy in \eqref{Eq:Free_Entropy} becomes
\begin{align}
f&=\lim_{N\to\infty}\frac{1}{N}\E_{\bu_{0},\bv_{0},\bZ_{0}}\log\Omega(\bZ_{0},\bY)\nonumber\\
&=\lim_{N\to\infty}\frac{1}{N}\E_{\bu_{0},\bv_{0},\bZ_{0}}\log\phi(\bY)+\lim_{N\to\infty}\frac{1}{N}\E_{\bu_{0},\bv_{0},\bZ_{0}}\log\varphi(\bZ_{0},\bY)\nonumber\\
&\overset{\text{(E)}}{=}\lim_{N\to\infty}\frac{1}{N}\E_{\bu_{0},\bv_{0},\bZ_{0}}\log\phi(\bu_{0},\bv_{0},\bZ_{0})+\log\varphi_{\alpha,\lambda}.\label{Eq:free_entropy}
\end{align}
In step (E), $\frac{1}{N}\E_{\bu_{0},\bv_{0},\bZ_{0}}\log\varphi(\bZ_{0},\bY)$ converges to $\log\varphi_{\alpha,\lambda}$, which depends only on $\alpha$ and $\lambda$ and is evaluated in Appendix \ref{Sec:C_y}. We also substitute $\bY=\frac{\lambda}{N}\bu_{0}\bv_{0}^{\top}+\bZ_{0}$ from \eqref{eq:model_def}, thereby making the dependence of $\phi(\cdot)$ on $\bu_{0}$, $\bv_{0}$, and $\bZ_{0}$ explicit.

It remains to evaluate $\phi(\bu_{0},\bv_{0},\bZ_{0})$. For this purpose, we define
\begin{align}
\tilde{\bZ}_{0}\triangleq\begin{bmatrix}\bm{0}&\bZ_{0}\\\bZ_{0}^{\top}&\bm{0}\end{bmatrix}\in\R^{(N+M)\times(N+M)},\quad\tilde{\bX}_{0}\triangleq\begin{bmatrix}\bu_{0}&\bm{0}\\\bm{0}&\bv_{0}\end{bmatrix}\in\R^{(N+M)\times2}.\label{Eq:Definition_tZ_tX}
\end{align}
The symmetrized form of model \eqref{eq:model_def} is
\begin{align}
\tilde{\bY}=\frac{\lambda}{N}\tilde{\bX}_{0}\tilde{\bI}\tilde{\bX}_{0}^{\top}+\tilde{\bZ}_{0}.
\end{align}
With this representation, $\phi(\bY)$ can be rewritten as
\begin{align}
\phi(\bY)&=\int{\dd P_{\su}(\bu_{a})\dd P_{\sv}(\bv_{a})}\,\exp\left(\frac{N}{2\pi\imagUnit}\oint_{\Gamma_{\sx}}{\dd x}\,xV^{\prime}(x^{2})\log\left|\tilde{\bI}+\frac{\lambda}{N}\tilde{\bX}_{a}^{\top}\cG_{\tilde{\bY}}(x)\tilde{\bX}_{a}\right|\right)\nonumber\\
&\overset{\text{(F)}}{=}\int{\cD\bB_{a}\dd P_{\su}(\bu_{a})\dd P_{\sv}(\bv_{a})}\,\exp\left(\frac{N}{2\pi\imagUnit}\oint_{\Gamma_{\sx}}{\dd x}\,xV^{\prime}(x^{2})\log|\tilde{\bI}+\lambda\bB_{a}(x)|\right)\delta\left(N\bB_{a}(x)-\tilde{\bX}_{a}^{\top}\cG_{\tilde{\bY}}(x)\tilde{\bX}_{a}\right)\nonumber\\
&\overset{\text{(G)}}{=}\int{\cD\bB_{a}\dd P_{\su}(\bu_{a})\dd P_{\sv}(\bv_{a})}\,\exp\left(\frac{N}{2\pi\imagUnit}\oint_{\Gamma_{\sx}}{\dd x}\,xV^{\prime}(x^{2})\log|\tilde{\bI}+\lambda\bB_{a}(x)|\right)\nonumber\\
&\quad\times\delta\left(N\bB_{a}(x)-\tilde{\bX}_{a}^{\top}\left[\cG_{\tilde{\bZ}_{0}}(x)+\frac{\lambda}{N}\cG_{\tilde{\bZ}_{0}}(x)\tilde{\bX}_{0}\left(\tilde{\bI}-\frac{\lambda}{N}\tilde{\bX}_{0}^{\top}\cG_{\tilde{\bZ}_{0}}(x)\tilde{\bX}_{0}\right)^{-1}\tilde{\bX}_{0}^{\top}\cG_{\tilde{\bZ}_{0}}(x)\right]\tilde{\bX}_{a}\right)\nonumber\\
&\overset{\text{(H)}}{=}\int{\cD[\bB_{a},\bG_{a},\bM_{a}]\dd P_{\su}(\bu_{a})\dd P_{\sv}(\bv_{a})}\,\exp\left(\frac{N}{2\pi\imagUnit}\oint_{\Gamma_{\sx}}{\dd x}\,xV^{\prime}(x^{2})\log|\tilde{\bI}+\lambda\bB_{a}(x)|\right)\nonumber\\
&\quad\times\delta\left(N\bB_{a}(x)-N\bG_{a}(x)-N\lambda\bM_{a}(x)(\tilde{\bI}-\lambda\bC(x))^{-1}\bM^{\top}_{a}(x)\right)\nonumber\\
&\quad\times\delta\left(N\bG_{a}(x)-\tilde{\bX}_{a}^{\top}\cG_{\tilde{\bZ}_{0}}(x)\tilde{\bX}_{a}\right)\delta\left(N\bM_{a}(x)-\tilde{\bX}_{a}^{\top}\cG_{\tilde{\bZ}_{0}}(x)\tilde{\bX}_{0}\right)\nonumber\\
&\overset{\text{(I)}}{=}\int{\cD[\bB_{a},\bG_{a},\bM_{a},\hat{\bB}_{a},\hat{\bG}_{a},\hat{\bM}_{a}]\dd P_{\su}(\bu_{a})\dd P_{\sv}(\bv_{a})}\,\exp\left(\frac{N}{2\pi\imagUnit}\oint_{\Gamma_{\sx}}{\dd x}\,xV^{\prime}(x^{2})\log|\tilde{\bI}+\lambda\bB_{a}(x)|\right)\nonumber\\
&\quad\times\exp\left(\oint_{\Gamma_{\sx}}{\dd x}\,\Tr\,\hat{\bB}_{a}(x)\left[N\bB_{a}(x)-N\bG_{a}(x)-N\lambda\bM_{a}(x)(\tilde{\bI}-\lambda\bC(x))^{-1}\bM^{\top}_{a}(x)\right]\right)\nonumber\\
&\quad\times\exp\left(\oint_{\Gamma_{\sx}}{\dd x}\,\Tr\,\hat{\bG}_{a}(x)\left(N\bG_{a}(x)-\tilde{\bX}_{a}^{\top}\cG_{\tilde{\bZ}_{0}}(x)\tilde{\bX}_{a}\right)\right)\exp\left(\oint_{\Gamma_{\sx}}{\dd x}\,\Tr\,\hat{\bM}_{a}^{\top}(x)\left(N\bM_{a}(x)-\tilde{\bX}_{a}^{\top}\cG_{\tilde{\bZ}_{0}}(x)\tilde{\bX}_{0}\right)\right).\label{Eq:Phi_Y}
\end{align}
The $2\times2$ matrix-valued functions introduced in steps (F) and (H) are
\begin{align}
\bB_{a}(x)\triangleq\frac{1}{N}\tilde{\bX}_{a}^{\top}\cG_{\tilde{\bY}}(x)\tilde{\bX}_{a},\quad\bG_{a}(x)\triangleq\frac{1}{N}\tilde{\bX}_{a}^{\top}\cG_{\tilde{\bZ}_{0}}(x)\tilde{\bX}_{a},\quad\bM_{a}(x)\triangleq\frac{1}{N}\tilde{\bX}_{a}^{\top}\cG_{\tilde{\bZ}_{0}}(x)\tilde{\bX}_{0},\quad\bC(x)\triangleq\frac{1}{N}\tilde{\bX}_{0}^{\top}\cG_{\tilde{\bZ}_{0}}(x)\tilde{\bX}_{0}.\label{Eq:Self_Overlap}
\end{align}
All but $\bM_{a}(x)$ are symmetric. Step (G) follows directly from the Woodbury matrix identity applied to the resolvent of $\tilde{\bY}$. Step (I) replaces the functional Dirac delta constraints by their Fourier representations and introduces the conjugate functions $\{\hat{\bB}_{a},\hat{\bG}_{a},\hat{\bM}_{a}\}$. The conjugate functions are likewise symmetric except for $\hat{\bM}_{a}(x)$. We use $\cD[\bB_{a},\bG_{a},\bM_{a}]$ and $\cD[\bB_{a},\bG_{a},\bM_{a},\hat{\bB}_{a},\hat{\bG}_{a},\hat{\bM}_{a}]$ as shorthand for the corresponding functional integration measures. This completes the preliminary reformulation of the free entropy.

\subsection{Derivation of Result \ref{Result:Free_Entropy} by the Replica Method}

\subsubsection{Replica Trick and Order Parameters}

We start the replica calculation from the functional integral representation of the free entropy \eqref{Eq:free_entropy} and \eqref{Eq:Phi_Y}. The expectation of the logarithm in \eqref{Eq:Phi_Y} is handled through the replica trick \cite{parisi1980sequence,mezard1987spin,mezard2009information,parisi2020theory}, which introduces an auxiliary parameter $n$, later interpreted as the number of replicas of the system, and rewrites the free entropy as
\begin{align}
f&=\lim_{N\to\infty}\frac{1}{N}\E_{\bu_{0},\bv_{0},\bZ_{0}}\log\phi(\bu_{0},\bv_{0},\bZ_{0})+\log\varphi_{\alpha,\lambda}\nonumber\\
&=\lim_{N\to\infty}\frac{1}{N}\lim_{n\to0}\frac{1}{n}\log\E_{\bu_{0},\bv_{0},\bZ_{0}}\phi^{n}(\bu_{0},\bv_{0},\bZ_{0})+\log\varphi_{\alpha,\lambda}\nonumber\\
&\overset{\text{(A)}}{=}\lim_{n\to0}\frac{1}{n}\underbrace{\lim_{N\to\infty}\frac{1}{N}\log\E_{\bu_{0},\bv_{0},\bZ_{0}}\phi^{n}(\bu_{0},\bv_{0},\bZ_{0})}_{\triangleq f_{n}}+\log\varphi_{\alpha,\lambda}.\label{Eq:f_n_Variational}
\end{align}
Here, (A) assumes the commutation of the thermodynamic and replica limits.

We next evaluate the replicated functional $f_{n}$. Although the replica limit in \eqref{Eq:f_n_Variational} is taken as a real limit $n\to0$, the calculation is first carried out for positive integer $n$, for which $\phi^{n}(\cdot)$ can be represented as a product over $n$ replicated candidate signal pairs:
\begin{align}
f_{n}&=\lim_{N\to\infty}\frac{1}{N}\log\E_{\bu_{0},\bv_{0},\bZ_{0}}\prod_{a=1}^{n}\phi(\bu_{0},\bv_{0},\bZ_{0}).\label{Eq:f_n_Replicated}
\end{align}
As before, the index $a=0$ is reserved for the ground-truth signal pair $\{\bu_{0},\bv_{0}\}$, while $a=1,\ldots,n$ index the replicated candidate signal pairs. The expression obtained for integer $n$ is then assumed to admit an analytic continuation to real $n$ in a neighborhood of zero, where the replica limit $n\to0$ is taken.

For the remainder, we employ the RS ansatz \cite{parisi1980sequence,mezard1987spin,mezard2009information,takahashi2022macroscopic}: all replicated order parameters are assumed to be invariant under arbitrary permutations of the candidate indices $a\in\{1,\cdots,n\}$. This applies both to the matrix-valued functions introduced above and to the scalar overlap parameters introduced below through the replica matrices. It substantially reduces the number of order parameters.

For Bayes-optimal inference, replica symmetry was established in \cite{barbier_strong_CMP} in the strong sense of concentration of all multioverlaps. The result holds for generic likelihoods, while the hypotheses concern only the prior, which must factorize over the signal components, and the concentration of the free entropy. Our likelihood comes from a RI noise ensemble, and is therefore covered. Our priors $P_{\su}(\cdot)$ and $P_{\sv}(\cdot)$ are factorized by construction. Only the concentration of the free entropy remains to be checked. This suggests that the present model is replica symmetric, and we adopt the RS ansatz on that basis. Under it,
\begin{align}
\{\bB_{a},\bG_{a},\bM_{a},\hat{\bB}_{a},\hat{\bG}_{a},\hat{\bM}_{a}\}=\{\bB,\bG,\bM,\hat{\bB},\hat{\bG},\hat{\bM}\},\quad a\in\{1,\cdots,n\}.
\end{align}
Without loss of generality, we rescale the conjugate functions so that all contour integrals carry the common prefactor $\frac{1}{2\pi\imagUnit}$:
\begin{align}
\{2\pi\imagUnit\hat{\bB},2\pi\imagUnit\hat{\bG},2\pi\imagUnit\hat{\bM}\}\to\{\hat{\bB},\hat{\bG},\hat{\bM}\}.
\end{align}

Under the RS ansatz, the replicated functional simplifies to
\begin{align}
f_{n}&=\lim_{N\to\infty}\frac{1}{N}\log\int{\cD[\bB,\bG,\bM,\hat{\bB},\hat{\bG},\hat{\bM}]\prod_{a=0}^{n}\dd P_{\su}(\bu_{a})\prod_{a=0}^{n}\dd P_{\sv}(\bv_{a})}\,\exp(Nn\tilde{g})\nonumber\\
&\quad\times\E_{\bZ_{0}}\exp\left(-\sum_{a=1}^{n}\frac{1}{2\pi\imagUnit}\oint_{\Gamma_{\sx}}{\dd x}\,\Tr\,\left(\hat{\bG}(x)\tilde{\bX}_{a}^{\top}\cG_{\tilde{\bZ}_{0}}(x)\tilde{\bX}_{a}+\hat{\bM}^{\top}(x)\tilde{\bX}_{a}^{\top}\cG_{\tilde{\bZ}_{0}}(x)\tilde{\bX}_{0}\right)\right),\label{eq:fn_rep}
\end{align}
where $\cD[\bB,\bG,\bM,\hat{\bB},\hat{\bG},\hat{\bM}]$ denotes the corresponding functional integration measure, and
\begin{align}
\tilde{g}(\bB,\bG,\bM,\hat{\bB},\hat{\bG},\hat{\bM})&\triangleq\frac{1}{2\pi\imagUnit}\oint_{\Gamma_{\sx}}{\dd x}\,\Bigg\{xV'(x^{2})\log|\tilde{\bI}+\lambda\bB(x)|+\Tr\,\hat{\bG}(x)\bG(x)+\Tr\,\hat{\bM}^{\top}(x)\bM(x)\nonumber\\
&\quad+\Tr\,\hat{\bB}(x)\left[\bB(x)-\bG(x)-\lambda\bM(x)(\tilde{\bI}-\lambda\bC(x))^{-1}\bM^{\top}(x)\right]\Bigg\},\label{Eq:bar_g}
\end{align}
For notational simplicity, the arguments of all such functions are displayed only in their definitions and omitted elsewhere.

We will show that the noise average $\E_{\bZ_{0}}\exp(\cdot)$ in \eqref{eq:fn_rep} depends on the replicas $\{\bu_{a}\}_{a=0}^{n}$ and $\{\bv_{a}\}_{a=0}^{n}$ only through their pairwise inner products. This motivates the introduction of the overlap matrices, commonly called the replica matrices:
\begin{align}
\bQ_{\su}\triangleq\frac{1}{N}\bU^{\top}\bU\in\R^{(n+1)\times(n+1)},\quad\bQ_{\sv}\triangleq\frac{1}{M}\bV^{\top}\bV\in\R^{(n+1)\times(n+1)},\label{Eq:RS_Matrices}
\end{align}
where $\bU\triangleq[\bu_{0},\ldots,\bu_{n}]\in\R^{N\times(n+1)}$ and $\bV\triangleq[\bv_{0},\ldots,\bv_{n}]\in\R^{M\times(n+1)}$. Enforcing these definitions with Dirac delta constraints gives
\begin{align}
f_{n}&=\lim_{N\to\infty}\frac{1}{N}\log\int{\cD[\bB,\bG,\bM,\hat{\bB},\hat{\bG},\hat{\bM}]\dd\bQ_{\su}\dd\bQ_{\sv}}\,\exp(Nn\tilde{g})\exp(Ng_{\sz,N,n})\nonumber\\
&\quad\times\left(\int{\prod_{a=0}^{n}\dd P_{\su}(\bu_{a})}\,\delta(N\bQ_{\su}-\bU^{\top}\bU)\right)\left(\int{\prod_{a=0}^{n}\dd P_{\sv}(\bv_{a})}\,\delta(M\bQ_{\sv}-\bV^{\top}\bV)\right)\nonumber\\
&=\lim_{N\to\infty}\frac{1}{N}\log\int{\cD[\bB,\bG,\bM,\hat{\bB},\hat{\bG},\hat{\bM}]\dd\bQ_{\su}\dd\bQ_{\sv}}\,\exp\left(N(n\tilde{g}+g_{\sz,N,n}+g_{\su,N,n}+\alpha g_{\sv,N,n})\right),\label{eq:fn_evaluate}
\end{align}
where
\begin{align}
g_{\su,N,n}(\bQ_{\su})&\triangleq\frac{1}{N}\log\int{\prod_{a=0}^{n}\dd P_{\su}(\bu_{a})}\,\prod_{0\leq a\leq a'\leq n}\delta(NQ_{\su,a,a'}-\bu_{a}^{\top}\bu_{a'}),\label{Eq:g_u_def}\\
g_{\sv,N,n}(\bQ_{\sv})&\triangleq \frac{1}{M}\log\int{\prod_{a=0}^{n}\dd P_{\sv}(\bv_{a})}\,\prod_{0\leq a\leq a'\leq n}\delta(MQ_{\sv,a,a'}-\bv_{a}^{\top}\bv_{a'}),\label{Eq:g_v_def}\\
g_{\sz,N,n}(\hat{\bG},\hat{\bM},\bQ_{\su},\bQ_{\sv})&\triangleq\frac{1}{N}\log\E_{\bZ_{0}}\exp\left(-\sum_{a=1}^{n}\frac{1}{2\pi\imagUnit}\oint_{\Gamma_{\sx}}{\dd x}\,\Tr\,\left(\hat{\bG}(x)\tilde{\bX}_{a}^{\top}\cG_{\tilde{\bZ}_{0}}(x)\tilde{\bX}_{a}+\hat{\bM}^{\top}(x)\tilde{\bX}_{a}^{\top}\cG_{\tilde{\bZ}_{0}}(x)\tilde{\bX}_{0}\right)\right).\label{Eq:g_z_{n}}
\end{align}
In \eqref{eq:fn_evaluate} one can replace all $N$-dependent functions $g_{\bullet,N,n}$ entering the exponent by their thermodynamic limit $g_{\bullet,n}\triangleq\lim_{N\to\infty}g_{\bullet,N,n}$ without changing the expression at leading order. The resulting expression possesses the right form to evaluate the remaining integrals over the order parameters by the saddle point method, which yields
\begin{align}
f_{n}&=\Extr(n\tilde{g}+g_{\sz,n}+g_{\su,n}+\alpha g_{\sv,n}),\label{Eq:f_n}
\end{align}
where the extremization is taken over all order parameters $\mathcal{M}=\{\bB,\bG,\bM,\hat{\bB},\hat{\bG},\hat{\bM},\bQ_{\su},\bQ_{\sv}\}$. The saddle points satisfy $\nabla_{\mathcal{M}}(n\tilde{g}+g_{\sz,n}+g_{\su,n}+\alpha g_{\sv,n})=\bm{0}$, and $f_{n}$ is given by the value at the dominant saddle point.

Under the RS ansatz, the replica matrices $\bQ_{\su}$ and $\bQ_{\sv}$ are replaced by $\cQ(\beta^{(0)}_{\su},m_{\su},q_{\su},\beta_{\su})$ and $\cQ(\beta^{(0)}_{\sv},m_{\sv},q_{\sv},\beta_{\sv})$, respectively, where
\begin{align}
\cQ(\beta^{(0)},m,q,\beta)&\triangleq\left[\begin{array}{c|c}\beta^{(0)}&m\vec{\bm{1}}\\\hline m\bm{1}&(\beta-q)\Eye+q\bm{1}\vec{\bm{1}}\end{array}\right].\label{Eq:RS_Matrix}
\end{align}
Here, $\beta^{(0)}$ and $\beta$ denote the self-overlaps of the ground-truth signal and a replica, respectively, while $m$ and $q$ denote the overlap between the ground-truth signal and a replica and that between two distinct replicas, respectively. Since the signal priors have unit second moments, $\E[u_{0}^{2}]=\E[v_{0}^{2}]=1$, the ground-truth self-overlaps have unit limiting values. We retain $\beta_{\su}^{(0)}$ and $\beta_{\sv}^{(0)}$ as auxiliary scalar order parameters until the corresponding conjugate parameters are determined, as detailed in Appendix \ref{Sec:Completion}. The extremization is therefore w.r.t. the scalar $\{\beta_{\su}^{(0)},m_{\su},q_{\su},\beta_{\su},\beta_{\sv}^{(0)},m_{\sv},q_{\sv},\beta_{\sv}\}$ and functional $\{\bB,\bG,\bM,\hat{\bB},\hat{\bG},\hat{\bM}\}$ order parameters.

\subsubsection{Simplifying the Terms Entering $f_{n}$}

Starting with the simplest one, the term $\tilde{g}$ is independent of the dimension $N$ and replica number $n$, and thus cannot be further simplified. All the remaining dependencies are contained in the other terms.

The simplest terms requiring further simplification are the prior-related terms $g_{\su,n}$ and $g_{\sv,n}$. Such terms appear in all replica calculations for Bayes-optimal inference problems with a dense underlying graphical structure and are thus standard to compute; for completeness, we provide the derivation in Appendix \ref{Sec:Proof_RM_Prior}. At leading order in $n$, as required for the replica limit $n\to 0$, this yields
\begin{align}
g_{\su,n}&=n\tilde{g}_{\su}+O(n^{2}),\label{Eq:G_SU_N}\\
\tilde{g}_{\su}&=\Extr_{\hat{m}_{\su},\hat{q}_{\su},\hat{\beta}_{\su}}\left(-m_{\su}\hat{m}_{\su}+\frac{1}{2}q_{\su}\hat{q}_{\su}+\frac{1}{2}\beta_{\su}\hat{\beta}_{\su}+\E_{u_{0},z_{\su}}\log\E_{u_{a}}\exp\left(\hat{m}_{\su}u_{0}u_{a}-\frac{\hat{\beta}_{\su}+\hat{q}_{\su}}{2}u_{a}^{2}+\sqrt{\hat{q}_{\su}}u_{a}z_{\su}\right)\right),\label{Eq:TG_SU}
\end{align}
with $z_{\su}\sim\Normal(0,1)$ and $u_{0},u_{a}\overset{\text{i.i.d.}}{\sim}P_{\su}(\cdot)$. $g_{\sv,n}$ is given by the same expression with the replacements $\su\to\sv$ and $u_{0},u_{a}\to v_{0},v_{a}\overset{\text{i.i.d.}}{\sim}P_{\sv}(\cdot)$.

The noise-related term $g_{\sz,n}$ is the one requiring special care. We recognize it as a special case of what we call a \textit{rectangular inhomogeneous spherical integral}. Unlike the so-called inhomogeneous spherical integral associated with uniform single-channel models defined and analyzed in \cite{barbier2023fundamental,barbier2025information}, the present formulation is rectangular and contains asymmetric couplings between the $\su$-channel and the $\sv$-channel. This requires the generalized framework developed in Result \ref{Lemma:HCIZ} below.

\begin{result}[Rectangular Inhomogeneous Spherical Integral]\label{Lemma:HCIZ}

Let $\bP\sim\Sf{Unif}(\mathbb{O}(N))$ and $\bQ\sim\Sf{Unif}(\mathbb{O}(M))$ be independent Haar orthogonal matrices. The indices $a,a'\leq n$ refer to the replica dimension, whereas indices $i$ and $j$ refer to the spin dimension associated with the matrix dimensions $N$ and $M$. Let $\bX\in\R^{N\times n}$ and $\bY\in\R^{M\times n}$ collect the replicas in the $\sx$- and $\sy$-channels, respectively. Their replica matrices are defined as
\begin{align}
\bQ_{\sx}=\frac{1}{N}\bX^{\top}\bX,\quad\bQ_{\sy}=\frac{1}{M}\bY^{\top}\bY\in\R^{n\times n}.
\end{align}
Let $\bA\in\R^{N\times n\times n}$, $\bB\in\R^{M\times n\times n}$, and $\bC\in\R^{N\times n\times n}$ denote the replica coupling arrays, with entries $A_{i,a,a'}$, $B_{i,a,a'}$, and $C_{i,a,a'}$, respectively. For each pair of replica indices $(a,a')$, define the diagonal matrices
\begin{align}
\bA_{a,a'}&\triangleq\Diag([A_{i,a,a'}]_{i\leq N})\in\R^{N\times N},\quad\bB_{a,a'}\triangleq\Diag([B_{i,a,a'}]_{i\leq M})\in\R^{M\times M},\quad\bC_{a,a'}\triangleq\Diag([C_{i,a,a'}]_{i\leq N})\in\R^{N\times N}.\label{Aaa_Baa_Caa}
\end{align}
For each spin index, define the matrices
\begin{align}
\bA_{i}&\triangleq[A_{i,a,a'}]_{a,a'\leq n}\in\R^{n\times n},\quad\bB_{i}\triangleq[B_{i,a,a'}]_{a,a'\leq n}\in\R^{n\times n},\quad\bC_{i}\triangleq[C_{i,a,a'}]_{a,a'\leq n}\in\R^{n\times n}.\label{Eq:Ai_Bi_Ci}
\end{align}

The rectangular inhomogeneous spherical integral is defined as
\begin{align}
g_{\sh,n}&\triangleq\lim_{N\to\infty}\frac{1}{N}\log\E_{\bP,\bQ}\exp\left(\sum_{1\leq a,a'\leq n}\bx_{a}^{\top}\bP^{\top}\bA_{a,a'}\bP\bx_{a'}+\by_{a}^{\top}\bQ^{\top}\bB_{a,a'}\bQ\by_{a'}+\bx_{a}^{\top}\bP^{\top}[\bC_{a,a'},\bm{0}]\bQ\by_{a'}\right).\label{Eq:HCIZ_Definition}
\end{align}
In the thermodynamic limit, this integral admits the following variational representation:
\begin{align}
g_{\sh,n}&=\Extr_{\tilde{\bQ}_{\sx},\tilde{\bQ}_{\sy}}\Bigg(\frac{1}{2}\Tr\,\bQ_{\sx}\tilde{\bQ}_{\sx}+\frac{\alpha}{2}\Tr\,\bQ_{\sy}\tilde{\bQ}_{\sy}-\frac{1}{2N}\sum_{i=1}^{N}\log|\tilde{\bQ}_{\sx}-2\bA_{i}|-\frac{1}{2N}\sum_{i=N+1}^{M}\log|\tilde{\bQ}_{\sy}-2\bB_{i}|\nonumber\\
&\quad-\frac{1}{2N}\sum_{i=1}^{N}\log|(\tilde{\bQ}_{\sy}-2\bB_{i})-\bC_{i}^{\top}(\tilde{\bQ}_{\sx}-2\bA_{i})^{-1}\bC_{i}|\Bigg)-\frac{1}{2}\log|\bQ_{\sx}|-\frac{\alpha}{2}\log|\bQ_{\sy}|-\frac{n(1+\alpha)}{2},\label{Eq:Gen_Spherical_Int_Result_Lemma}
\end{align}
where $\tilde{\bQ}_{\sx},\tilde{\bQ}_{\sy}\in\R^{n\times n}$ are the conjugate replica matrices associated with $\bQ_{\sx}$ and $\bQ_{\sy}$, respectively. 

A derivation of this result is provided in Appendix \ref{app:spherical_int}.

\end{result}

Result \ref{Lemma:HCIZ} is stated for replica indices $a\in\{1,\ldots,n\}$. In $g_{\sz,n}$, the ground-truth signal pair $\{\bu_{0},\bv_{0}\}$ is assigned the additional index $a=0$. Hence, the replica index set becomes $\{0,1,\ldots,n\}$, and the replica matrices are enlarged from $n\times n$ to $(n+1)\times(n+1)$. Accordingly, the constant $-n(1+\alpha)/2$ in \eqref{Eq:Gen_Spherical_Int_Result_Lemma} becomes $-(n+1)(1+\alpha)/2$.

We now rewrite $g_{\sz,n}$ in the form of \eqref{Eq:HCIZ_Definition} and identify the corresponding matrices $\bA_{a,a'}$, $\bB_{a,a'}$, and $\bC_{a,a'}$. We first write the $2\times2$ conjugate matrix-valued functions $\hat{\bG}(x)$ and $\hat{\bM}(x)$ entrywise and use the block representation of $\cG_{\tilde{\bZ}_{0}}(x)$ in \eqref{Eq:Block_Structure}. Here, $\hat{\bG}(x)$ is symmetric, whereas $\hat{\bM}(x)$ need not be. For the replica calculation, we fix a realization of $\bs$ and assume that averaging over the singular values does not change the asymptotic free entropy \cite{barbier2023fundamental,barbier2025information}. Using the SVD $\bZ_{0}=\bP^{\top}[\Diag(\bs),\bm{0}]\bQ$, the remaining expectation is taken over the independent Haar orthogonal matrices $\bP$ and $\bQ$. The contour integrals are then collected into the functions $\{\tilde{M}_{\Sf{uu}},\tilde{M}_{\Sf{uv}},\tilde{M}_{\Sf{vu}},\tilde{M}_{\Sf{vv}},\tilde{G}_{\Sf{uu}},\tilde{G}_{\Sf{uv}},\tilde{G}_{\Sf{vv}}\}$ defined in \eqref{Eq:tM_uu}-\eqref{Eq:tG_vv}, yielding \eqref{Eq:g_Z_n}. Comparing this expression with \eqref{Eq:HCIZ_Definition} gives the following nonzero diagonal matrices:
\begin{align}
\bA_{a,a}(\bs)&=-\Diag(\tilde{G}_{\Sf{uu}}(\bs)),\quad\bA_{0,a}(\bs)=\bA_{a,0}(\bs)=-\frac{1}{2}\Diag(\tilde{M}_{\Sf{uu}}(\bs)),\label{Eq:A_aa}\\
\bB_{a,a}(\bs)&=-\begin{bmatrix}\Diag(\tilde{G}_{\Sf{vv}}(\bs))&\bm{0}\\\bm{0}&\tilde{G}_{\Sf{vv}}(0)\Eye\end{bmatrix},\quad\bB_{0,a}(\bs)=\bB_{a,0}(\bs)=-\frac{1}{2}\begin{bmatrix}\Diag(\tilde{M}_{\Sf{vv}}(\bs))&\bm{0}\\\bm{0}&\tilde{M}_{\Sf{vv}}(0)\Eye\end{bmatrix},\label{Eq:B_aa}\\
\bC_{a,a}(\bs)&=-2\Diag(\tilde{G}_{\Sf{uv}}(\bs)),\quad\bC_{0,a}(\bs)=-\Diag(\tilde{M}_{\Sf{vu}}(\bs)),\quad\bC_{a,0}(\bs)=-\Diag(\tilde{M}_{\Sf{uv}}(\bs)),\label{Eq:C_aa}
\end{align}
where $a\in\{1,\ldots,n\}$, and all other matrices are zero. The corresponding matrices for each spin index are
\begin{align}
\bA(s_{i})&=\cC\left(-\frac{1}{2}\tilde{M}_{\Sf{uu}}(s_{i}),-\frac{1}{2}\tilde{M}_{\Sf{uu}}(s_{i}),-\tilde{G}_{\Sf{uu}}(s_{i})\right),\quad i\in\{1,\cdots,N\},\label{Eq:A_i}\\
\bB(s_{i})&=\cC\left(-\frac{1}{2}\tilde{M}_{\Sf{vv}}(s_{i}),-\frac{1}{2}\tilde{M}_{\Sf{vv}}(s_{i}),-\tilde{G}_{\Sf{vv}}(s_{i})\right),\quad i\in\{1,\cdots,M\},\label{Eq:B_i}\\
\bC(s_{i})&=\cC\left(-\tilde{M}_{\Sf{vu}}(s_{i}),-\tilde{M}_{\Sf{uv}}(s_{i}),-2\tilde{G}_{\Sf{uv}}(s_{i})\right),\quad i\in\{1,\cdots,N\},\label{Eq:C_i}
\end{align}
where $\cC$ is defined in \eqref{Eq:C_Matrix}, and $s_{i}=0$ for $i\in\{N+1,\cdots,M\}$ in $\bB(s_{i})$.

With these identifications, Result \ref{Lemma:HCIZ} yields the variational representation in \eqref{Eq:tI_HCIZ}. In the thermodynamic limit, the sums over the $N$ singular-value terms converge to expectations over $s^{2}\sim\mu(\cdot)$, while the remaining $M-N$ zero singular values in the $\sv$-channel give the contribution proportional to $\alpha-1$, as shown in \eqref{Eq:gz_n}.

Under the RS ansatz, the conjugate replica matrices are parameterized as
\begin{align}
\tilde{\bQ}_{\su}=\cQ(\tilde{\beta}^{(0)}_{\su},\tilde{m}_{\su},\tilde{q}_{\su},\tilde{\beta}_{\su}),\quad\tilde{\bQ}_{\sv}=\cQ(\tilde{\beta}^{(0)}_{\sv},\tilde{m}_{\sv},\tilde{q}_{\sv},\tilde{\beta}_{\sv}).\label{Eq:tQ_u_tQ_v}
\end{align}
The variational problem is thus reduced to an extremization over the scalar RS parameters. Expanding around $n=0$ gives
\begin{align}
g_{\sz,n}&=n\tilde{g}_{\sz}+O(n^{2}),\label{Eq:g_sz_n}
\end{align}
where
\begin{align}
\tilde{g}_{\sz}&=\Extr_{\tilde{m}_{\su},\tilde{q}_{\su},\tilde{\beta}_{\su},\tilde{m}_{\sv},\tilde{q}_{\sv},\tilde{\beta}_{\sv}}\Bigg(\frac{1}{2}(2m_{\su}\tilde{m}_{\su}-q_{\su}\tilde{q}_{\su}+\beta_{\su}\tilde{\beta}_{\su})+\frac{\alpha}{2}(2m_{\sv}\tilde{m}_{\sv}-q_{\sv}\tilde{q}_{\sv}+\beta_{\sv}\tilde{\beta}_{\sv})-\frac{\alpha-1}{2}\left(\log\tilde{\cG}_{\sv}(0)+\frac{\tilde{q}_{\sv}-\tilde{\cM}_{\sv}^{2}(0)}{\tilde{\cG}_{\sv}(0)}\right)\nonumber\\
&\quad-\frac{1}{2}\E_{s^{2}\sim\mu}\left[\log\tilde{\cG}_{\su}(s)+\frac{\tilde{q}_{\su}-\tilde{\cM}_{\su}^{2}(s)}{\tilde{\cG}_{\su}(s)}+\frac{\bar{\cN}(s)}{\bar{\cD}(s)}-\frac{\tilde{M}_{\Sf{uv}}^{2}(s)}{\tilde{\cG}_{\su}(s)}+\log\bar{\cD}(s)\right]\Bigg)\nonumber\\
&\quad-\frac{1}{2}\left(\log(\beta_{\su}-q_{\su})+\frac{q_{\su}-m_{\su}^{2}}{\beta_{\su}-q_{\su}}\right)-\frac{\alpha}{2}\left(\log(\beta_{\sv}-q_{\sv})+\frac{q_{\sv}-m_{\sv}^{2}}{\beta_{\sv}-q_{\sv}}\right)-\frac{1+\alpha}{2}.\label{Eq:tb_u_0_tb_v_0}
\end{align}
Here, the functions $\{\bar{\cN},\bar{\cD},\tilde{\cM}_{\su},\tilde{\cM}_{\sv},\tilde{\cG}_{\su},\tilde{\cG}_{\sv}\}$ are defined in \eqref{Eq:b_N}-\eqref{Eq:t_CG_v}. Details of the derivation are given in Appendix \ref{Sec:Proof_RM_Noise}.\\

%
%
%

\subsubsection{Final Steps}

Using \eqref{Eq:f_n_Variational}, we have the free entropy
\begin{align}
f&=\lim_{n\to0}\frac{1}{n}f_{n}+\log\varphi_{\alpha,\lambda}\nonumber\\
&\overset{\text{(A)}}{=}\lim_{n\to0}\frac{1}{n}\Extr\left(n\tilde{g}+g_{\su,n}+\alpha g_{\sv,n}+g_{\sz,n}\right)+\log\varphi_{\alpha,\lambda}\nonumber\\
&\overset{\text{(B)}}{=}\log\varphi_{\alpha,\lambda}+\Extr\left(\tilde{g}+\tilde{g}_{\su}+\alpha\tilde{g}_{\sv}+\tilde{g}_{\sz}\right)\nonumber\\
&=\log\varphi_{\alpha,\lambda}+\Extr\Bigg(\frac{1}{2\pi\imagUnit}\oint_{\Gamma_{\sx}}{\dd x}\,\Bigg\{xV'(x^{2})\log|\tilde{\bI}+\lambda\bB(x)|+\Tr\,\hat{\bG}(x)\bG(x)+\Tr\,\hat{\bM}^{\top}(x)\bM(x)\nonumber\\
&\quad+\Tr\,\hat{\bB}(x)\left[\bB(x)-\bG(x)-\lambda\bM(x)(\tilde{\bI}-\lambda\bC(x))^{-1}\bM^{\top}(x)\right]\Bigg\}\nonumber\\
&\quad+\E_{u_{0},z_{\su}}\log\E_{u_{a}}\exp\left(\hat{m}_{\su}u_{0}u_{a}-\frac{\hat{\beta}_{\su}+\hat{q}_{\su}}{2}u_{a}^{2}+\sqrt{\hat{q}_{\su}}u_{a}z_{\su}\right)+\alpha\E_{v_{0},z_{\sv}}\log\E_{v_{a}}\exp\left(\hat{m}_{\sv}v_{0}v_{a}-\frac{\hat{\beta}_{\sv}+\hat{q}_{\sv}}{2}v_{a}^{2}+\sqrt{\hat{q}_{\sv}}v_{a}z_{\sv}\right)\nonumber\\
&\quad-\frac{\alpha-1}{2}\left(\log\tilde{\cG}_{\sv}(0)+\frac{\tilde{q}_{\sv}-\tilde{\cM}_{\sv}^{2}(0)}{\tilde{\cG}_{\sv}(0)}\right)-\frac{1}{2}\E_{s^{2}\sim\mu}\left[\log\tilde{\cG}_{\su}(s)+\frac{\tilde{q}_{\su}-\tilde{\cM}_{\su}^{2}(s)}{\tilde{\cG}_{\su}(s)}+\frac{\bar{\cN}(s)}{\bar{\cD}(s)}-\frac{\tilde{M}_{\Sf{uv}}^{2}(s)}{\tilde{\cG}_{\su}(s)}+\log\bar{\cD}(s)\right]\nonumber\\
&\quad+\frac{1}{2}\left(-2m_{\su}\hat{m}_{\su}+q_{\su}\hat{q}_{\su}+\beta_{\su}\hat{\beta}_{\su}\right)+\frac{\alpha}{2}\left(-2m_{\sv}\hat{m}_{\sv}+q_{\sv}\hat{q}_{\sv}+\beta_{\sv}\hat{\beta}_{\sv}\right)+\frac{1}{2}(2m_{\su}\tilde{m}_{\su}-q_{\su}\tilde{q}_{\su}+\beta_{\su}\tilde{\beta}_{\su})\nonumber\\
&\quad+\frac{\alpha}{2}(2m_{\sv}\tilde{m}_{\sv}-q_{\sv}\tilde{q}_{\sv}+\beta_{\sv}\tilde{\beta}_{\sv})-\frac{1}{2}\left(\log(\beta_{\su}-q_{\su})+\frac{q_{\su}-m_{\su}^{2}}{\beta_{\su}-q_{\su}}\right)-\frac{\alpha}{2}\left(\log(\beta_{\sv}-q_{\sv})+\frac{q_{\sv}-m_{\sv}^{2}}{\beta_{\sv}-q_{\sv}}\right)-\frac{1+\alpha}{2}\Bigg) .\label{Eq:Intermediate_Free_Entropy}
\end{align}
Here, (A) follows from \eqref{Eq:f_n}, while (B) follows from the first-order expansions obtained above. In \eqref{Eq:Intermediate_Free_Entropy}, $\Extr(\cdot)$ denotes extremization w.r.t. the scalar parameters $\{m_{\bullet},q_{\bullet},\beta_{\bullet},\hat{m}_{\bullet},\hat{q}_{\bullet},\hat{\beta}_{\bullet},\tilde{m}_{\bullet},\tilde{q}_{\bullet},\tilde{\beta}_{\bullet}\}_{\bullet\in\{\su,\sv\}}$ and the functions $\{\bB,\bG,\bM,\hat{\bB},\hat{\bG},\hat{\bM}\}$. For completeness, the remaining details of the derivation are provided in Appendix \ref{Sec:Completion}.

Stationarity w.r.t. the functions $\{\bB,\bG,\bM,\hat{\bM}\}$ and all scalar parameters gives the saddle point equations in \eqref{Eq:bB}-\eqref{Eq:hM_vv}. At this stage, some scalar parameters that become equal under the Nishimori identities are still written separately. Using these equalities reduces the saddle point equations to the form stated in Result \ref{Result:Saddle_Point} and simplifies \eqref{Eq:Intermediate_Free_Entropy} to the free entropy in Result \ref{Result:Free_Entropy}. The simplification of the free entropy is completed in Section \ref{Sec:Final_Free_Entropy}.

\subsection{Derivation of Result \ref{Result:Saddle_Point}}

\subsubsection{Saddle Point Equations}

Setting the derivatives of the free entropy in \eqref{Eq:Intermediate_Free_Entropy} to zero gives the full saddle point system. To characterize the asymptotic Bayes-optimal performance, as measured by the MMSE, we only need the following subset of equations \cite{shinzato2008learning,shinzato2008perceptron,takahashi2022macroscopic,barbier2023fundamental,barbier2025information}. The labels in parentheses indicate the corresponding stationarity conditions:
\begin{align}
(\bB)&:\quad\bm{0}=\lambda xV'(x^{2})(\tilde{\bI}+\lambda\bB(x))^{-1}+\hat{\bB}(x),\label{Eq:bB}\\
(\bM)&:\quad\bm{0}=\hat{\bM}(x)-2\lambda\hat{\bB}(x)\bM(x)(\tilde{\bI}-\lambda\bC(x))^{-1},\label{Eq:bM}\\
(\bG)&:\quad\bm{0}=\hat{\bG}(x)-\hat{\bB}(x),\label{Eq:bG}\\
(\hat{m}_{\su})&:\quad m_{\su}=\E_{u_{0},z_{\su}}[u_{0}\langle u_{a}\rangle],\label{Eq:msu}\\
(\hat{q}_{\su})&:\quad q_{\su}\overset{\text{(A)}}{=}\E_{u_{0},z_{\su}}[\langle u_{a}\rangle^{2}],\label{Eq:qsu}\\
(\hat{\beta}_{\su})&:\quad\beta_{\su}=\E_{u_{0},z_{\su}}[\langle u^{2}_{a}\rangle],\label{Eq:bsu}\\
(\hat{m}_{\sv})&:\quad m_{\sv}=\E_{v_{0},z_{\sv}}[v_{0}\langle v_{a}\rangle],\label{Eq:msv}\\
(\hat{q}_{\sv})&:\quad q_{\sv}\overset{\text{(A)}}{=}\E_{v_{0},z_{\sv}}[\langle v_{a}\rangle^{2}],\label{Eq:qsv}\\
(\hat{\beta}_{\sv})&:\quad\beta_{\sv}=\E_{v_{0},z_{\sv}}[\langle v^{2}_{a}\rangle],\label{Eq:bsv}\\
(q_{\su})+(\beta_{\su})&:\quad\tilde{\beta}_{\su}-\tilde{q}_{\su}=\frac{1}{\beta_{\su}-q_{\su}}-\hat{\beta}_{\su}-\hat{q}_{\su},\label{Eq:qsu_bsu}\\
(q_{\sv})+(\beta_{\sv})&:\quad\tilde{\beta}_{\sv}-\tilde{q}_{\sv}=\frac{1}{\beta_{\sv}-q_{\sv}}-\hat{\beta}_{\sv}-\hat{q}_{\sv},\label{Eq:qsv_bsv}\\
(\tilde{q}_{\su})+(\tilde{\beta}_{\su})&:\quad\beta_{\su}-q_{\su}=\E_{s^{2}\sim\mu}\frac{1}{\tilde{\cG}_{\su}(s)-\frac{4\tilde{G}_{\Sf{uv}}^{2}(s)}{\tilde{\cG}_{\sv}(s)}},\label{Eq:tqsu_tbsu}\\
(\tilde{q}_{\sv})+(\tilde{\beta}_{\sv})&:\quad\beta_{\sv}-q_{\sv}=\frac{1}{\alpha}\E_{s^{2}\sim\mu}\frac{1}{\tilde{\cG}_{\sv}(s)-\frac{4\tilde{G}_{\Sf{uv}}^{2}(s)}{\tilde{\cG}_{\su}(s)}}+\left(1-\frac{1}{\alpha}\right)\frac{1}{\tilde{\cG}_{\sv}(0)},\label{Eq:tqsv_tbsv}\\
(\hat{M}_{\Sf{uu}})&:\quad M_{\Sf{uu}}(x)=\E_{s^{2}\sim\mu}\left[-\frac{1}{\tilde{\cG}_{\su}(s)}\left[\tilde{\cM}_{\su}(s)-\frac{2\tilde{G}_{\Sf{uv}}(s)}{\bar{\cD}(s)}\left(\tilde{M}_{\Sf{vu}}(s)-\frac{2\tilde{\cM}_{\su}(s)\tilde{G}_{\Sf{uv}}(s)}{\tilde{\cG}_{\su}(s)}\right)\right]\frac{x}{x^{2}-s^{2}}\right],\label{Eq:hM_uu}\\
(\hat{M}_{\Sf{uv}})&:\quad M_{\Sf{uv}}(x)=\E_{s^{2}\sim\mu}\left[-\frac{1}{\tilde{\cG}_{\su}(s)}\left[\tilde{M}_{\Sf{uv}}(s)-\frac{2\tilde{G}_{\Sf{uv}}(s)}{\bar{\cD}(s)}\left(\tilde{\cM}_{\sv}(s)-\frac{2\tilde{M}_{\Sf{uv}}(s)\tilde{G}_{\Sf{uv}}(s)}{\tilde{\cG}_{\su}(s)}\right)\right]\frac{s}{x^{2}-s^{2}}\right],\label{Eq:hM_uv}\\
(\hat{M}_{\Sf{vu}})&:\quad M_{\Sf{vu}}(x)=\E_{s^{2}\sim\mu}\left[-\frac{1}{\bar{\cD}(s)}\left(\tilde{M}_{\Sf{vu}}(s)-\frac{2\tilde{\cM}_{\su}(s)\tilde{G}_{\Sf{uv}}(s)}{\tilde{\cG}_{\su}(s)}\right)\frac{s}{x^{2}-s^{2}}\right],\label{Eq:hM_vu}\\
(\hat{M}_{\Sf{vv}})&:\quad M_{\Sf{vv}}(x)=\E_{s^{2}\sim\tilde{\mu}}\left[-\frac{\alpha}{\bar{\cD}(s)}\left(\tilde{\cM}_{\sv}(s)-\frac{2\tilde{M}_{\Sf{uv}}(s)\tilde{G}_{\Sf{uv}}(s)}{\tilde{\cG}_{\su}(s)}\right)\frac{x}{x^{2}-s^{2}}\right].\label{Eq:hM_vv}
\end{align}
Here, (A) follows from Stein's lemma, and the functions $\{\tilde{G}_{\Sf{uv}},\tilde{\cG}_{\su},\tilde{\cG}_{\sv}\}$ are defined in \eqref{Eq:tG_uv}, \eqref{Eq:t_CG_u}, and \eqref{Eq:t_CG_v}. For the $\su$-channel, let $z_{\su}\sim\cN(0,1)$ and $u_{0},u_{a}\overset{\text{i.i.d.}}{\sim}P_{\su}(\cdot)$. The scalar posterior moments are
\begin{align}
\langle u_{a}^{k}\rangle&=\frac{\E_{u_{a}}\left[u_{a}^{k}\exp\left(\hat{m}_{\su}u_{0}u_{a}-\frac{\hat{\beta}_{\su}+\hat{q}_{\su}}{2}u_{a}^{2}+\sqrt{\hat{q}_{\su}}u_{a}z_{\su}\right)\right]}{\E_{u_{a}}\exp\left(\hat{m}_{\su}u_{0}u_{a}-\frac{\hat{\beta}_{\su}+\hat{q}_{\su}}{2}u_{a}^{2}+\sqrt{\hat{q}_{\su}}u_{a}z_{\su}\right)},\quad k\in\{1,2\}.
\end{align}
The $\sv$-channel is obtained by replacing $\su$ with $\sv$ and $\{u_{0},u_{a}\}$ with $\{v_{0},v_{a}\}$.

To obtain \eqref{Eq:tqsu_tbsu}, we add the stationarity conditions w.r.t. $\tilde{q}_{\su}$ and $\tilde{\beta}_{\su}$. Within the functions, $\tilde{\beta}_{\su}$ appears only through the difference $\tilde{\beta}_{\su}-\tilde{q}_{\su}$, as in $\tilde{\cG}_{\su}(\cdot)$. Differentiating this difference w.r.t. the two parameters gives opposite signs, so these contributions cancel when the conditions are added. The remaining contributions from the functions come from terms that depend explicitly on $\tilde{q}_{\su}$. The same argument gives \eqref{Eq:tqsv_tbsv} for the $\sv$-channel.\\

\subsubsection{Simplifications}

We now simplify the saddle point equations in the Bayes-optimal setting. Assuming concentration of the corresponding empirical quantities in the thermodynamic limit, their values can be related to the appropriate expectations through the Nishimori identity \cite{nishimori01}. These relations fix several parameters in the replica calculation and substantially reduce the saddle point system. We first recall the Nishimori identity and then apply the resulting relations to \eqref{Eq:bB}-\eqref{Eq:tqsv_tbsv}, leading to the equations stated in Result \ref{Result:Saddle_Point}.

\begin{lemma}[Nishimori identity \cite{nishimori01}] \label{rmk:Nishimori}

In the Bayes-optimal setting, let $\{\bu_{0},\bY\}$ be drawn from the model in \eqref{eq:model_def} and, conditionally on $\bY$, let $\{\bu_{1},\cdots,\bu_{n}\}$ be i.i.d. samples from the posterior $P_{\su}(\cdot|\bY)$. Then, for any bounded function $f(\cdot)$,
\begin{align}
\E_{\bu_{0},\bY}\langle f(\bY,\bu_{0},\bu_{2},\cdots)\rangle&=\E_{\bY}\langle f(\bY,\bu_{1},\bu_{2},\cdots)\rangle,
\end{align}
where $\langle\cdot\rangle$ denotes the posterior expectation over $\{\bu_{a}\}_{a=1}^{n}$. The same identity holds for the $\sv$-channel.

\end{lemma}

The identity follows directly from the tower property of conditional expectation; see, e.g., \cite{barbier2019optimal}. Together with concentration in the thermodynamic limit, it gives the following relations:
\begin{itemize}

\item The overlap parameters satisfy
\begin{align}
\beta^{(0)}_{\su}=\beta_{\su}=1,\quad \beta^{(0)}_{\sv}=\beta_{\sv}=1,\quad m_{\su}=q_{\su},\quad m_{\sv}=q_{\sv}.\label{Eq:Overlap_Consistency}
\end{align}

\item The conjugate parameters associated with the prior-related terms $g_{\su,n}$ and $g_{\sv,n}$ satisfy
\begin{align}
\hat{\beta}^{(0)}_{\su}=\hat{\beta}_{\su}=0,\quad\hat{\beta}^{(0)}_{\sv}=\hat{\beta}_{\sv}=0,\quad\hat{m}_{\su}=\hat{q}_{\su},\quad\hat{m}_{\sv}=\hat{q}_{\sv}.\label{Eq:Prior_Side_Conjugate}
\end{align}

\item The conjugate parameters associated with the noise-related term $g_{\sz,n}$ satisfy
\begin{align}
\tilde{\beta}^{(0)}_{\su}=\tilde{\beta}_{\su}=1,\quad\tilde{\beta}^{(0)}_{\sv}=\tilde{\beta}_{\sv}=1,\quad\tilde{m}_{\su}=\tilde{q}_{\su},\quad\tilde{m}_{\sv}=\tilde{q}_{\sv}.\label{Eq:Noise_Side_Conjugate}
\end{align}

\end{itemize}
Applying \eqref{Eq:Overlap_Consistency}-\eqref{Eq:Prior_Side_Conjugate} to the stationarity conditions \eqref{Eq:msu}-\eqref{Eq:bsv} yields \eqref{Eq:RM_m_u}-\eqref{Eq:RM_m_v}, as stated in Result \ref{Result:Saddle_Point}.

We next turn to the matrix-valued functions. From \eqref{Eq:Self_Overlap}, $\bC(x)$ has the deterministic limit
\begin{align}
\bC(x)&=\frac{1}{N}\tilde{\bX}_{0}^{\top}\cG_{\tilde{\bZ}_{0}}(x)\tilde{\bX}_{0}\label{Eq:C_X}\\
&\overset{\text{(A)}}{=}\frac{1}{N}\begin{bmatrix}\bu_{0}&\bm{0}\\\bm{0}&\bv_{0}\end{bmatrix}^{\top}\begin{bmatrix}x\cG_{\bZ_{0}\bZ_{0}^{\top}}(x^{2})&\cG_{\bZ_{0}\bZ_{0}^{\top}}(x^{2})\bZ_{0}\\\bZ_{0}^{\top}\cG_{\bZ_{0}\bZ_{0}^{\top}}(x^{2})&x\cG_{\bZ_{0}^{\top}\bZ_{0}}(x^{2})\end{bmatrix}\begin{bmatrix}\bu_{0}&\bm{0}\\\bm{0}&\bv_{0}\end{bmatrix}\nonumber\\
&=\begin{bmatrix}\frac{1}{N}x\bu_{0}^{\top}\cG_{\bZ_{0}\bZ_{0}^{\top}}(x^{2})\bu_{0}&\frac{1}{N}\bu_{0}^{\top}\cG_{\bZ_{0}\bZ_{0}^{\top}}(x^{2})\bZ_{0}\bv_{0}\\\frac{1}{N}\bv_{0}^{\top}\bZ_{0}^{\top}\cG_{\bZ_{0}\bZ_{0}^{\top}}(x^{2})\bu_{0}&\frac{1}{N}x\bv_{0}^{\top}\cG_{\bZ_{0}^{\top}\bZ_{0}}(x^{2})\bv_{0}\end{bmatrix}\nonumber\\
&\overset{\text{(B)}}{\to}\bar{\bC}(x)\triangleq\begin{bmatrix}c_{1}(x)&0\\0&c_{2}(x)\end{bmatrix}.\label{Eq:C_x}
\end{align}
Here, (A) follows from the block-matrix inversion formula in \eqref{Eq:Block_Structure}, obtained via the Schur complement. In (B), the off-diagonal terms vanish in the thermodynamic limit, while the diagonal terms converge to
\begin{align}
c_{1}(x)&\triangleq\lim_{N\to\infty}\frac{1}{N}x\bu_{0}^{\top}\cG_{\bZ_{0}\bZ_{0}^{\top}}(x^{2})\bu_{0}=\E_{s^{2}\sim\mu}\frac{x}{x^{2}-s^{2}},\label{Eq:C_1_x}\\
c_{2}(x)&\triangleq\lim_{N\to\infty}\frac{1}{N}x\bv_{0}^{\top}\cG_{\bZ_{0}^{\top}\bZ_{0}}(x^{2})\bv_{0}=\alpha\left[\frac{1}{\alpha}\E_{s^{2}\sim\mu}\frac{x}{x^{2}-s^{2}}+\left(1-\frac{1}{\alpha}\right)\frac{1}{x}\right].\label{Eq:C_2_x}
\end{align}

Using the Nishimori identity, $\bB(x)$ in \eqref{Eq:Self_Overlap} can be written as
\begin{align}
\bB(x)&=\frac{1}{N}\E\langle\tilde{\bX}_{a}^{\top}\cG_{\tilde{\bY}}(x)\tilde{\bX}_{a}\rangle\nonumber\\
&=\frac{1}{N}\E\tilde{\bX}_{0}^{\top}\cG_{\tilde{\bY}}(x)\tilde{\bX}_{0}\nonumber\\
&\overset{\text{(C)}}{=}\frac{1}{N}\E\tilde{\bX}_{0}^{\top}\left[\cG_{\tilde{\bZ}_{0}}(x)+\frac{\lambda}{N}\cG_{\tilde{\bZ}_{0}}(x)\tilde{\bX}_{0}\left(\tilde{\bI}-\frac{\lambda}{N}\tilde{\bX}_{0}^{\top}\cG_{\tilde{\bZ}_{0}}(x)\tilde{\bX}_{0}\right)^{-1}\tilde{\bX}_{0}^{\top}\cG_{\tilde{\bZ}_{0}}(x)\right]\tilde{\bX}_{0}\nonumber\\
&\overset{\text{\eqref{Eq:C_X}}}{=}\E\left[\bC(x)(\tilde{\bI}-\lambda\bC(x))^{-1}\tilde{\bI}\right]\nonumber\\
&\overset{\text{(D)}}{\to}\bar{\bC}(x)(\tilde{\bI}-\lambda\bar{\bC}(x))^{-1}\tilde{\bI}.\label{Eq:B_X}
\end{align}
Here, (C) follows from the Woodbury matrix identity applied to the resolvent $\cG_{\tilde{\bY}}(x)$, as shown in step (G) of \eqref{Eq:Phi_Y}. In (D), we use the convergence $\bC(x)\to\bar{\bC}(x)$ established in \eqref{Eq:C_x}.

Combining the limiting expression in \eqref{Eq:B_X} with the stationarity conditions \eqref{Eq:bB} and \eqref{Eq:bG} gives
\begin{align}
\hat{\bG}(x)&=-\lambda xV'(x^{2})(\tilde{\bI}+\lambda\bB(x))^{-1}\nonumber\\
&\overset{\text{\eqref{Eq:B_X}}}{\to}-\lambda xV'(x^{2})\left[\tilde{\bI}+\lambda\bar{\bC}(x)(\tilde{\bI}-\lambda\bar{\bC}(x))^{-1}\tilde{\bI}\right]^{-1}\nonumber\\
&\overset{\eqref{Eq:C_x}}{=}-\lambda xV'(x^{2})\begin{bmatrix}-\alpha\lambda\left[\frac{1}{\alpha}\E_{s^{2}\sim\mu}\frac{x}{x^{2}-s^{2}}+\left(1-\frac{1}{\alpha}\right)\frac{1}{x}\right]&1\\1&-\lambda\E_{s^{2}\sim\mu}\frac{x}{x^{2}-s^{2}}\end{bmatrix}.\label{Eq:H_G_x}
\end{align}
Comparing \eqref{Eq:H_G_x} with the componentwise form of $\hat{\bG}(x)$ in \eqref{Eq:Block_Structure}, we obtain
\begin{align}
\hat{G}_{\Sf{uu}}(x)&=\alpha\lambda^{2}xV'(x^{2})\left[\frac{1}{\alpha}\E_{s^{2}\sim\mu}\frac{x}{x^{2}-s^{2}}+\left(1-\frac{1}{\alpha}\right)\frac{1}{x}\right],\\
\hat{G}_{\Sf{vv}}(x)&=\lambda^{2}xV'(x^{2})\E_{s^{2}\sim\mu}\frac{x}{x^{2}-s^{2}},\\
\hat{G}_{\Sf{uv}}(x)&=-\lambda xV'(x^{2}).
\end{align}

Similarly, \eqref{Eq:bM} and \eqref{Eq:bG} give
\begin{align}
\hat{\bM}(x)&=2\lambda\hat{\bG}(x)\bM(x)(\tilde{\bI}-\lambda\bC(x))^{-1}\nonumber\\
&\overset{\text{(E)}}{\to}-2\lambda^{2}xV'(x^{2})\tilde{\bI}\left(\frac{1}{\lambda}\tilde{\bI}-\bar{\bC}(x)\right)\tilde{\bI}\bM(x)\left(\frac{1}{\lambda}\tilde{\bI}-\bar{\bC}(x)\right)^{-1}\nonumber\\
&=-\frac{2\lambda^{2}xV'(x^{2})}{c_{1}(x)c_{2}(x)-\frac{1}{\lambda^{2}}}\begin{bmatrix}\bar{M}_{\Sf{uu}}(x)&\bar{M}_{\Sf{uv}}(x)\\\bar{M}_{\Sf{vu}}(x)&\bar{M}_{\Sf{vv}}(x)\end{bmatrix}.\label{Eq:hM_x_Func}
\end{align}
Here, (E) follows from the limiting expressions in \eqref{Eq:C_x} and \eqref{Eq:H_G_x}. The entries of the matrix above are
\begin{align}
\bar{M}_{\Sf{uu}}(x)=&c^{2}_{2}(x)M_{\Sf{uu}}(x)-\frac{1}{\lambda}c_{2}(x)M_{\Sf{vu}}(x)+\frac{1}{\lambda}c_{2}(x)M_{\Sf{uv}}(x)-\frac{1}{\lambda^{2}}M_{\Sf{vv}}(x),\label{Eq:T_uu}\\
\bar{M}_{\Sf{uv}}(x)=&\frac{1}{\lambda}c_{2}(x)M_{\Sf{uu}}(x)-\frac{1}{\lambda^{2}}M_{\Sf{vu}}(x)+c_{1}(x)c_{2}(x)M_{\Sf{uv}}(x)-\frac{1}{\lambda}c_{1}(x)M_{\Sf{vv}}(x),\label{Eq:T_uv}\\
\bar{M}_{\Sf{vu}}(x)=&-\frac{1}{\lambda}c_{2}(x)M_{\Sf{uu}}(x)+c_{1}(x)c_{2}(x)M_{\Sf{vu}}(x)-\frac{1}{\lambda^{2}}M_{\Sf{uv}}(x)+\frac{1}{\lambda}c_{1}(x)M_{\Sf{vv}}(x),\label{Eq:T_vu}\\
\bar{M}_{\Sf{vv}}(x)=&-\frac{1}{\lambda^{2}}M_{\Sf{uu}}(x)+\frac{1}{\lambda}c_{1}(x)M_{\Sf{vu}}(x)-\frac{1}{\lambda}c_{1}(x)M_{\Sf{uv}}(x)+c^{2}_{1}(x)M_{\Sf{vv}}(x).\label{Eq:T_vv}
\end{align}

We next evaluate the functions entering the stationarity conditions \eqref{Eq:tqsu_tbsu}-\eqref{Eq:tqsv_tbsv}. Substituting $\hat{G}_{\Sf{uu}}(x)$ into $\tilde{G}_{\Sf{uu}}(z)$ in \eqref{Eq:tG_uu} gives
\begin{align}
\tilde{G}_{\Sf{uu}}(z)&=\alpha\lambda^{2}\frac{1}{2\pi\imagUnit}\oint_{\Gamma_{\sx}}{\dd x}\,V'(x^{2})\left[\frac{1}{\alpha}\E_{s^{2}\sim\mu}\frac{x^{2}}{x^{2}-s^{2}}+\left(1-\frac{1}{\alpha}\right)\right]\frac{x}{x^{2}-z^{2}}\nonumber\\
&\overset{\text{(F)}}{=}\alpha\lambda^{2}\frac{1}{4\pi\imagUnit}\oint_{\Gamma_{\Sf{y}}}{\dd y}\,V'(y)\left[\frac{1}{\alpha}\E_{s^{2}\sim\mu}\frac{y}{y-s^{2}}+\left(1-\frac{1}{\alpha}\right)\right]\frac{1}{y-z^{2}}\nonumber\\
&=\frac{\alpha}{2}\lambda^{2}\left[\frac{1}{\alpha}\E_{s^{2}\sim\mu}\frac{1}{z^{2}-s^{2}}\frac{1}{2\pi\imagUnit}\oint_{\Gamma_{\Sf{y}}}{\dd y}\left(\frac{1}{y-z^{2}}-\frac{1}{y-s^{2}}\right)yV'(y)+\left(1-\frac{1}{\alpha}\right)\frac{1}{2\pi\imagUnit}\oint_{\Gamma_{\Sf{y}}}{\dd y}\,V'(y)\frac{1}{y-z^{2}}\right]\nonumber\\
&\overset{\text{(G)}}{=}\frac{\alpha}{2}\lambda^{2}\left[\frac{1}{\alpha}\E_{s^{2}\sim\mu}\frac{z^{2}V'(z^{2})-s^{2}V'(s^{2})}{z^{2}-s^{2}}+\left(1-\frac{1}{\alpha}\right)V'(z^{2})\right]\nonumber\\
&=\frac{\alpha}{2}\lambda^{2}\left[\frac{1}{\alpha}z^{2}\E_{s^{2}\sim\mu}\frac{V'(z^{2})-V'(s^{2})}{z^{2}-s^{2}}+\frac{1}{\alpha}\E_{s^{2}\sim\mu}V'(s^{2})+\left(1-\frac{1}{\alpha}\right)V'(z^{2})\right].\label{Eq:tG_UU}
\end{align}
Here, (F) follows from the change of variable $y=x^{2}$, with $\dd y=2x\dd x$. (G) follows from Cauchy's integral formula. Similarly, applying the same procedure to $\tilde{G}_{\Sf{vv}}(z)$ and $\tilde{G}_{\Sf{uv}}(z)$ in \eqref{Eq:tG_uv}-\eqref{Eq:tG_vv} gives
\begin{align}
\tilde{G}_{\Sf{vv}}(z)&=\frac{1}{2}\lambda^{2}\left(z^{2}\E_{s^{2}\sim\mu}\frac{V'(z^{2})-V'(s^{2})}{z^{2}-s^{2}}+\E_{s^{2}\sim\mu}V'(s^{2})\right),\label{Eq:tG_VV}\\
\tilde{G}_{\Sf{uv}}(z)&=-\frac{1}{2}\lambda zV'(z^{2}).\label{Eq:tG_UV}
\end{align}

Using the relations in \eqref{Eq:Overlap_Consistency}-\eqref{Eq:Noise_Side_Conjugate}, the stationarity conditions \eqref{Eq:qsu_bsu}-\eqref{Eq:qsv_bsv} reduce to
\begin{align}
\tilde{m}_{\su}&=1-(1-m_{\su})^{-1}+\hat{m}_{\su},\label{Eq:tm_u}\\
\tilde{m}_{\sv}&=1-(1-m_{\sv})^{-1}+\hat{m}_{\sv}.\label{Eq:tm_v}
\end{align}
Combining these relations with \eqref{Eq:tG_UU}-\eqref{Eq:tG_VV}, the functions $\tilde{\cG}_{\su}(x)$ and $\tilde{\cG}_{\sv}(x)$ in \eqref{Eq:t_CG_u}-\eqref{Eq:t_CG_v} become
\begin{align}
\tilde{\cG}_{\su}(x)&=1-\tilde{m}_{\su}+2\tilde{G}_{\Sf{uu}}(x)=(1-m_{\su})^{-1}-\hat{m}_{\su}+\lambda^{2}\left[x^{2}\E_{s^{2}\sim\mu}\frac{V'(x^{2})-V'(s^{2})}{x^{2}-s^{2}}+\E_{s^{2}\sim\mu}V'(s^{2})+(\alpha-1)V'(x^{2})\right],\label{Eq:tG_U}\\
\tilde{\cG}_{\sv}(x)&=1-\tilde{m}_{\sv}+2\tilde{G}_{\Sf{vv}}(x)=(1-m_{\sv})^{-1}-\hat{m}_{\sv}+\lambda^{2}\left(x^{2}\E_{s^{2}\sim\mu}\frac{V'(x^{2})-V'(s^{2})}{x^{2}-s^{2}}+\E_{s^{2}\sim\mu}V'(s^{2})\right).\label{Eq:tG_V}
\end{align}
Substituting these expressions into \eqref{Eq:tqsu_tbsu}-\eqref{Eq:tqsv_tbsv} gives
\begin{align}
1-m_{\su}&=\E_{s^{2}\sim\mu}\left[\tilde{\cG}_{\su}(s)-\frac{4\tilde{G}_{\Sf{uv}}^{2}(s)}{\tilde{\cG}_{\sv}(s)}\right]^{-1},\label{Eq:Eean_beta0u_mu}\\
1-m_{\sv}&=\frac{1}{\alpha}\E_{s^{2}\sim\mu}\left[\tilde{\cG}_{\sv}(s)-\frac{4\tilde{G}_{\Sf{uv}}^{2}(s)}{\tilde{\cG}_{\su}(s)}\right]^{-1}+\left(1-\frac{1}{\alpha}\right)(\tilde{\cG}_{\sv}(0))^{-1}=\E_{s^{2}\sim\tilde{\mu}}\left[\tilde{\cG}_{\sv}(s)-\frac{4\tilde{G}_{\Sf{uv}}^{2}(s)}{\tilde{\cG}_{\su}(s)}\right]^{-1}.\label{Eq:Eean_beta0v_mv}
\end{align}

To rewrite \eqref{Eq:Eean_beta0u_mu}-\eqref{Eq:Eean_beta0v_mv} in the form stated in Result \ref{Result:Saddle_Point}, we first relate the functions obtained above to $\{\cJ_{\Sf{uu}},\cJ_{\Sf{vv}},\cJ_{\Sf{uv}}\}$. Comparing \eqref{Eq:tG_UU}-\eqref{Eq:tG_UV} with \eqref{Eq:Juu_x}-\eqref{Eq:Juv_x} gives
\begin{align}
2\tilde{G}_{\Sf{uu}}(x)=\cJ_{\Sf{uu}}(x^{2}),\quad2\tilde{G}_{\Sf{vv}}(x)=\cJ_{\Sf{vv}}(x^{2}),\quad4\tilde{G}_{\Sf{uv}}^{2}(x)=\cJ_{\Sf{uv}}(x^{2}).
\end{align}
Therefore,
\begin{align}
\tilde{\cG}_{\su}(x)=(1-m_{\su})^{-1}-\hat{m}_{\su}+\cJ_{\Sf{uu}}(x^{2}),\quad\tilde{\cG}_{\sv}(x)=(1-m_{\sv})^{-1}-\hat{m}_{\sv}+\cJ_{\Sf{vv}}(x^{2}).
\end{align}
The definitions of $\cJ_{\su}(\cdot)$ and $\cJ_{\sv}(\cdot)$ in \eqref{Eq:Ju_x}-\eqref{Eq:Jv_x} then imply
\begin{align}
\tilde{\cG}_{\su}(x)-\frac{4\tilde{G}_{\Sf{uv}}^{2}(x)}{\tilde{\cG}_{\sv}(x)}&=(1-m_{\su})^{-1}-\hat{m}_{\su}+\cJ_{\Sf{uu}}(x^{2})-\frac{\cJ_{\Sf{uv}}(x^{2})}{(1-m_{\sv})^{-1}-\hat{m}_{\sv}+\cJ_{\Sf{vv}}(x^{2})}=(1-m_{\su})^{-1}-\hat{m}_{\su}-\cJ_{\su}(x^{2}),\\
\tilde{\cG}_{\sv}(x)-\frac{4\tilde{G}_{\Sf{uv}}^{2}(x)}{\tilde{\cG}_{\su}(x)}&=(1-m_{\sv})^{-1}-\hat{m}_{\sv}+\cJ_{\Sf{vv}}(x^{2})-\frac{\cJ_{\Sf{uv}}(x^{2})}{(1-m_{\su})^{-1}-\hat{m}_{\su}+\cJ_{\Sf{uu}}(x^{2})}=(1-m_{\sv})^{-1}-\hat{m}_{\sv}-\cJ_{\sv}(x^{2}).
\end{align}
Hence, \eqref{Eq:Eean_beta0u_mu}-\eqref{Eq:Eean_beta0v_mv} can be rewritten as
\begin{align}
1-m_{\su}&=\E_{\gamma\sim\mu}\left[(1-m_{\su})^{-1}-\hat{m}_{\su}-\cJ_{\su}(\gamma)\right]^{-1}=\cS_{\cJ_{\su};\mu}((1-m_{\su})^{-1}-\hat{m}_{\su}),\label{Eq:R_u_Pre}\\
1-m_{\sv}&=\E_{\gamma\sim\tilde{\mu}}\left[(1-m_{\sv})^{-1}-\hat{m}_{\sv}-\cJ_{\sv}(\gamma)\right]^{-1}=\cS_{\cJ_{\sv};\tilde{\mu}}((1-m_{\sv})^{-1}-\hat{m}_{\sv}).\label{Eq:R_v_Pre}
\end{align}
Equations \eqref{Eq:R_u_Pre}-\eqref{Eq:R_v_Pre} are now in the form of the Stieltjes transform defined in \eqref{Eq:S_Tran}. Using the relation between the Stieltjes and R-transforms in \eqref{Eq:R_Tran}, \eqref{Eq:R_u_Pre} gives
\begin{align}
(1-m_{\su})^{-1}-\hat{m}_{\su}&=\cS_{\cJ_{\su};\mu}^{-1}(1-m_{\su})=(1-m_{\su})^{-1}+\cR_{\cJ_{\su};\mu}(1-m_{\su}).
\end{align}
It follows that
\begin{align}
\hat{m}_{\su}&=-\cR_{\cJ_{\su};\mu}(1-m_{\su}),
\end{align}
recovering \eqref{Eq:RM_hm_u}. For the $\sv$-channel, \eqref{Eq:R_v_Pre} similarly gives
\begin{align}
\hat{m}_{\sv}
&=-\cR_{\cJ_{\sv};\tilde{\mu}}(1-m_{\sv}),
\end{align}
and hence \eqref{Eq:RM_hm_v}. Together with \eqref{Eq:RM_m_u}-\eqref{Eq:RM_m_v}, these equations give the saddle point system stated in Result \ref{Result:Saddle_Point}.

\subsection{Final Simplification of the Free Entropy}\label{Sec:Final_Free_Entropy}

The replica calculation above led to the intermediate free entropy in \eqref{Eq:Intermediate_Free_Entropy}. At that stage, the saddle point equations had not yet been simplified using the Nishimori identities, so the free entropy was left in terms of the full set of RS parameters. Having now established Result \ref{Result:Saddle_Point}, we return to \eqref{Eq:Intermediate_Free_Entropy} to complete the simplification. Under the relations in \eqref{Eq:Overlap_Consistency}-\eqref{Eq:Noise_Side_Conjugate}, all scalar terms reduce directly to those in \eqref{eq:RSfreeEnt}. The remaining term is the contour integral.

Let $\cI$ denote the contour contribution in \eqref{Eq:Intermediate_Free_Entropy}. We have
\begin{align}
\cI&=\frac{1}{2\pi\imagUnit}\oint_{\Gamma_{\sx}}{\dd x}\,\Bigg\{xV'(x^{2})\log|\tilde{\bI}+\lambda\bB(x)|+\Tr\,\hat{\bG}(x)\bG(x)+\Tr\,\hat{\bM}^{\top}(x)\bM(x)\nonumber\\
&\quad+\Tr\,\hat{\bB}(x)\left[\bB(x)-\bG(x)-\bM(x)\lambda(\tilde{\bI}-\lambda\bC(x))^{-1}\bM^{\top}(x)\right]\Bigg\}\nonumber\\
&\overset{\text{\eqref{Eq:bG}}}{=}\frac{1}{2\pi\imagUnit}\oint_{\Gamma_{\sx}}{\dd x}\,\Bigg\{xV'(x^{2})\log|\tilde{\bI}+\lambda\bB(x)|+\Tr\,\hat{\bM}^{\top}(x)\bM(x)+\Tr\,\hat{\bB}(x)\left[\bB(x)-\lambda\bM(x)(\tilde{\bI}-\lambda\bC(x))^{-1}\bM^{\top}(x)\right]\Bigg\}\nonumber\\
&\overset{\text{\eqref{Eq:bM}}}{=}\frac{1}{2\pi\imagUnit}\oint_{\Gamma_{\sx}}{\dd x}\,\left(xV'(x^{2})\log|\tilde{\bI}+\lambda\bB(x)|+\Tr\,\hat{\bB}(x)\bB(x)+\frac{1}{2}\Tr\,\hat{\bM}^{\top}(x)\bM(x)\right)\nonumber\\
&\overset{\text{(A)}}{=}\frac{1}{4\pi\imagUnit}\oint_{\Gamma_{\sx}}{\dd x}\,\Tr\,\hat{\bM}^{\top}(x)\bM(x)+\psi_{\alpha,\lambda}.\label{Eq:Contour_Free_Entropy}
\end{align}
For (A), we use \eqref{Eq:bB} and \eqref{Eq:B_X}, which show that, in the thermodynamic limit, $\bB(x)$ and $\hat{\bB}(x)$ are independent of the remaining scalar parameters $\{m_{\su},m_{\sv},\hat{m}_{\su},\hat{m}_{\sv}\}$. Their contribution is therefore common to all scalar stationary solutions and is given by
\begin{align}
\psi_{\alpha,\lambda}&\triangleq\frac{1}{2\pi\imagUnit}\oint_{\Gamma_{\sx}}{\dd x}\,\left(xV'(x^{2})\log|\tilde{\bI}+\lambda\bB(x)|+\Tr\,\hat{\bB}(x)\bB(x)\right).\label{eq:free_entropy_additive_constant}
\end{align}
Finally, substituting \eqref{Eq:Contour_Free_Entropy} together with \eqref{Eq:Overlap_Consistency}-\eqref{Eq:Noise_Side_Conjugate} into \eqref{Eq:Intermediate_Free_Entropy} gives the RS free entropy in \eqref{eq:RSfreeEnt}, completing Result \ref{Result:Free_Entropy}.

\section{Conclusion and Future Work}\label{Sec:Conclusion}

In this work, we studied the IT limits of rectangular rank-one spiked matrix models with RI noise. Using the replica method, we derived the asymptotic free entropy and the corresponding Bayes-optimal saddle point equations. We then showed that, under an explicit parameter mapping, these equations coincide with the SE fixed point equations of the OAMP in \cite{Chen2026orthogonal}. This correspondence provides an algorithmic interpretation of the replica prediction and supports the predicted Bayes optimality of OAMP when the SE reaches the information-theoretically selected branch. The numerical results further confirm the correspondence between the replica saddle point equations and the SE fixed points of OAMP. In the unique informative fixed point regime, both descriptions yield the same performance, whereas in the multiple fixed point regime, free entropy comparison identifies the IT branch and the OAMP SE may instead converge to a metastable branch.

Several questions remain open. A natural next step is to establish the replica free entropy and the resulting Bayes-optimal MMSE characterization rigorously. Another direction is to extend the present rank-one analysis to finite-rank rectangular spiked models. Ongoing work aims to develop both OAMP-type algorithms and the corresponding replica analysis for this setting. Finally, a better understanding of hard phases and the associated statistical-to-computational gaps remains an important direction for future work.

\section*{Acknowledgements}

Part of this work was done while F.C was affiliated with the International Center for Theoretical Physics. J.B and F.C were funded by the European Union (ERC, CHORAL, Project No. 101039794). Views and opinions expressed are however those of the authors only and do not necessarily reflect those of the European Union or the European Research Council. Neither the European Union nor the granting authority can be held responsible for them. R.S. acknowledges discussions with Yizhou Xu at the initial stage of the work. Q.C. thanks Haohua Chen for helpful discussions while establishing the equivalence between the RS saddle point equations and the SE fixed point equations of OAMP. J.~Ma was supported by the National Key R\&D Program of China under Grant No.~2024YFA1014200, and the National Natural Science Foundation of China under Grant No.~62571526 and No.~12622133.

\bibliographystyle{IEEEtran}
\bibliography{IEEEabrv,CQ_original_mod,CQ_bib_mod}

\appendices

\section{Parameters and Functions Entering the RS Variational functional $f_{\Sf{RS}}$}\label{App:RS_Coupling_Functions}

Result \ref{Result:Free_Entropy} expresses the asymptotic free entropy $f$ through an extremization of the RS variational functional $f_{\Sf{RS}}$. The parameters and functions entering the functional are introduced at different stages of the replica derivation. For ease of reference, their final expressions are collected here.

Let $\operatorname{Sol}_{\Sf{RS}}\triangleq\left\{\left(m_{\su,\ell},m_{\sv,\ell},\hat{m}_{\su,\ell},\hat{m}_{\sv,\ell}\right)\right\}_{\ell=1}^{L}$ denote the set of scalar stationary solutions to the saddle point equations in Result \ref{Result:Saddle_Point}, where $L$ is the number of such solutions. We fix one solution in $\operatorname{Sol}_{\Sf{RS}}$ and suppress its branch index $\ell$ in what follows. The parameters and functions needed to evaluate $f_{\Sf{RS}}$ on this branch are organized into three groups:
\begin{itemize}

\item \textit{Directly evaluated parameters and functions}. Once a scalar stationary solution is fixed, the following expressions can be evaluated:
\begin{align}
\tilde{m}_{\su}&=1-(1-m_{\su})^{-1}+\hat{m}_{\su},\\
\tilde{m}_{\sv}&=1-(1-m_{\sv})^{-1}+\hat{m}_{\sv},\\
c_{\Sf{1}}(x)&=\E_{s^{2}\sim\mu}\frac{x}{x^{2}-s^{2}},\\
c_{\Sf{2}}(x)&=\alpha\left[\frac{1}{\alpha}\E_{s^{2}\sim\mu}\frac{x}{x^{2}-s^{2}}+\left(1-\frac{1}{\alpha}\right)\frac{1}{x}\right],\\
\log\varphi_{\alpha,\lambda}&=\frac{1}{2\pi\imagUnit}\oint_{\Gamma_{\sx}}{\dd x}\,xV'(x^{2})\log\left(\lambda^{2}c_{\Sf{1}}(x)c_{\Sf{2}}(x)-1\right),\\
\psi_{\alpha,\lambda}&=\frac{1}{2\pi\imagUnit}\oint_{\Gamma_{\sx}}{\dd x}\,\left(xV'(x^{2})\log|\tilde{\bI}+\lambda\bB(x)|+\Tr\,\hat{\bB}(x)\bB(x)\right),\\
\tilde{G}_{\Sf{uv}}(z)&=-\frac{1}{2}\lambda zV'(z^{2}),\label{Eq:FE_tG_UV}\\
\tilde{\cG}_{\su}(z)&=(1-m_{\su})^{-1}-\hat{m}_{\su}+\lambda^{2}\left[z^{2}\E_{s^{2}\sim\mu}\frac{V'(z^{2})-V'(s^{2})}{z^{2}-s^{2}}+\E_{s^{2}\sim\mu}V'(s^{2})+(\alpha-1)V'(z^{2})\right],\label{Eq:FE_tG_U}\\
\tilde{\cG}_{\sv}(z)&=(1-m_{\sv})^{-1}-\hat{m}_{\sv}+\lambda^{2}\left(z^{2}\E_{s^{2}\sim\mu}\frac{V'(z^{2})-V'(s^{2})}{z^{2}-s^{2}}+\E_{s^{2}\sim\mu}V'(s^{2})\right),\label{Eq:FE_tG_V}\\
\bar{\cD}(z)&=\tilde{\cG}_{\sv}(z)-\frac{4\tilde{G}_{\Sf{uv}}^{2}(z)}{\tilde{\cG}_{\su}(z)}.\label{Eq:FE_b_D}
\end{align}

\item \textit{Closed self-consistency map}. For the fixed scalar stationary solution, the following relations jointly determine the functions:
\begin{align}
\tilde{\cM}_{\su}(z)&=\tilde{m}_{\su}+\tilde{M}_{\Sf{uu}}(z),\label{Eq:FE_t_CM_u}\\
\tilde{\cM}_{\sv}(z)&=\tilde{m}_{\sv}+\tilde{M}_{\Sf{vv}}(z),\label{Eq:FE_t_CM_v}\\
M_{\Sf{uu}}(x)&=\E_{s^{2}\sim\mu}\left[-\frac{1}{\tilde{\cG}_{\su}(s)}\left[\tilde{\cM}_{\su}(s)-\frac{2\tilde{G}_{\Sf{uv}}(s)}{\bar{\cD}(s)}\left(\tilde{M}_{\Sf{vu}}(s)-\frac{2\tilde{\cM}_{\su}(s)\tilde{G}_{\Sf{uv}}(s)}{\tilde{\cG}_{\su}(s)}\right)\right]\frac{x}{x^{2}-s^{2}}\right],\label{Eq:FE_M_uu}\\
M_{\Sf{uv}}(x)&=\E_{s^{2}\sim\mu}\left[-\frac{1}{\tilde{\cG}_{\su}(s)}\left[\tilde{M}_{\Sf{uv}}(s)-\frac{2\tilde{G}_{\Sf{uv}}(s)}{\bar{\cD}(s)}\left(\tilde{\cM}_{\sv}(s)-\frac{2\tilde{M}_{\Sf{uv}}(s)\tilde{G}_{\Sf{uv}}(s)}{\tilde{\cG}_{\su}(s)}\right)\right]\frac{s}{x^{2}-s^{2}}\right],\label{Eq:FE_M_uv}\\
M_{\Sf{vu}}(x)&=\E_{s^{2}\sim\mu}\left[-\frac{1}{\bar{\cD}(s)}\left(\tilde{M}_{\Sf{vu}}(s)-\frac{2\tilde{\cM}_{\su}(s)\tilde{G}_{\Sf{uv}}(s)}{\tilde{\cG}_{\su}(s)}\right)\frac{s}{x^{2}-s^{2}}\right],\label{Eq:FE_M_vu}\\
M_{\Sf{vv}}(x)&=\E_{s^{2}\sim\tilde{\mu}}\left[-\frac{\alpha}{\bar{\cD}(s)}\left(\tilde{\cM}_{\sv}(s)-\frac{2\tilde{M}_{\Sf{uv}}(s)\tilde{G}_{\Sf{uv}}(s)}{\tilde{\cG}_{\su}(s)}\right)\frac{x}{x^{2}-s^{2}}\right],\label{Eq:FE_M_vv}\\
\hat{M}_{\Sf{uu}}(x)&=-\frac{2\lambda^{2}xV'(x^{2})}{c_{\Sf{1}}(x)c_{\Sf{2}}(x)-\lambda^{-2}}\left(c_{\Sf{2}}^{2}(x)M_{\Sf{uu}}(x)-\frac{1}{\lambda}c_{\Sf{2}}(x)M_{\Sf{vu}}(x)+\frac{1}{\lambda}c_{\Sf{2}}(x)M_{\Sf{uv}}(x)-\frac{1}{\lambda^{2}}M_{\Sf{vv}}(x)\right),\label{Eq:FE_hM_uu}\\
\hat{M}_{\Sf{uv}}(x)&=-\frac{2\lambda^{2}xV'(x^{2})}{c_{\Sf{1}}(x)c_{\Sf{2}}(x)-\lambda^{-2}}\left(\frac{1}{\lambda}c_{\Sf{2}}(x)M_{\Sf{uu}}(x)-\frac{1}{\lambda^{2}}M_{\Sf{vu}}(x)+c_{\Sf{1}}(x)c_{\Sf{2}}(x)M_{\Sf{uv}}(x)-\frac{1}{\lambda}c_{\Sf{1}}(x)M_{\Sf{vv}}(x)\right),\label{Eq:FE_hM_uv}\\
\hat{M}_{\Sf{vu}}(x)&=-\frac{2\lambda^{2}xV'(x^{2})}{c_{\Sf{1}}(x)c_{\Sf{2}}(x)-\lambda^{-2}}\left(-\frac{1}{\lambda}c_{\Sf{2}}(x)M_{\Sf{uu}}(x)+c_{\Sf{1}}(x)c_{\Sf{2}}(x)M_{\Sf{vu}}(x)-\frac{1}{\lambda^{2}}M_{\Sf{uv}}(x)+\frac{1}{\lambda}c_{\Sf{1}}(x)M_{\Sf{vv}}(x)\right),\label{Eq:FE_hM_vu}\\
\hat{M}_{\Sf{vv}}(x)&=-\frac{2\lambda^{2}xV'(x^{2})}{c_{\Sf{1}}(x)c_{\Sf{2}}(x)-\lambda^{-2}}\left(-\frac{1}{\lambda^{2}}M_{\Sf{uu}}(x)+\frac{1}{\lambda}c_{\Sf{1}}(x)M_{\Sf{vu}}(x)-\frac{1}{\lambda}c_{\Sf{1}}(x)M_{\Sf{uv}}(x)+c_{\Sf{1}}^{2}(x)M_{\Sf{vv}}(x)\right),\label{Eq:FE_hM_vv}\\
\tilde{M}_{\Sf{uu}}(z)&=\frac{1}{2\pi\imagUnit}\oint_{\Gamma_{\sx}}{\dd x}\,\hat{M}_{\Sf{uu}}(x)\frac{x}{x^{2}-z^{2}},\label{Eq:FE_tM_uu}\\
\tilde{M}_{\Sf{uv}}(z)&=\frac{1}{2\pi\imagUnit}\oint_{\Gamma_{\sx}}{\dd x}\,\hat{M}_{\Sf{uv}}(x)\frac{z}{x^{2}-z^{2}},\label{Eq:FE_tM_uv}\\
\tilde{M}_{\Sf{vu}}(z)&=\frac{1}{2\pi\imagUnit}\oint_{\Gamma_{\sx}}{\dd x}\,\hat{M}_{\Sf{vu}}(x)\frac{z}{x^{2}-z^{2}},\label{Eq:FE_tM_vu}\\
\tilde{M}_{\Sf{vv}}(z)&=\frac{1}{2\pi\imagUnit}\oint_{\Gamma_{\sx}}{\dd x}\,\hat{M}_{\Sf{vv}}(x)\frac{x}{x^{2}-z^{2}}.\label{Eq:FE_tM_vv}
\end{align}

\item \textit{Functions evaluated at the fixed point}. Once the preceding relations have been solved, the remaining scalar function entering the potential is
\begin{align}
\bar{\cN}(z)&=\tilde{m}_{\sv}-\tilde{M}_{\Sf{vu}}^{2}(z)+2\tilde{G}_{\Sf{uv}}(z)\left(\frac{2\tilde{\cM}_{\su}(z)\tilde{M}_{\Sf{vu}}(z)}{\tilde{\cG}_{\su}(z)}+\frac{2(\tilde{m}_{\su}-\tilde{\cM}_{\su}^{2}(z))\tilde{G}_{\Sf{uv}}(z)}{\tilde{\cG}_{\su}^{2}(z)}\right)-\left(\tilde{\cM}_{\sv}(z)-\frac{2\tilde{M}_{\Sf{uv}}(z)\tilde{G}_{\Sf{uv}}(z)}{\tilde{\cG}_{\su}(z)}\right)^{2}.\label{Eq:FE_b_N}
\end{align}
The matrix-valued functions appearing in the contour term are then assembled from their entries as
\begin{align}
\bM(x)&=\begin{bmatrix}M_{\Sf{uu}}(x)&M_{\Sf{uv}}(x)\\M_{\Sf{vu}}(x)&M_{\Sf{vv}}(x)\end{bmatrix},\quad\hat{\bM}(x)=\begin{bmatrix}\hat{M}_{\Sf{uu}}(x)&\hat{M}_{\Sf{uv}}(x)\\\hat{M}_{\Sf{vu}}(x)&\hat{M}_{\Sf{vv}}(x)\end{bmatrix}.\label{Eq:FE_Matrix_Entries}
\end{align}

\end{itemize}

\section{Numerical Implementation}\label{App:NumericalSP}

\subsection{Numerical Solution of the Saddle Point Equations}\label{Sec:Numerical_Solution}

We solve the saddle point equations numerically using two complementary procedures. The first uses fixed-point iteration to efficiently trace a solution branch. The second performs two-dimensional grid-based residual minimization. Although computationally more demanding, the grid search can reveal whether additional stationary branches coexist.\\

\subsubsection{Fixed-Point Iteration}

The saddle point equations in Result \ref{Result:Saddle_Point} are treated as update rules for $\{m^{t}_{\su},m^{t}_{\sv},\hat{m}^{t}_{\su},\hat{m}^{t}_{\sv}\}$:
\begin{align}
m^{t}_{\su}&=\E_{u_{0},z_{\su}}[u_{0}\langle u_{a}\rangle_{t}],\label{Eq:RM_m_u_{t}}\\
m^{t}_{\sv}&=\E_{v_{0},z_{\sv}}[v_{0}\langle v_{a}\rangle_{t}],\label{Eq:RM_m_v_{t}}\\
\hat{m}^{t+1}_{\su}&=-\cR_{\cJ^{t}_{\su};\mu}(1-m^{t}_{\su}),\label{Eq:RM_hm_u_{t}}\\
\hat{m}^{t+1}_{\sv}&=-\cR_{\cJ^{t}_{\sv};\tilde{\mu}}(1-m^{t}_{\sv}).\label{Eq:RM_hm_v_{t}}
\end{align}
The iteration is initialized at $\hat m_{\su}^{0}=\hat m_{\sv}^{0}=\hat m_{\rm init}$, with $0<\hat m_{\rm init}\ll 1$. The subscript $t$ in $\langle\,\cdot\,\rangle_{t}$ indicates that the scalar posterior in \eqref{eq:scalPosterior} is evaluated at the current values $\{\hat m^{t}_{\su},\hat m^{t}_{\sv}\}$. The functions $\{\cJ^{t}_{\su},\cJ^{t}_{\sv}\}$ are constructed from \eqref{Eq:Ju_x}-\eqref{Eq:Jv_x}. Since their R-transforms generally have no closed-form expressions, we evaluate them numerically through the inverse Stieltjes transform. For the $\su$-channel,
\begin{align}
\cR_{\cJ^{t}_{\su};\mu}(1-m^{t}_{\su})=\underbrace{\cS^{-1}_{\cJ^{t}_{\su};\mu}(1-m^{t}_{\su})}_{\triangleq\xi^{t}_{\su}}-(1-m^{t}_{\su})^{-1}\quad\implies\quad1-m^{t}_{\su}=\cS_{\cJ^{t}_{\su};\mu}(\xi^{t}_{\su})=\E_{\gamma\sim\mu}[\xi^{t}_{\su}-\cJ^{t}_{\su}(\gamma)]^{-1}.
\end{align}
We determine $\xi^{t}_{\su}$ by bisection on the right-hand branch $\left[\max_{\gamma\in\operatorname{supp}(\mu)}\cJ^{t}_{\su}(\gamma)+\eta,+\infty\right)$ where $\eta>0$ is introduced for numerical stability. The $\sv$-channel is evaluated analogously. Fixed-point iteration is computationally efficient and therefore well suited to tracing a solution branch. However, it returns only the branch selected by the iteration and does not rule out the coexistence of other stationary branches.\\

\subsubsection{Grid-Based Residual Minimization}

To search for additional stationary solutions, we supplement the iteration with grid-based residual minimization. Eliminating the intermediate overlaps $\{m_{\su},m_{\sv}\}$ from Result \ref{Result:Saddle_Point} gives a coupled system for the SNR-related order parameters:
\begin{align}
\hat{m}_{\su}=F_{\su}(\hat{m}_{\su},\hat{m}_{\sv}),\quad\hat{m}_{\sv}=F_{\sv}(\hat{m}_{\su},\hat{m}_{\sv}).\label{Eq:Unit_Saddle_Point}
\end{align}
We define the residual error as
\begin{align}
\mathcal{E}(\hat{m}_{\su},\hat{m}_{\sv})&\triangleq\left(F_{\su}(\hat{m}_{\su},\hat{m}_{\sv})-\hat{m}_{\su}\right)^{2}+\left(F_{\sv}(\hat{m}_{\su},\hat{m}_{\sv})-\hat{m}_{\sv}\right)^{2}.
\end{align}
The residual is evaluated over a prescribed two-dimensional domain of $\{\hat{m}_{\su},\hat{m}_{\sv}\}$. Grid points associated with local minima of $\mathcal{E}(\hat{m}_{\su},\hat{m}_{\sv})$ are then used to initialize a local refinement step, yielding high-precision stationary solutions.

\subsection{Free Entropy Evaluation and Branch Selection}\label{Sec:Branch_Selection}

The grid-based residual minimization described above returns the stationary solutions detected within the prescribed search domain. We collect their scalar order parameters as $\operatorname{Sol}_{\Sf{RS}}=\left\{\left(m_{\su,\ell},m_{\sv,\ell},\hat{m}_{\su,\ell},\hat{m}_{\sv,\ell}\right)\right\}_{\ell=1}^{L}$. To determine which of these solutions corresponds to the IT limit, we must evaluate the replica free entropy at each solution and compare the resulting values. An element of $\operatorname{Sol}_{\Sf{RS}}$ does not directly provide all the quantities required for this evaluation. The expression in Result \ref{Result:Free_Entropy} also depends on the self-consistent functions $\{\tilde{\bM},\bM,\hat{\bM}\}$. For each scalar stationary solution, these functions satisfy the closed self-consistency map
\begin{align}
\tilde{\bM}\xrightarrow{\text{\eqref{Eq:FE_M_uu}-\eqref{Eq:FE_M_vv}}}\bM\xrightarrow{\text{\eqref{Eq:FE_hM_uu}-\eqref{Eq:FE_hM_vv}}}\hat{\bM}\xrightarrow{\text{\eqref{Eq:FE_tM_uu}-\eqref{Eq:FE_tM_vv}}}\tilde{\bM}.
\end{align}
The functions must therefore be obtained from a fixed point of this map before the free entropy can be evaluated.

Once the scalar order parameters are fixed, the closed map is affine. We replace the spectral and contour integrals by finite quadrature sums, which converts the functional fixed-point problem into a finite-dimensional linear system. After solving this system for every element of $\operatorname{Sol}_{\Sf{RS}}
$, we substitute the resulting functions into Result \ref{Result:Free_Entropy} and compare the corresponding free entropies. The additive terms $\log\varphi_{\alpha,\lambda}$ and $\psi_{\alpha,\lambda}$ are independent of the stationary solution and therefore do not affect this comparison. We omit these common terms and denote the remaining branch-dependent contribution by $\tilde{f}$.\\

\subsubsection{Discretization and Function Solution}

Let $\{\gamma_{i},w_{i}\}_{i=1}^{Q}$ denote the Gauss-Legendre nodes and effective weights on the spectral interval $[\gamma_{\min},\gamma_{\max}]$, and set $s_{i}\triangleq\sqrt{\gamma_{i}}$. The weights $w_{i}$ include the density of $\mu(\cdot)$. Since $\tilde{\mu}=\frac{1}{\alpha}\mu+(1-\frac{1}{\alpha})\delta_{0}$, the corresponding bulk weights are $\frac{1}{\alpha}w_{i}$, supplemented by an atom of mass $1-\frac{1}{\alpha}$ at zero.

We discretize $\Gamma_{\sy}$ in the $\sy$-plane. For a small $\epsilon>0$, define
\begin{align}
c_{\Sf{\Gamma}}&\triangleq\frac{\gamma_{\max}+\gamma_{\min}}{2},\quad r_{\Sf{\Gamma}}\triangleq(1+\epsilon)\frac{\gamma_{\max}-\gamma_{\min}}{2},\\
\theta_{j}&\triangleq\frac{2\pi(j-1)}{Q_{\Sf{\Gamma}}},\quad y_{j}\triangleq c_{\Sf{\Gamma}}+r_{\Sf{\Gamma}}\exp(\imagUnit\theta_{j}),\quad j\in\{1,\cdots,Q_{\Sf{\Gamma}}\}.\label{Eq:Numerical_Contour_y}
\end{align}
We choose $\epsilon$ sufficiently small that $r_{\Sf{\Gamma}}<c_{\Sf{\Gamma}}$. The resulting contour encloses $[\gamma_{\min},\gamma_{\max}]$ but excludes the origin. We discretize $\Gamma_{\sx}$ in the $\sx$-plane using the positive square-root branch, with $Q_{\Sf{\Gamma}}$ nodes and weights given by
\begin{align}
z_{j}&\triangleq\sqrt{y_{j}},\quad\rho_{j}\triangleq\frac{1}{2\pi\imagUnit}\frac{2\pi}{Q_{\Sf{\Gamma}}}\frac{\imagUnit r_{\Sf{\Gamma}}\exp(\imagUnit\theta_{j})}{2z_{j}},\quad j\in\{1,\cdots,Q_{\Sf{\Gamma}}\}.\label{Eq:Numerical_Contour_z}
\end{align}

We represent the unknown functions by
\begin{align}
x_{\Sf{vv},0}&\triangleq\tilde{M}_{\Sf{vv}}(0),\quad x_{\bullet,i}\triangleq\tilde{M}_{\bullet}(s_{i}),\quad\bullet\in\{\Sf{uu},\Sf{uv},\Sf{vu},\Sf{vv}\},\quad i\in\{1,\cdots,Q\},\\
\bx&\triangleq\left[x_{\Sf{uu},1},x_{\Sf{uv},1},x_{\Sf{vu},1},x_{\Sf{vv},1},\cdots,x_{\Sf{uu},Q},x_{\Sf{uv},Q},x_{\Sf{vu},Q},x_{\Sf{vv},Q},x_{\Sf{vv},0}\right]^{\top}\in\mathbb{C}^{4Q+1}.\label{Eq:Aux_Vector_Def}
\end{align}
For each element of $\operatorname{Sol}_{\Sf{RS}}$, we first evaluate $\{\tilde{G}_{\Sf{uv}},\tilde{\cG}_{\su},\tilde{\cG}_{\sv},\bar{\cD}\}$ from \eqref{Eq:FE_tG_UV}-\eqref{Eq:FE_b_D}. For a trial vector $\bx$, we then evaluate $\{\tilde{\cM}_{\su},\tilde{\cM}_{\sv}\}$ from \eqref{Eq:FE_t_CM_u}-\eqref{Eq:FE_t_CM_v}, with $\tilde{M}_{\bullet}(s_{i})$ represented by $x_{\bullet,i}$ and $\tilde{M}_{\Sf{vv}}(0)$ represented by $x_{\Sf{vv},0}$.

The spectral integrals in \eqref{Eq:FE_M_uu}-\eqref{Eq:FE_M_vv} are evaluated at the contour nodes $\{z_{j}\}_{j=1}^{Q_{\Sf{\Gamma}}}$ using the quadrature sums
\begin{align}
M_{\Sf{uu}}(z_{j})&=\sum_{i=1}^{Q}\left\{-\frac{1}{\tilde{\cG}_{\su}(s_{i})}\left[\tilde{\cM}_{\su}(s_{i})-\frac{2\tilde{G}_{\Sf{uv}}(s_{i})}{\bar{\cD}(s_{i})}\left(\tilde{M}_{\Sf{vu}}(s_{i})-\frac{2\tilde{\cM}_{\su}(s_{i})\tilde{G}_{\Sf{uv}}(s_{i})}{\tilde{\cG}_{\su}(s_{i})}\right)\right]\frac{z_{j}w_{i}}{z_{j}^{2}-s_{i}^{2}}\right\},\label{Eq:Discrete_M_uu}\\
M_{\Sf{uv}}(z_{j})&=\sum_{i=1}^{Q}\left\{-\frac{1}{\tilde{\cG}_{\su}(s_{i})}\left[\tilde{M}_{\Sf{uv}}(s_{i})-\frac{2\tilde{G}_{\Sf{uv}}(s_{i})}{\bar{\cD}(s_{i})}\left(\tilde{\cM}_{\sv}(s_{i})-\frac{2\tilde{M}_{\Sf{uv}}(s_{i})\tilde{G}_{\Sf{uv}}(s_{i})}{\tilde{\cG}_{\su}(s_{i})}\right)\right]\frac{s_{i}w_{i}}{z_{j}^{2}-s_{i}^{2}}\right\},\\
M_{\Sf{vu}}(z_{j})&=\sum_{i=1}^{Q}\left[-\frac{1}{\bar{\cD}(s_{i})}\left(\tilde{M}_{\Sf{vu}}(s_{i})-\frac{2\tilde{\cM}_{\su}(s_{i})\tilde{G}_{\Sf{uv}}(s_{i})}{\tilde{\cG}_{\su}(s_{i})}\right)\frac{s_{i}w_{i}}{z_{j}^{2}-s_{i}^{2}}\right],\\
M_{\Sf{vv}}(z_{j})&=\frac{1}{\alpha}\sum_{i=1}^{Q}\left[-\frac{\alpha}{\bar{\cD}(s_{i})}\left(\tilde{\cM}_{\sv}(s_{i})-\frac{2\tilde{M}_{\Sf{uv}}(s_{i})\tilde{G}_{\Sf{uv}}(s_{i})}{\tilde{\cG}_{\su}(s_{i})}\right)\frac{z_{j}w_{i}}{z_{j}^{2}-s_{i}^{2}}\right]+\left(1-\frac{1}{\alpha}\right)\left(-\frac{\alpha\tilde{\cM}_{\sv}(0)}{\tilde{\cG}_{\sv}(0)z_{j}}\right).\label{Eq:Discrete_M_vv}
\end{align}
Equations \eqref{Eq:FE_hM_uu}-\eqref{Eq:FE_hM_vv} then determine
\begin{align}
\bigl\{\hat{M}_{\Sf{uu}}(z_{j}),\hat{M}_{\Sf{uv}}(z_{j}),\hat{M}_{\Sf{vu}}(z_{j}),\hat{M}_{\Sf{vv}}(z_{j})\bigr\}_{j=1}^{Q_{\Sf{\Gamma}}}.
\end{align}

The contour quadrature applied to \eqref{Eq:FE_tM_uu}-\eqref{Eq:FE_tM_vv} gives the following updates at the spectral nodes $\{s_{i}\}_{i=1}^{Q}$ and, for the $\Sf{vv}$ component, at the zero atom:
\begin{align}
x_{\Sf{uu},i}^{\Sf{new}}&=\sum_{j=1}^{Q_{\Sf{\Gamma}}}\hat{M}_{\Sf{uu}}(z_{j})\frac{z_{j}\rho_{j}}{z_{j}^{2}-s_{i}^{2}},\label{Eq:Discrete_tM_uu}\\
x_{\Sf{uv},i}^{\Sf{new}}&=\sum_{j=1}^{Q_{\Sf{\Gamma}}}\hat{M}_{\Sf{uv}}(z_{j})\frac{s_{i}\rho_{j}}{z_{j}^{2}-s_{i}^{2}},\\
x_{\Sf{vu},i}^{\Sf{new}}&=\sum_{j=1}^{Q_{\Sf{\Gamma}}}\hat{M}_{\Sf{vu}}(z_{j})\frac{s_{i}\rho_{j}}{z_{j}^{2}-s_{i}^{2}},\\
x_{\Sf{vv},i}^{\Sf{new}}&=\sum_{j=1}^{Q_{\Sf{\Gamma}}}\hat{M}_{\Sf{vv}}(z_{j})\frac{z_{j}\rho_{j}}{z_{j}^{2}-s_{i}^{2}},\\
x_{\Sf{vv},0}^{\Sf{new}}&=\sum_{j=1}^{Q_{\Sf{\Gamma}}}\hat{M}_{\Sf{vv}}(z_{j})\frac{\rho_{j}}{z_{j}}.\label{Eq:Discrete_tM_vv_0}
\end{align}

Let $\mathcal{T}(\bx)$ denote the vector formed by the updates above. At fixed scalar order parameters, this map is affine:
\begin{align}
\mathcal{T}(\bx)&=\bm{K}\bx+\bm{g},\quad\bm{g}\triangleq\mathcal{T}(\bm{0}),\quad\bm{K}_{:,i}\triangleq\mathcal{T}(\bm{e}_{i})-\mathcal{T}(\bm{0}),\quad i\in\{1,\cdots,4Q+1\},\label{Eq:Aux_Affine_Map}
\end{align}
where $\bm{e}_{i}$ is the $i$-th column of the $(4Q+1)\times(4Q+1)$ identity matrix. Hence, the fixed-point condition $\bx^{\star}=\mathcal{T}(\bx^{\star})$ is equivalent to
\begin{align}
(\bI-\bm{K})\bx^{\star}&=\bm{g}.\label{Eq:Aux_Linear_System}
\end{align}

\subsubsection{Free Entropy and Branch Selection}

For each scalar stationary solution, solving \eqref{Eq:Aux_Linear_System} gives the corresponding fixed point $\bx^{\star}$. Substituting $\bx^{\star}$ into the discrete relations above determines $\{\tilde{\cM}_{\su},\tilde{\cM}_{\sv},\bM,\hat{\bM}\}$, after which $\bar{\cN}(s_{i})$ is evaluated from \eqref{Eq:FE_b_N} at the spectral nodes $\{s_{i}\}_{i=1}^{Q}$. Together with $\{\tilde{G}_{\Sf{uv}},\tilde{\cG}_{\su},\tilde{\cG}_{\sv},\bar{\cD}\}$ computed above, these quantities provide all the function-dependent terms required in Result \ref{Result:Free_Entropy}. Let $\Psi_{\su}(\hat{m}_{\su})$ and $\Psi_{\sv}(\hat{m}_{\sv})$ denote the scalar-channel log-partition terms in Result \ref{Result:Free_Entropy}. The evaluation of these scalar terms is standard and will not be detailed further. The branch-dependent free entropy is then evaluated as
\begin{align}
\tilde{f}&=\frac{1}{2}\sum_{j=1}^{Q_{\Sf{\Gamma}}}\rho_{j}\left(\hat{M}_{\Sf{uu}}(z_{j})M_{\Sf{uu}}(z_{j})+\hat{M}_{\Sf{uv}}(z_{j})M_{\Sf{uv}}(z_{j})+\hat{M}_{\Sf{vu}}(z_{j})M_{\Sf{vu}}(z_{j})+\hat{M}_{\Sf{vv}}(z_{j})M_{\Sf{vv}}(z_{j})\right)\nonumber\\
&\quad-\frac{1}{2}\log(1-m_{\su})-\frac{\alpha}{2}\log(1-m_{\sv})-\frac{m_{\su}}{2(1-m_{\su})}-\frac{\alpha m_{\sv}}{2(1-m_{\sv})}+\Psi_{\su}(\hat{m}_{\su})+\alpha\Psi_{\sv}(\hat{m}_{\sv})\nonumber\\
&\quad-\frac{\alpha-1}{2}\left(\log\tilde{\cG}_{\sv}(0)+\frac{\tilde{m}_{\sv}-\tilde{\cM}_{\sv}^{2}(0)}{\tilde{\cG}_{\sv}(0)}\right)-\frac{1}{2}\sum_{i=1}^{Q}w_{i}\left(\log\tilde{\cG}_{\su}(s_{i})+\frac{\tilde{m}_{\su}-\tilde{\cM}_{\su}^{2}(s_{i})}{\tilde{\cG}_{\su}(s_{i})}+\frac{\bar{\cN}(s_{i})}{\bar{\cD}(s_{i})}-\frac{x_{\Sf{uv},i}^{2}}{\tilde{\cG}_{\su}(s_{i})}+\log\bar{\cD}(s_{i})\right).\label{Eq:Numerical_Free_Entropy}
\end{align}
The exact value of \eqref{Eq:Numerical_Free_Entropy} is real. A small imaginary component may remain after numerical contour integration; this component is a quadrature error, and only $\Re(\tilde{f})$ is used in the branch comparison.

Let $\tilde{f}_{\ell}$ denote the value of \eqref{Eq:Numerical_Free_Entropy} for the $\ell$-th element of $\operatorname{Sol}_{\Sf{RS}}
$. The branch corresponding to the IT limit is selected according to
\begin{align}
\ell_{\Sf{IT}}&=\mathop{\operatorname{argmax}}\limits_{\ell\in\{1,\cdots,L\}}\Re(\tilde{f}_{\ell}).
\end{align}
The branch indexed by $\ell_{\Sf{IT}}$ is reported as the IT limit, while the remaining coexisting stationary branches are reported as metastable states.


\section{Derivation of \texorpdfstring{$\varphi(Z_{0},Y)$}{$\varphi(Z_{0},Y)$}}\label{Sec:C_y}

We start from the definition of $\varphi(\bZ_{0},\bY)$ in \eqref{Eq:Partition_Func}:
\begin{align}
\varphi(\bZ_{0},\bY)&=\exp\left(-\frac{N}{2}\left(\Tr\,V(\bY\bY^{\top})-\Tr\,V(\bZ_{0}\bZ_{0}^{\top})\right)\right).
\end{align}
Applying the contour representation in \eqref{Eq:Trace_V} to the two trace terms gives
\begin{align}
\varphi(\bZ_{0},\bY)&=\exp\left(-\frac{N}{2}\left[\frac{1}{2\pi\imagUnit}\oint_{\Gamma_{\sx}}{\dd x}\,V(x^{2})\Big(\Tr\,\cG_{\tilde{\bY}}(x)-\Tr\,\cG_{\tilde{\bZ}_{0}}(x)\Big)\right]\right)\nonumber\\
&\overset{\text{(A)}}{=}\exp\left(-\frac{N}{2}\left[\frac{1}{2\pi\imagUnit}\oint_{\Gamma_{\sx}}{\dd x}\,V(x^{2})\Tr\,\frac{\lambda}{N}\cG_{\tilde{\bZ}_{0}}(x)\tilde{\bX}_{0}(\tilde{\bI}-\lambda\bC(x))^{-1}\tilde{\bX}_{0}^{\top}\cG_{\tilde{\bZ}_{0}}(x)\right]\right)\nonumber\\
&\overset{\text{(B)}}{=}\exp\left(-\frac{N}{2}\left(\frac{1}{2\pi\imagUnit}\oint_{\Gamma_{\sx}}{\dd x}\,V(x^{2})\frac{\dd}{\dd x}\log|\tilde{\bI}-\lambda\bC(x)|\right)\right)\nonumber\\
&\overset{\text{(C)}}{=}\exp\left(\frac{N}{2\pi\imagUnit}\oint_{\Gamma_{\sx}}{\dd x}\,xV'(x^{2})\log|\tilde{\bI}-\lambda\bC(x)|\right).
\end{align}
Here, (A) follows from the Woodbury matrix identity applied to the resolvent $\cG_{\tilde{\bY}}(x)$, as shown in step (G) of \eqref{Eq:Phi_Y}. (B) follows from the same logarithmic-determinant differentiation used in step (B) of \eqref{Eq:Partition_Func}, now applied to $\tilde{\bI}-\lambda\bC(x)$. (C) follows from integration by parts along the contour.

We next explain the role of the subtraction $\Tr\,V(\bZ_{0}\bZ_{0}^{\top})$. The trace $\Tr\,V(\bY\bY^{\top})$ contains $N$ eigenvalue contributions of order one and is therefore $O(N)$. Without the subtraction, the factor $\frac{N}{2}$ in the exponent would produce an $O(N^{2})$ contribution, so that the normalized free entropy $\frac{1}{N}\E\log\Omega(\bZ_{0},\bY)$ would diverge linearly with $N$. Equivalently, this is why the factor $\exp\left(\frac{N}{2}\Tr\,V(\bZ_{0}\bZ_{0}^{\top})\right)
$ must be included in the partition function \eqref{Eq:PartitionFunction}. The subtraction removes precisely this extensive noise contribution. Indeed, $\bY\bY^{\top}$ differs from $\bZ_{0}\bZ_{0}^{\top}$ only through the finite-rank signal perturbation. Under the scaling considered here, this perturbation changes the spectral trace only by an $O(1)$ amount, so that
\begin{align}
\Tr\,V(\bY\bY^{\top})-\Tr\,V(\bZ_{0}\bZ_{0}^{\top})=O(1).
\end{align}
The exponent defining $\varphi(\bZ_{0},\bY)$ is consequently $O(N)$ rather than $O(N^{2})$, giving an $O(1)$ contribution to the normalized free entropy. The same subtraction is also required for the exact relation between the free entropy and the mutual information in \eqref{Eq:MI_identity}.

It remains to evaluate this contribution in the thermodynamic limit. From \eqref{Eq:C_x}, $\bC(x)\to\bar{\bC}(x)$, where $\bar{\bC}(x)$ is the deterministic limit obtained from the asymptotic noise spectrum. Hence,
\begin{align}
\log\varphi_{\alpha,\lambda}
&\triangleq\lim_{N\to\infty}\frac{1}{N}\E_{\bu_{0},\bv_{0},\bZ_{0}}\log\varphi(\bZ_{0},\bY)\nonumber\\
&=\lim_{N\to\infty}\E_{\bu_{0},\bv_{0},\bZ_{0}}\left[\frac{1}{2\pi\imagUnit}\oint_{\Gamma_{\sx}}{\dd x}\,xV'(x^{2})\log|\tilde{\bI}-\lambda\bC(x)|\right]\nonumber\\
&\to\frac{1}{2\pi\imagUnit}\oint_{\Gamma_{\sx}}{\dd x}\,xV'(x^{2})\log|\tilde{\bI}-\lambda\bar{\bC}(x)|\nonumber\\
&=\frac{1}{2\pi\imagUnit}\oint_{\Gamma_{\sx}}{\dd x}\,xV'(x^{2})\log\left|\begin{bmatrix}0&1\\1&0\end{bmatrix}-\lambda\begin{bmatrix}c_{1}(x)&0\\0&c_{2}(x)\end{bmatrix}\right|\nonumber\\
&=\frac{1}{2\pi\imagUnit}\oint_{\Gamma_{\sx}}{\dd x}\,xV'(x^{2})\log(\lambda^{2}c_{1}(x)c_{2}(x)-1).\label{eq:logvarphi}
\end{align}
Here, $\{c_{1},c_{2}\}$ are defined in \eqref{Eq:C_1_x}-\eqref{Eq:C_2_x}, respectively. The quantity $\log\varphi_{\alpha,\lambda}$ depends only on the model parameters and the limiting noise spectrum. In particular, it does not depend on the replica order parameters introduced in the variational calculation. It therefore enters the free entropy as an additive term outside the replica extremization and generates no additional saddle point equations. Nevertheless, $\log\varphi_{\alpha,\lambda}$ must be retained when the absolute value of the free entropy is needed, in particular when the mutual information is evaluated through \eqref{Eq:MI_identity}. At fixed model parameters, the same term is shared by all stationary solutions and therefore does not affect free entropy comparisons between different fixed point branches.

\section{Derivation of Result \ref{Result:Free_Entropy}}\label{Sec:Proof_RM}

This appendix gives the detailed replica calculation leading to Result \ref{Result:Free_Entropy}. The prior-related terms $g_{\su,n}$ and $g_{\sv,n}$ are evaluated first, followed by the noise-related term $g_{\sz,n}$. Since only the replica limit $n\to0$ enters the free entropy, all quantities are expanded to first order in $n$. We begin with several algebraic formulas for RS matrices used throughout the calculation. The prior-related terms then reduce to integrals over scalar variables under the RS ansatz, while the noise-related term is evaluated using the rectangular inhomogeneous spherical integral in Result \ref{Lemma:HCIZ}.

\subsection{Algebraic Formulas for RS Matrices}\label{Sec:Property_RS}

We begin with a few algebraic formulas for RS matrices that will be used repeatedly in the derivation. Recall the $(n+1)\times(n+1)$ RS matrix $\bQ=\cQ(\beta^{(0)},m,q,\beta)$ introduced in \eqref{Eq:RS_Matrix}, which can be written as
\begin{align}
\cQ(\beta^{(0)},m,q,\beta)&\triangleq\left[\begin{array}{c|c}\beta^{(0)}&m\vec{\bm{1}}\\\hline m\bm{1}&(\beta-q)\Eye+q\bm{1}\vec{\bm{1}}\end{array}\right].
\end{align}
Here, $\beta^{(0)}$ and $\beta$ denote the self-overlaps of the ground-truth signal and a replica, respectively, while $m$ and $q$ denote the overlap between the ground-truth signal and a replica and the overlap between two distinct replicas, respectively.

The matrix $\bQ$ admits the block diagonalization \cite{shinzato2008perceptron,takahashi2022macroscopic}
\begin{align}
\bQ=\bU\left[\begin{array}{c|ccc}\beta^{(0)}&\sqrt{n}m&&\\\hline\sqrt{n}m&\beta-q+n q&&\\&&\beta-q&\\&&&\ddots\end{array}\right]\bU^{\dagger},\quad\bU\triangleq\left[\begin{array}{c|c}1&\vec{\bm{0}}\\\hline\bm{0}&\frac{1}{\sqrt{n}}\bm{E}_{n}\end{array}\right],\label{Eq:Block_Diagonal}
\end{align}
where $\bm{E}_{n}$ is the $n\times n$ discrete Fourier transform matrix. Its log-determinant therefore has the first-order expansion
\begin{align}
\log|\bQ|&=\log\beta^{(0)}+n\log(\beta-q)+n\frac{\beta^{(0)}q-m^{2}}{\beta^{(0)}(\beta-q)}+O(n^{2}).
\end{align}

The inverse $\hat{\bQ}\triangleq\bQ^{-1}$ retains the RS form and is parameterized as $\hat{\bQ}=\cQ(\hat{\beta}^{(0)}(n),\hat{m}(n),\hat{q}(n),\hat{\beta}(n))$. Expanding its parameters around $n=0$ gives \cite{shinzato2008perceptron,takahashi2022macroscopic}
\begin{align}
\hat{\beta}^{(0)}(n)&=\frac{1}{\beta^{(0)}}+n\frac{m^{2}}{(\beta^{(0)})^{2}(\beta-q)}+O(n^{2}),\\
\hat{m}(n)&=-\frac{m}{\beta^{(0)}(\beta-q)}+n\frac{m(\beta^{(0)}q-m^{2})}{(\beta^{(0)})^{2}(\beta-q)^{2}}+O(n^{2}),\\
\hat{q}(n)&=-\frac{\beta^{(0)}q-m^{2}}{\beta^{(0)}(\beta-q)^{2}}+n\frac{(\beta^{(0)}q-m^{2})^{2}}{(\beta^{(0)})^{2}(\beta-q)^{3}}+O(n^{2}),\\
\hat{\beta}(n)&=\frac{1}{\beta-q}+\hat{q}(n)+O(n^{2}).
\end{align}

We also use $(n+1)\times(n+1)$ replica coupling matrices of the form $\bC=\cC(h_{1},h_{2},h_{3})$, where
\begin{align}
\cC(h_{1},h_{2},h_{3})\triangleq\left[\begin{array}{c|ccc}0&h_{1}&\cdots\\\hline h_{2}&h_{3}&0\\\vdots&0&\ddots\end{array}\right].\label{Eq:C_Matrix}
\end{align}
The transformed matrix $\tilde{\bQ}\triangleq\bC^{\top}\bQ\bC$ also has the RS form and can be written as $\tilde{\bQ}=\cQ(\tilde{\beta}^{(0)}(n),\tilde{m}(n),\tilde{q}(n),\tilde{\beta}(n))$. Expanding these parameters around $n=0$ to first order gives
\begin{align}
\tilde{\beta}^{(0)}(n)&=n(\beta-q)h_{2}^{2}+O(n^{2}),\\
\tilde{m}(n)&=(\beta-q)h_{2}h_{3}+n(mh_{1}+qh_{3})h_{2}+O(n^{2}),\\
\tilde{q}(n)&=\beta^{(0)}h_{1}^{2}+(2mh_{1}+qh_{3})h_{3}+O(n^{2}),\\
\tilde{\beta}(n)&=\beta^{(0)}h_{1}^{2}+(2mh_{1}+\beta h_{3})h_{3}+O(n^{2}).
\end{align}

\subsection{Derivation of Prior-Related Terms $g_{\su,n}$ and $g_{\sv,n}$}\label{Sec:Proof_RM_Prior}

The calculations for $g_{\su,n}$ and $g_{\sv,n}$ have the same form, so we give the derivation for the $\su$-channel. From \eqref{Eq:g_u_def},
\begin{align}
g_{\su,n}&=\lim_{N\to\infty}\frac{1}{N}\log\int{\prod_{a=0}^{n}\dd P_{\su}(\bu_{a})}\,\prod_{0\leq a\leq a'\leq n}\delta(NQ_{\su,a,a'}-\bu_{a}^{\top}\bu_{a'}).
\end{align}

We impose the overlap constraints through the Fourier representation of the Dirac delta function \cite{shinzato2008learning,shinzato2008perceptron,takahashi2022macroscopic}. Introducing the conjugate replica matrix $\bm{\Lambda}_{\su}\in\R^{(n+1)\times(n+1)}$ gives
\begin{align}
\prod_{0\leq a\leq a'\leq n}\delta\left(NQ_{\su,a,a'}-\bu_{a}^{\top}\bu_{a'}\right)&=C\int{\dd\bm{\Lambda}_{\su}}\,\exp\left(\frac{N}{2}\Tr\bQ_{\su}\bm{\Lambda}_{\su}\right)\prod_{i=1}^{N}\exp\left(-\frac{1}{2}\vec{\bu}_{i}\bm{\Lambda}_{\su}\vec{\bu}^{\top}_{i}\right),
\end{align}
where $C$ collects the Fourier normalization factors. Each of the $(n+1)(n+2)/2$ constraints contributes a factor polynomial in $N$, so that $N^{-1}\log C=O(N^{-1}\log N)$ vanishes in the thermodynamic limit. On the other hand,  $\dd\bm{\Lambda}_{\su}$ denotes integration over the independent entries of $\bm{\Lambda}_{\su}$, and the row vector $\vec{\bu}_{i}\triangleq[u_{0,i},\cdots,u_{n,i}]$ collects the $n+1$ replica components at spin index $i$. The imaginary unit is absorbed into the conjugate variables, whose independent entries are integrated along the imaginary axes \cite{barbier2023fundamental}. Under the RS ansatz, the conjugate matrix is parameterized as $\bm{\Lambda}_{\su}=\cQ(\hat{\beta}^{(0)}_{\su},-\hat{m}_{\su},-\hat{q}_{\su},\hat{\beta}_{\su})$.

In the thermodynamic limit, the integral over the conjugate replica matrix $\bm{\Lambda}_{\su}$ is dominated by its saddle point, yielding
\begin{align}
g_{\su,n}&=\Extr_{\bm{\Lambda}_{\su}}\left(\frac{1}{2}\Tr\bQ_{\su}\bm{\Lambda}_{\su}+\log\int{\dd P_{\su}(\vec{\bu}_{i})}\,\exp\left(-\frac{1}{2}\vec{\bu}_{i}\bm{\Lambda}_{\su}\vec{\bu}_{i}^{\top}\right)\right),
\end{align}
where $P_{\su}(\vec{\bu}_{i})\triangleq\prod_{a=0}^{n}P_{\su}(u_{a,i})$. Substituting the RS parameterizations of $\bQ_{\su}$ and $\bm{\Lambda}_{\su}$ into the trace term gives
\begin{align}
\Tr\bQ_{\su}\bm{\Lambda}_{\su}&=\beta_{\su}^{(0)}\hat{\beta}_{\su}^{(0)}+n(-2m_{\su}\hat{m}_{\su}+q_{\su}\hat{q}_{\su}+\beta_{\su}\hat{\beta}_{\su})-n^{2}q_{\su}\hat{q}_{\su}.\label{Eq:Qu_Lu}
\end{align}
For the second term, expanding the quadratic form gives
\begin{align}
&\quad\log\int{\dd P_{\su}(\vec{\bu}_{i})}\,\exp\left(-\frac{1}{2}\vec{\bu}_{i}\bm{\Lambda}_{\su}\vec{\bu}_{i}^{\top}\right)\nonumber\\
&=\log\int{\dd P_{\su}(\vec{\bu}_{i})}\,\exp\left(-\frac{1}{2}\left[\hat{\beta}^{(0)}_{\su}u^{2}_{0}-2\hat{m}_{\su}u_{0}\left(\sum^{n}_{a=1}u_{a}\right)-\hat{q}_{\su}\left(\sum^{n}_{a=1}u_{a}\right)^{2}+(\hat{\beta}_{\su}+\hat{q}_{\su})\left(\sum^{n}_{a=1}u^{2}_{a}\right)\right]\right)\nonumber\\
&\overset{\text{(A)}}{=}\log\int{\dd z_{\su}\dd P_{\su}(u_{0})}\,\Normal[z_{\su}|0,1]\exp\left(-\frac{\hat{\beta}^{(0)}_{\su}}{2}u_{0}^{2}\right)\left(\int{\dd P_{\su}(u_{a})}\,\exp\left(\hat{m}_{\su}u_{0}u_{a}-\frac{\hat{\beta}_{\su}+\hat{q}_{\su}}{2}u_{a}^{2}+\sqrt{\hat{q}_{\su}}u_{a}z_{\su}\right)\right)^{n}.
\end{align}
Here, (A) follows from the Hubbard-Stratonovich transformation applied to the term $\left(\sum_{a=1}^{n}u_{a}\right)^{2}$, with $z_{\su}\sim\Normal(0,1)$. This decouples the replicas, and the integrals over $\{u_{1},\cdots,u_{n}\}$ factorize into the $n$-th power of the same scalar integral.

Combining the two terms and expanding to first order around $n=0$ gives
\begin{align}
g_{\su,n}&=\Extr_{\hat{\beta}^{(0)}_{\su}}(\bar{g}_{\su}+n\tilde{g}_{\su})+O(n^{2}),\label{Eq:g_u_n}
\end{align}
where
\begin{align}
\bar{g}_{\su}&=\frac{1}{2}\beta^{(0)}_{\su}\hat{\beta}^{(0)}_{\su}+\log Z_{\su},\\
\tilde{g}_{\su}&=\Extr_{\hat{m}_{\su},\hat{q}_{\su},\hat{\beta}_{\su}}\left(\frac{1}{2}\left(-2m_{\su}\hat{m}_{\su}+q_{\su}\hat{q}_{\su}+\beta_{\su}\hat{\beta}_{\su}\right) +\frac{1}{Z_{\su}}\E_{u_{0},z_{\su}}\left[e^{-\frac{\hat{\beta}^{(0)}_{\su}u_{0}^{2}}{2}}\log\E_{u_{a}}\exp\left(\hat{m}_{\su}u_{0}u_{a}-\frac{\hat{\beta}_{\su}+\hat{q}_{\su}}{2}u_{a}^{2}+\sqrt{\hat{q}_{\su}}u_{a}z_{\su}\right)\right]\right).\label{Eq:tg_u}
\end{align}
Here, $u_{0},u_{a}\overset{\text{i.i.d.}}{\sim}P_{\su}(\cdot)$, $z_{\su}\sim\Normal(0,1)$, and
\begin{align}
Z_{\su}&\triangleq\E_{u_{0}}\exp\left(-\frac{\hat{\beta}^{(0)}_{\su}}{2}u_{0}^{2}\right).\label{Eq:Z_u}
\end{align}

Stationarity of $\bar{g}_{\su}$ w.r.t. $\hat{\beta}^{(0)}_{\su}$ gives
\begin{align}
\beta^{(0)}_{\su}&=\frac{1}{Z_{\su}}\E_{u_{0}}\left[u_{0}^{2}\exp\left(-\frac{\hat{\beta}^{(0)}_{\su}}{2}u_{0}^{2}\right)\right]=1.\label{Eq:beta_0_u}
\end{align}
Equation \eqref{Eq:hb_u_hb_v} gives $\hat{\beta}^{(0)}_{\su}=0$, so \eqref{Eq:Z_u} reduces to $Z_{\su}=1$. Substituting these results into \eqref{Eq:g_u_n} yields the final expression for $g_{\su,n}$ in \eqref{Eq:G_SU_N}-\eqref{Eq:TG_SU}. The $\sv$-channel follows from the same calculation by replacing $\{u_{0},u_{a}\}$ with $\{v_{0},v_{a}\}\overset{\text{i.i.d.}}{\sim}P_{\sv}(\cdot)$ and the $\su$-channel parameters with their $\sv$-channel counterparts.

\subsection{Derivation of Noise-Related Term $g_{\sz,n}$}\label{Sec:Proof_RM_Noise}

We next evaluate the noise-related term $g_{\sz,n}$. To apply the rectangular inhomogeneous spherical integral in Result \ref{Lemma:HCIZ}, we first rewrite $g_{\sz,n}$ in a suitable form.

We start from the representation of $g_{\sz,n}$ in \eqref{Eq:g_z_{n}}:
\begin{align}
g_{\sz,n}&=\lim_{N\to\infty}\frac{1}{N}\log\E_{\bZ_{0}}\exp\left(-\sum_{a=1}^{n}\frac{1}{2\pi\imagUnit}\oint_{\Gamma_{\sx}}{\dd x}\,\Tr\,\left(\hat{\bG}(x)\tilde{\bX}_{a}^{\top}\cG_{\tilde{\bZ}_{0}}(x)\tilde{\bX}_{a}+\hat{\bM}^{\top}(x)\tilde{\bX}_{a}^{\top}\cG_{\tilde{\bZ}_{0}}(x)\tilde{\bX}_{0}\right)\right).
\end{align}
We write the conjugate functions in component form, with $\hat{\bG}(x)$ symmetric and $\hat{\bM}(x)$ not necessarily symmetric. Using the Schur complement, the resolvent $\cG_{\tilde{\bZ}_{0}}(x)$ can be written in block form as
\begin{align}
\hat{\bM}(x)&=\begin{bmatrix}\hat{M}_{\Sf{uu}}(x)&\hat{M}_{\Sf{uv}}(x)\\\hat{M}_{\Sf{vu}}(x)&\hat{M}_{\Sf{vv}}(x)\end{bmatrix},\quad\hat{\bG}(x)=\begin{bmatrix}\hat{G}_{\Sf{uu}}(x)&\hat{G}_{\Sf{uv}}(x)\\\hat{G}_{\Sf{uv}}(x)&\hat{G}_{\Sf{vv}}(x)\end{bmatrix},\quad\cG_{\tilde{\bZ}_{0}}(x)=\begin{bmatrix}x\cG_{\bZ_{0}\bZ_{0}^{\top}}(x^{2})&\cG_{\bZ_{0}\bZ_{0}^{\top}}(x^{2})\bZ_{0}\\\bZ_{0}^{\top}\cG_{\bZ_{0}\bZ_{0}^{\top}}(x^{2})&x\cG_{\bZ_{0}^{\top}\bZ_{0}}(x^{2})\end{bmatrix}.\label{Eq:Block_Structure}
\end{align}
Substituting these expressions gives
\begin{align}
g_{\sz,n}&=\lim_{N\to\infty}\frac{1}{N}\log\E_{\bZ_{0}}\exp\Bigg(-\sum_{a=1}^{n}\frac{1}{2\pi\imagUnit}\oint_{\Gamma_{\sx}}{\dd x}\,\Bigg(\hat{G}_{\Sf{uu}}(x)x\bu_{a}^{\top}\cG_{\bZ_{0}\bZ_{0}^{\top}}(x^{2})\bu_{a}+\hat{M}_{\Sf{uu}}(x)x\bu_{a}^{\top}\cG_{\bZ_{0}\bZ_{0}^{\top}}(x^{2})\bu_{0}\nonumber\\
&\quad+\hat{G}_{\Sf{vv}}(x)x\bv_{a}^{\top}\cG_{\bZ_{0}^{\top}\bZ_{0}}(x^{2})\bv_{a}+\hat{M}_{\Sf{vv}}(x)x\bv_{a}^{\top}\cG_{\bZ_{0}^{\top}\bZ_{0}}(x^{2})\bv_{0}+2\hat{G}_{\Sf{uv}}(x)\bu_{a}^{\top}\cG_{\bZ_{0}\bZ_{0}^{\top}}(x^{2})\bZ_{0}\bv_{a}\nonumber\\
&\quad+\hat{M}_{\Sf{vu}}(x)\bu_{0}^{\top}\cG_{\bZ_{0}\bZ_{0}^{\top}}(x^{2})\bZ_{0}\bv_{a}+\hat{M}_{\Sf{uv}}(x)\bu_{a}^{\top}\cG_{\bZ_{0}\bZ_{0}^{\top}}(x^{2})\bZ_{0}\bv_{0}\Bigg)\Bigg).
\end{align}

We evaluate the noise average with the singular values $\bs$ held fixed, assuming that averaging over $\bs$ does not change the asymptotic free entropy \cite{barbier2023fundamental,barbier2025information}. Using the SVD $\bZ_{0}=\bP^{\top}[\Diag(\bs),\bm{0}]\bQ$, the remaining expectation is taken over the independent Haar orthogonal matrices $\bP\sim\Sf{Unif}(\mathbb{O}(N))$ and $\bQ\sim\Sf{Unif}(\mathbb{O}(M))$. Hence,
\begin{align}
g_{\sz,n}&=\lim_{N\to\infty}\frac{1}{N}\log\E_{\bP,\bQ}\exp\Bigg(-\sum_{a=1}^{n}\frac{1}{2\pi\imagUnit}\oint_{\Gamma_{\sx}}{\dd x}\Bigg(\hat{G}_{\Sf{uu}}(x)x\bu_{a}^{\top}\bP^{\top}\Diag\left(\frac{1}{x^{2}-\bs^{2}}\right)\bP\bu_{a}+\hat{M}_{\Sf{uu}}(x)x\bu_{a}^{\top}\bP^{\top}\Diag\left(\frac{1}{x^{2}-\bs^{2}}\right)\bP\bu_{0}\nonumber\\
&\quad+\hat{G}_{\Sf{vv}}(x)x\bv_{a}^{\top}\bQ^{\top}\begin{bmatrix}\Diag\left(\frac{1}{x^{2}-\bs^{2}}\right)&\bm{0}\\\bm{0}&\frac{1}{x^{2}}\Eye\end{bmatrix}\bQ\bv_{a}+\hat{M}_{\Sf{vv}}(x)x\bv_{a}^{\top}\bQ^{\top}\begin{bmatrix}\Diag\left(\frac{1}{x^{2}-\bs^{2}}\right)&\bm{0}\\\bm{0}&\frac{1}{x^{2}}\Eye\end{bmatrix}\bQ\bv_{0}\nonumber\\
&\quad+2\hat{G}_{\Sf{uv}}(x)\bu_{a}^{\top}\bP^{\top}\left[\Diag\left(\frac{\bs}{x^{2}-\bs^{2}}\right),\bm{0}\right]\bQ\bv_{a}+\hat{M}_{\Sf{vu}}(x)\bu_{0}^{\top}\bP^{\top}\left[\Diag\left(\frac{\bs}{x^{2}-\bs^{2}}\right),\bm{0}\right]\bQ\bv_{a}\nonumber\\
&\quad+\hat{M}_{\Sf{uv}}(x)\bu_{a}^{\top}\bP^{\top}\left[\Diag\left(\frac{\bs}{x^{2}-\bs^{2}}\right),\bm{0}\right]\bQ\bv_{0}\Bigg)\Bigg)\nonumber\\
&=\lim_{N\to\infty}\frac{1}{N}\log\E_{\bP,\bQ}\exp\Bigg(-\sum_{a=1}^{n}\Bigg(\bu_{a}^{\top}\bP^{\top}\Diag(\tilde{G}_{\Sf{uu}}(\bs))\bP\bu_{a}+\bu_{a}^{\top}\bP^{\top}\Diag(\tilde{M}_{\Sf{uu}}(\bs))\bP\bu_{0}\nonumber\\
&\quad+\bv_{a}^{\top}\bQ^{\top}\begin{bmatrix}\Diag(\tilde{G}_{\Sf{vv}}(\bs))&\bm{0}\\\bm{0}&\tilde{G}_{\Sf{vv}}(0)\Eye\end{bmatrix}\bQ\bv_{a}+\bv_{a}^{\top}\bQ^{\top}\begin{bmatrix}\Diag(\tilde{M}_{\Sf{vv}}(\bs))&\bm{0}\\\bm{0}&\tilde{M}_{\Sf{vv}}(0)\Eye\end{bmatrix}\bQ\bv_{0}\nonumber\\
&\quad+2\bu_{a}^{\top}\bP^{\top}[\Diag(\tilde{G}_{\Sf{uv}}(\bs)),\bm{0}]\bQ\bv_{a}+\bu_{0}^{\top}\bP^{\top}[\Diag(\tilde{M}_{\Sf{vu}}(\bs)),\bm{0}]\bQ\bv_{a}+\bu_{a}^{\top}\bP^{\top}[\Diag(\tilde{M}_{\Sf{uv}}(\bs)),\bm{0}]\bQ\bv_{0}\Bigg)\Bigg),\label{Eq:g_Z_n}
\end{align}
where the functions are defined via the following contour integrals:
\begin{align}
\tilde{M}_{\Sf{uu}}(z)&\triangleq\frac{1}{2\pi\imagUnit}\oint_{\Gamma_{\sx}}{\dd x}\,\hat{M}_{\Sf{uu}}(x)\frac{x}{x^{2}-z^{2}},\label{Eq:tM_uu}\\
\tilde{M}_{\Sf{uv}}(z)&\triangleq\frac{1}{2\pi\imagUnit}\oint_{\Gamma_{\sx}}{\dd x}\,\hat{M}_{\Sf{uv}}(x)\frac{z}{x^{2}-z^{2}},\label{Eq:tM_uv}\\
\tilde{M}_{\Sf{vu}}(z)&\triangleq\frac{1}{2\pi\imagUnit}\oint_{\Gamma_{\sx}}{\dd x}\,\hat{M}_{\Sf{vu}}(x)\frac{z}{x^{2}-z^{2}},\label{Eq:tM_vu}\\
\tilde{M}_{\Sf{vv}}(z)&\triangleq\frac{1}{2\pi\imagUnit}\oint_{\Gamma_{\sx}}{\dd x}\,\hat{M}_{\Sf{vv}}(x)\frac{x}{x^{2}-z^{2}},\label{Eq:tM_vv}\\
\tilde{G}_{\Sf{uu}}(z)&\triangleq\frac{1}{2\pi\imagUnit}\oint_{\Gamma_{\sx}}{\dd x}\,\hat{G}_{\Sf{uu}}(x)\frac{x}{x^{2}-z^{2}},\label{Eq:tG_uu}\\
\tilde{G}_{\Sf{uv}}(z)&\triangleq\frac{1}{2\pi\imagUnit}\oint_{\Gamma_{\sx}}{\dd x}\,\hat{G}_{\Sf{uv}}(x)\frac{z}{x^{2}-z^{2}},\label{Eq:tG_uv}\\
\tilde{G}_{\Sf{vv}}(z)&\triangleq\frac{1}{2\pi\imagUnit}\oint_{\Gamma_{\sx}}{\dd x}\,\hat{G}_{\Sf{vv}}(x)\frac{x}{x^{2}-z^{2}}.\label{Eq:tG_vv}
\end{align}

As shown in the main text, matching \eqref{Eq:g_Z_n} with the general form in \eqref{Eq:HCIZ_Definition} identifies the nonzero matrices $\bA_{a,a'}(\bs)$, $\bB_{a,a'}(\bs)$, and $\bC_{a,a'}(\bs)$ in \eqref{Eq:A_aa}-\eqref{Eq:C_aa}. Collecting these entries for each spin index gives the matrices $\bA(s_{i})$, $\bB(s_{i})$, and $\bC(s_{i})$ in \eqref{Eq:A_i}-\eqref{Eq:C_i}. With these identifications, Result \ref{Lemma:HCIZ} applied to \eqref{Eq:g_Z_n} gives
\begin{align}
g_{\sz,n}&=\Extr_{\tilde{\bQ}_{\su},\tilde{\bQ}_{\sv}}\Bigg(\frac{1}{2}\Tr\bQ_{\su}\tilde{\bQ}_{\su}+\frac{\alpha}{2}\Tr\bQ_{\sv}\tilde{\bQ}_{\sv}-\frac{1}{2N}\sum_{i=1}^{N}\log|\tilde{\bQ}_{\su}-2\bA(s_{i})|-\frac{1}{2N}\sum_{i=N+1}^{M}\log|\tilde{\bQ}_{\sv}-2\bB(s_{i})|\nonumber\\
&\quad-\frac{1}{2N}\sum_{i=1}^{N}\log|(\tilde{\bQ}_{\sv}-2\bB(s_{i}))-\bC^{\top}(s_{i})(\tilde{\bQ}_{\su}-2\bA(s_{i}))^{-1}\bC(s_{i})|\Bigg)-\frac{1}{2}\log|\bQ_{\su}|-\frac{\alpha}{2}\log|\bQ_{\sv}| -\frac{(n+1)(1+\alpha)}{2}.\label{Eq:tI_HCIZ}
\end{align}
Here, $\bQ_{\su}$ and $\bQ_{\sv}$ are defined in \eqref{Eq:RS_Matrices}, and the corresponding conjugate RS matrices are parameterized as $\tilde{\bQ}_{\su}=\cQ(\tilde{\beta}^{(0)}_{\su},\tilde{m}_{\su},\tilde{q}_{\su},\tilde{\beta}_{\su})$ and $\tilde{\bQ}_{\sv}=\cQ(\tilde{\beta}^{(0)}_{\sv},\tilde{m}_{\sv},\tilde{q}_{\sv},\tilde{\beta}_{\sv})$, as shown in \eqref{Eq:tQ_u_tQ_v}.

Using the convergence of the empirical distribution of $\{s_{i}^{2}\}_{i=1}^{N}$ to $\mu(\cdot)$, \eqref{Eq:tI_HCIZ} becomes
\begin{align}
g_{\sz,n}&=\Extr_{\tilde{\bQ}_{\su},\tilde{\bQ}_{\sv}}\Bigg(\frac{1}{2}\Tr\bQ_{\su}\tilde{\bQ}_{\su}+\frac{\alpha}{2}\Tr\bQ_{\sv}\tilde{\bQ}_{\sv}-\frac{1}{2}\E_{s^{2}\sim\mu}\log|\tilde{\bQ}_{\su}-2\bA(s)|-\frac{\alpha-1}{2}\log|\tilde{\bQ}_{\sv}-2\bB(0)|\nonumber\\
&\quad-\frac{1}{2}\E_{s^{2}\sim\mu}\log|(\tilde{\bQ}_{\sv}-2\bB(s))-\bC^{\top}(s)(\tilde{\bQ}_{\su}-2\bA(s))^{-1}\bC(s)|\Bigg)-\frac{1}{2}\log|\bQ_{\su}|-\frac{\alpha}{2}\log|\bQ_{\sv}| -\frac{(n+1)(1+\alpha)}{2}.\label{Eq:gz_n}
\end{align}
All terms in \eqref{Eq:gz_n} except the coupled log-determinant can be expanded directly using the formulas for RS matrices in Appendix \ref{Sec:Property_RS}. We therefore focus on
\begin{align}
\log|(\tilde{\bQ}_{\sv}-2\bB(s))-\bC^{\top}(s)(\tilde{\bQ}_{\su}-2\bA(s))^{-1}\bC(s)|.
\end{align}

To expand the coupled log-determinant, introduce the four matrices
\begin{align}
\bM_{1}&\triangleq\tilde{\bQ}_{\su}-2\bA(s),\quad\bM_{2}\triangleq\bM_{1}^{-1},\quad\bM_{3}\triangleq\bC^{\top}(s)\bM_{2}\bC(s),\quad\bM_{4}\triangleq(\tilde{\bQ}_{\sv}-2\bB(s))-\bM_{3}.
\end{align}
All four matrices retain the RS structure. The matrix $\bM_{1}$ is RS by construction; $\bM_{2}$ remains RS under inversion; $\bM_{3}$ has the RS form by the transformation $\bC^{\top}\bQ\bC$ derived in Appendix \ref{Sec:Property_RS}; and $\bM_{4}$ is again RS as the difference of two RS matrices. We therefore parameterize
\begin{align}
\bM_{k}&=\cQ(\beta^{(0)}_{k}(n),m_{k}(n),q_{k}(n),\beta_{k}(n)),\quad k\in\{1,2,3,4\}.
\end{align}

For convenience, introduce
\begin{align}
\bar{\cN}(z)&\triangleq\tilde{\beta}^{(0)}_{\sv}\left[\tilde{q}_{\sv}-\frac{\tilde{M}_{\Sf{vu}}^{2}(z)}{\tilde{\beta}^{(0)}_{\su}}+2\tilde{G}_{\Sf{uv}}(z)\left(\frac{2\tilde{\cM}_{\su}(z)\tilde{M}_{\Sf{vu}}(z)}{\tilde{\beta}^{(0)}_{\su}\tilde{\cG}_{\su}(z)}+\frac{2(\tilde{\beta}^{(0)}_{\su}\tilde{q}_{\su}-\tilde{\cM}_{\su}^{2}(z))\tilde{G}_{\Sf{uv}}(z)}{\tilde{\beta}^{(0)}_{\su}\tilde{\cG}_{\su}^{2}(z)}\right)\right]\nonumber\\
&\quad-\left(\tilde{\cM}_{\sv}(z)-\frac{2\tilde{M}_{\Sf{uv}}(z)\tilde{G}_{\Sf{uv}}(z)}{\tilde{\cG}_{\su}(z)}\right)^{2},\label{Eq:b_N}\\
\bar{\cD}(z)&\triangleq\tilde{\beta}^{(0)}_{\sv}\left(\tilde{\cG}_{\sv}(z)-\frac{4\tilde{G}_{\Sf{uv}}^{2}(z)}{\tilde{\cG}_{\su}(z)}\right),\label{Eq:b_D}\\
\tilde{\cM}_{\su}(z)&\triangleq\tilde{m}_{\su}+\tilde{M}_{\Sf{uu}}(z),\label{Eq:t_CM_u}\\
\tilde{\cM}_{\sv}(z)&\triangleq\tilde{m}_{\sv}+\tilde{M}_{\Sf{vv}}(z),\label{Eq:t_CM_v}\\
\tilde{\cG}_{\su}(z)&\triangleq\tilde{\beta}_{\su}-\tilde{q}_{\su}+2\tilde{G}_{\Sf{uu}}(z),\label{Eq:t_CG_u}\\
\tilde{\cG}_{\sv}(z)&\triangleq\tilde{\beta}_{\sv}-\tilde{q}_{\sv}+2\tilde{G}_{\Sf{vv}}(z).\label{Eq:t_CG_v}
\end{align}

We first determine the RS parameters of $\bM_{1}$. Comparing $\bM_{1}=\tilde{\bQ}_{\su}-2\bA(s)$ with the form $\cQ(\beta^{(0)},m,q,\beta)$ gives
\begin{align}
\beta^{(0)}_{1}&=\tilde{\beta}^{(0)}_{\su},\quad m_{1}=\tilde{\cM}_{\su}(s),\quad q_{1}=\tilde{q}_{\su},\quad\beta_{1}=\tilde{\beta}_{\su}+2\tilde{G}_{\Sf{uu}}(s),\quad\beta_{1}-q_{1}=\tilde{\cG}_{\su}(s).\label{Eq:M1_beta_q}
\end{align}

We next consider $\bM_{2}=\bM_{1}^{-1}$. Applying the inverse formulas for an RS matrix gives
\begin{align}
\beta^{(0)}_{2}(n)&=\frac{1}{\beta^{(0)}_{1}}+n\frac{m^{2}_{1}}{(\beta^{(0)}_{1})^{2}(\beta_{1}-q_{1})}+O(n^{2}),\nonumber\\
m_{2}(n)&=-\frac{m_{1}}{\beta^{(0)}_{1}(\beta_{1}-q_{1})}+n\frac{m_{1}(\beta^{(0)}_{1}q_{1}-m^{2}_{1})}{(\beta^{(0)}_{1})^{2}(\beta_{1}-q_{1})^{2}}+O(n^{2}),\nonumber\\
q_{2}(n)&=-\frac{\beta^{(0)}_{1}q_{1}-m^{2}_{1}}{\beta^{(0)}_{1}(\beta_{1}-q_{1})^{2}}+n\frac{(\beta^{(0)}_{1}q_{1}-m^{2}_{1})^{2}}{(\beta^{(0)}_{1})^{2}(\beta_{1}-q_{1})^{3}}+O(n^{2}),\nonumber\\
\beta_{2}(n)&=\frac{1}{\beta_{1}-q_{1}}+q_{2}(n)+O(n^{2}).
\end{align}
Substituting the parameters of $\bM_{1}$ and setting $n=0$ gives
\begin{align}
\beta^{(0)}_{2}(0)=\frac{1}{\tilde{\beta}^{(0)}_{\su}},\quad m_{2}(0)=-\frac{\tilde{\cM}_{\su}(s)}{\tilde{\beta}^{(0)}_{\su}\tilde{\cG}_{\su}(s)},\quad q_{2}(0)=-\frac{\tilde{\beta}^{(0)}_{\su}\tilde{q}_{\su}-\tilde{\cM}_{\su}^{2}(s)}{\tilde{\beta}^{(0)}_{\su}\tilde{\cG}_{\su}^{2}(s)},\quad\beta_{2}(0)-q_{2}(0)=\frac{1}{\tilde{\cG}_{\su}(s)}.\label{Eq:M2_Zero}
\end{align}

For $\bM_{3}=\bC^{\top}(s)\bM_{2}\bC(s)$, comparison of $\bC(s)$ with the form in \eqref{Eq:C_Matrix} gives
\begin{align}
h_{\Sf{1}}=-\tilde{M}_{\Sf{vu}}(s),\quad h_{\Sf{2}}=-\tilde{M}_{\Sf{uv}}(s),\quad h_{\Sf{3}}=-2\tilde{G}_{\Sf{uv}}(s).
\end{align}
Using these values in the RS transformation formulas for $\bC^{\top}\bQ\bC$ in Appendix \ref{Sec:Property_RS}, we obtain
\begin{align}
\beta^{(0)}_{3}(n)&=n(\beta_{2}(0)-q_{2}(0))\tilde{M}^{2}_{\Sf{uv}}(s)+O(n^{2}),\\
m_{3}(n)&=2(\beta_{2}(n)-q_{2}(n))\tilde{M}_{\Sf{uv}}(s)\tilde{G}_{\Sf{uv}}(s)+n(m_{2}(0)\tilde{M}_{\Sf{vu}}(s)+2q_{2}(0)\tilde{G}_{\Sf{uv}}(s))\tilde{M}_{\Sf{uv}}(s)+O(n^{2}),\\
q_{3}(n)&=\beta^{(0)}_{2}(n)\tilde{M}^{2}_{\Sf{vu}}(s)+2\tilde{G}_{\Sf{uv}}(s)(2m_{2}(n)\tilde{M}_{\Sf{vu}}(s)+2q_{2}(n)\tilde{G}_{\Sf{uv}}(s)),\\
\beta_{3}(n)&=\beta^{(0)}_{2}(n)\tilde{M}^{2}_{\Sf{vu}}(s)+2\tilde{G}_{\Sf{uv}}(s)(2m_{2}(n)\tilde{M}_{\Sf{vu}}(s)+2\beta_{2}(n)\tilde{G}_{\Sf{uv}}(s)).
\end{align}
Using \eqref{Eq:M2_Zero}, the quantities needed below reduce to
\begin{align}
\beta^{(0)}_{3}(n)&=n\frac{\tilde{M}^{2}_{\Sf{uv}}(s)}{\tilde{\cG}_{\su}(s)}+O(n^{2}),\label{Eq:M3_beta0}\\
m_{3}(0)&=\frac{2\tilde{M}_{\Sf{uv}}(s)\tilde{G}_{\Sf{uv}}(s)}{\tilde{\cG}_{\su}(s)},\label{Eq:M3_m}\\
q_{3}(0)&=\frac{\tilde{M}_{\Sf{vu}}^{2}(s)}{\tilde{\beta}^{(0)}_{\su}}-2\tilde{G}_{\Sf{uv}}(s)\left(\frac{2\tilde{\cM}_{\su}(s)\tilde{M}_{\Sf{vu}}(s)}{\tilde{\beta}^{(0)}_{\su}\tilde{\cG}_{\su}(s)}+\frac{2(\tilde{\beta}^{(0)}_{\su}\tilde{q}_{\su}-\tilde{\cM}_{\su}^{2}(s))\tilde{G}_{\Sf{uv}}(s)}{\tilde{\beta}^{(0)}_{\su}\tilde{\cG}_{\su}^{2}(s)}\right),\label{Eq:M3_q}\\
\beta_{3}(0)-q_{3}(0)&=4(\beta_{2}(0)-q_{2}(0))\tilde{G}^{2}_{\Sf{uv}}(s)=\frac{4\tilde{G}^{2}_{\Sf{uv}}(s)}{\tilde{\cG}_{\su}(s)}.\label{Eq:M3_beta_q}
\end{align}

The parameters of $\bM_{4}=(\tilde{\bQ}_{\sv}-2\bB(s))-\bM_{3}$ follow by direct subtraction:
\begin{align}
\beta^{(0)}_{4}(n)&=\tilde{\beta}^{(0)}_{\sv}-\beta^{(0)}_{3}(n)+O(n^{2}),\nonumber\\
m_{4}(n)&=\tilde{m}_{\sv}+\tilde{M}_{\Sf{vv}}(s)-m_{3}(n)+O(n^{2}),\nonumber\\
q_{4}(n)&=\tilde{q}_{\sv}-q_{3}(n)+O(n^{2}),\nonumber\\
\beta_{4}(n)&=\tilde{\beta}_{\sv}+2\tilde{G}_{\Sf{vv}}(s)-\beta_{3}(n)+O(n^{2}).
\end{align}
Substituting \eqref{Eq:M3_beta0}-\eqref{Eq:M3_beta_q} gives
\begin{align}
\beta^{(0)}_{4}(n)&=\tilde{\beta}^{(0)}_{\sv}-n\frac{\tilde{M}^{2}_{\Sf{uv}}(s)}{\tilde{\cG}_{\su}(s)}+O(n^{2}),\label{Eq:M4_beta0}\\
m_{4}(0)&=\tilde{\cM}_{\sv}(s)-\frac{2\tilde{M}_{\Sf{uv}}(s)\tilde{G}_{\Sf{uv}}(s)}{\tilde{\cG}_{\su}(s)},\label{Eq:M4_m}\\
q_{4}(0)&=\tilde{q}_{\sv}-\frac{\tilde{M}_{\Sf{vu}}^{2}(s)}{\tilde{\beta}^{(0)}_{\su}}+2\tilde{G}_{\Sf{uv}}(s)\left(\frac{2\tilde{\cM}_{\su}(s)\tilde{M}_{\Sf{vu}}(s)}{\tilde{\beta}^{(0)}_{\su}\tilde{\cG}_{\su}(s)}+\frac{2(\tilde{\beta}^{(0)}_{\su}\tilde{q}_{\su}-\tilde{\cM}_{\su}^{2}(s))\tilde{G}_{\Sf{uv}}(s)}{\tilde{\beta}^{(0)}_{\su}\tilde{\cG}_{\su}^{2}(s)}\right),\label{Eq:M4_q}\\
\beta_{4}(0)-q_{4}(0)&=\tilde{\cG}_{\sv}(s)-\frac{4\tilde{G}^{2}_{\Sf{uv}}(s)}{\tilde{\cG}_{\su}(s)}=\frac{\bar{\cD}(s)}{\tilde{\beta}^{(0)}_{\sv}}.\label{Eq:M4_beta_q}
\end{align}
Since $\bM_{4}$ has the RS form, applying the log-determinant formula gives
\begin{align}
\log|\bM_{4}|&\overset{\text{(A)}}{=}\log\beta^{(0)}_{4}(n)+n\log(\beta_{4}(0)-q_{4}(0))+n\frac{\beta^{(0)}_{4}(0)q_{4}(0)-m^{2}_{4}(0)}{\beta^{(0)}_{4}(0)(\beta_{4}(0)-q_{4}(0))}+O(n^{2})\nonumber\\
&\overset{\text{(B)}}{=}\log\left(\tilde{\beta}^{(0)}_{\sv}-n\frac{\tilde{M}^{2}_{\Sf{uv}}(s)}{\tilde{\cG}_{\su}(s)}\right)+n\log\frac{\bar{\cD}(s)}{\tilde{\beta}^{(0)}_{\sv}}+n\frac{\bar{\cN}(s)}{\bar{\cD}(s)}+O(n^{2})\nonumber\\
&\overset{\text{(C)}}{=}\log\tilde{\beta}^{(0)}_{\sv}-n\frac{\tilde{M}^{2}_{\Sf{uv}}(s)}{\tilde{\beta}^{(0)}_{\sv}\tilde{\cG}_{\su}(s)}+n\log\frac{\bar{\cD}(s)}{\tilde{\beta}^{(0)}_{\sv}}+n\frac{\bar{\cN}(s)}{\bar{\cD}(s)}+O(n^{2}).\label{Eq:M4_LogDet_Expansion}
\end{align}
In (A), the last two terms already carry a factor $n$, so only the corresponding parameters at $n=0$ are required. In (B), we substitute \eqref{Eq:M4_beta0}-\eqref{Eq:M4_beta_q}, together with $\beta^{(0)}_{4}(0)=\tilde{\beta}^{(0)}_{\sv}$ and the definitions of $\{\bar{\cN},\bar{\cD}\}$ in \eqref{Eq:b_N}-\eqref{Eq:b_D}. In (C), the first logarithm is expanded to first order around $n=0$.

Combining \eqref{Eq:M4_LogDet_Expansion} with the remaining terms in \eqref{Eq:gz_n} and collecting the constant and first-order terms in $n$ gives
\begin{align}
g_{\sz,n}&=\Extr_{\tilde{\beta}^{(0)}_{\su},\tilde{\beta}^{(0)}_{\sv}}(\bar{g}_{\sz}+n\tilde{g}_{\sz})+O(n^{2}),
\end{align}
where
\begin{align}
\bar{g}_{\sz}&=\frac{1}{2}\beta^{(0)}_{\su}\tilde{\beta}^{(0)}_{\su}+\frac{\alpha}{2}\beta^{(0)}_{\sv}\tilde{\beta}^{(0)}_{\sv}-\frac{1}{2}\log\tilde{\beta}^{(0)}_{\su}-\frac{\alpha}{2}\log\tilde{\beta}^{(0)}_{\sv}-\frac{1}{2}\log \beta_{\su}^{(0)}-\frac{\alpha}{2}\log\beta_{\sv}^{(0)}-\frac{1+\alpha}{2},\\
\tilde{g}_{\sz}&=\Extr_{\tilde{m}_{\su},\tilde{q}_{\su},\tilde{\beta}_{\su},\tilde{m}_{\sv},\tilde{q}_{\sv},\tilde{\beta}_{\sv}}\Bigg(\frac{1}{2}(2m_{\su}\tilde{m}_{\su}-q_{\su}\tilde{q}_{\su}+\beta_{\su}\tilde{\beta}_{\su})+\frac{\alpha}{2}(2m_{\sv}\tilde{m}_{\sv}-q_{\sv}\tilde{q}_{\sv}+\beta_{\sv}\tilde{\beta}_{\sv})-\frac{\alpha-1}{2}\left(\log\tilde{\cG}_{\sv}(0)+\frac{\tilde{\beta}^{(0)}_{\sv}\tilde{q}_{\sv}-\tilde{\cM}_{\sv}^{2}(0)}{\tilde{\beta}^{(0)}_{\sv}\tilde{\cG}_{\sv}(0)}\right)\nonumber\\
&\quad-\frac{1}{2}\E_{s^{2}\sim\mu}\left[\log\tilde{\cG}_{\su}(s)+\frac{\tilde{\beta}^{(0)}_{\su}\tilde{q}_{\su}-\tilde{\cM}_{\su}^{2}(s)}{\tilde{\beta}^{(0)}_{\su}\tilde{\cG}_{\su}(s)}+\frac{\bar{\cN}(s)}{\bar{\cD}(s)}-\frac{\tilde{M}_{\Sf{uv}}^{2}(s)}{\tilde{\beta}^{(0)}_{\sv}\tilde{\cG}_{\su}(s)}++\log\frac{\bar{\cD}(s)}{\tilde{\beta}^{(0)}_{\sv}}\right]\Bigg)\nonumber\\
&\quad-\frac{1}{2}\left(\log(\beta_{\su}-q_{\su})+\frac{\beta_{\su}^{(0)}q_{\su}-m_{\su}^{2}}{\beta_{\su}^{(0)}(\beta_{\su}-q_{\su})}\right)-\frac{\alpha}{2}\left(\log(\beta_{\sv}-q_{\sv})+\frac{\beta_{\sv}^{(0)}q_{\sv}-m_{\sv}^{2}}{\beta_{\sv}^{(0)}(\beta_{\sv}-q_{\sv})}\right) -\frac{1+\alpha}{2}.
\end{align}

Stationarity of $\bar{g}_{\sz}$ w.r.t. $\{\tilde{\beta}^{(0)}_{\su},\tilde{\beta}^{(0)}_{\sv}\}$ gives
\begin{align}
\tilde{\beta}^{(0)}_{\su}=\frac{1}{\beta^{(0)}_{\su}}=1,\quad\tilde{\beta}^{(0)}_{\sv}=\frac{1}{\beta^{(0)}_{\sv}}=1,\label{Eq:tb_u_tb_v}
\end{align}
where $\beta^{(0)}_{\su}=\beta^{(0)}_{\sv}=1$ has been used. Substituting these relations into $\bar{g}_{\sz}$ and $\tilde{g}_{\sz}$ gives the final form of $g_{\sz,n}$ in \eqref{Eq:g_sz_n}-\eqref{Eq:tb_u_0_tb_v_0}.

\subsection{Evaluation of the Replica Limit}\label{Sec:Completion}

The four expansions derived above provide the $O(1)$ and $O(n)$ contributions to the four constituent terms in \eqref{Eq:f_n}. Their combination gives
\begin{align}
f_{n}&=\Extr(n\tilde{g}+g_{\sz,n}+g_{\su,n}+\alpha g_{\sv,n})\nonumber\\
&=\Extr_{\beta_{\su}^{(0)},\beta_{\sv}^{(0)}}(\bar{f}+n\tilde{f})+O(n^{2}),
\end{align}
where
\begin{align}
\bar{f}&=\bar{g}_{\su}+\alpha\bar{g}_{\sv}+\bar{g}_{\sz},\\
\tilde{f}&=\Extr\left(\tilde{g}+\tilde{g}_{\su}+\alpha\tilde{g}_{\sv}+\tilde{g}_{\sz}\right).
\end{align}

Stationarity of $\bar{f}$ w.r.t. $\{\beta^{(0)}_{\su},\beta^{(0)}_{\sv}\}$ then gives
\begin{align}
0=\hat{\beta}^{(0)}_{\su}+\tilde{\beta}^{(0)}_{\su}-\frac{1}{\beta^{(0)}_{\su}},\quad0=\hat{\beta}^{(0)}_{\sv}+\tilde{\beta}^{(0)}_{\sv}-\frac{1}{\beta^{(0)}_{\sv}}.
\end{align}
Equation \eqref{Eq:tb_u_tb_v} gives $\tilde{\beta}^{(0)}_{\su}=\frac{1}{\beta^{(0)}_{\su}}=1$ and $\tilde{\beta}^{(0)}_{\sv}=\frac{1}{\beta^{(0)}_{\sv}}=1$. These relations reduce the stationarity conditions to
\begin{align}
\hat{\beta}^{(0)}_{\su}=0,\quad\hat{\beta}^{(0)}_{\sv}=0.\label{Eq:hb_u_hb_v}
\end{align}

We finally return to the replica representation of the free entropy in \eqref{Eq:f_n_Variational}. The definition of $f_{n}$ implies $\left.f_{n}\right|_{n=0}=0$, and hence $\lim_{n\to0}\frac{1}{n}f_{n}$ takes the indeterminate form $\frac{0}{0}$. L'Hôpital's rule expresses this limit in terms of $\left.\frac{\partial f_{n}}{\partial n}\right|_{n=0}$, while the first-order expansion above identifies this derivative with $\tilde{f}$. Therefore,
\begin{align}
f&=\lim_{n\to0}\frac{1}{n}f_n+\log\varphi_{\alpha,\lambda}=\tilde{f}+\log\varphi_{\alpha,\lambda}.
\end{align}
We then insert the explicit expression for $\tilde{f}$ into this result and obtain the intermediate free entropy representation stated in \eqref{Eq:Intermediate_Free_Entropy}.

\section{Derivation of the Rectangular Inhomogeneous Spherical Integral}\label{app:spherical_int}

We derive the variational representation stated in Result \ref{Lemma:HCIZ}. Throughout the derivation, we use the notation introduced there. Thus, $\bP\sim\Sf{Unif}(\mathbb{O}(N))$ and $\bQ\sim\Sf{Unif}(\mathbb{O}(M))$ are independent Haar orthogonal matrices, while $\bX\in\R^{N\times n}$ and $\bY\in\R^{M\times n}$ collect the replicas in the $\sx$- and $\sy$-channels, respectively. Their replica matrices are $\bQ_{\sx}=\frac{1}{N}\bX^{\top}\bX$ and $\bQ_{\sy}=\frac{1}{M}\bY^{\top}\bY$.

For convenience, we recall from \eqref{Aaa_Baa_Caa}-\eqref{Eq:Ai_Bi_Ci} the two representations of the replica coupling arrays $\bA$, $\bB$, and $\bC$ used below:
\begin{alignat}{3}
\bA_{a,a'}&=\Diag([A_{i,a,a'}]_{i\leq N}),\quad&\bB_{a,a'}&=\Diag([B_{i,a,a'}]_{i\leq M}),\quad&\bC_{a,a'}&=\Diag([C_{i,a,a'}]_{i\leq N}),\\
\bA_{i}&=[A_{i,a,a'}]_{a,a'\leq n},\quad&\bB_{i}&=[B_{i,a,a'}]_{a,a'\leq n},\quad&\bC_{i}&=[C_{i,a,a'}]_{a,a'\leq n}.
\end{alignat}

Recall from \eqref{Eq:HCIZ_Definition} that the rectangular inhomogeneous spherical integral is
\begin{align}
g_{\sh,n}&\triangleq\lim_{N\to\infty}\frac{1}{N}\log\E_{\bP,\bQ}\exp\left(\sum_{1\leq a,a'\leq n}\bx_{a}^{\top}\bP^{\top}\bA_{a,a'}\bP\bx_{a'}+\by_{a}^{\top}\bQ^{\top}\bB_{a,a'}\bQ\by_{a'}+\bx_{a}^{\top}\bP^{\top}[\bC_{a,a'},\bm{0}]\bQ\by_{a'}\right).\label{Eq:HCIZ_Definition_1}
\end{align}
To evaluate this Haar average, set $\tilde{\bX}\triangleq\bP\bX$ and $\tilde{\bY}\triangleq\bQ\bY$, and denote their $a$-th columns by $\tilde{\bx}_{a}$ and $\tilde{\by}_{a}$, respectively. These orthogonal transformations preserve the replica matrices, so that $\bQ_{\sx}=\frac{1}{N}\tilde{\bX}^{\top}\tilde{\bX}$ and $\bQ_{\sy}=\frac{1}{M}\tilde{\bY}^{\top}\tilde{\bY}$. Haar invariance further implies that the conditional laws of $\tilde{\bX}$ and $\tilde{\bY}$ depend on the original replicas only through these fixed overlap matrices \cite{barbier2023fundamental,barbier2025information}. Since $\bP$ and $\bQ$ are independent, the joint density of $\tilde{\bX}$ and $\tilde{\bY}$ can therefore be written as
\begin{align}
P(\tilde{\bX}\mid\bX)P(\tilde{\bY}\mid\bY)&=P(\tilde{\bX}\mid\bQ_{\sx})P(\tilde{\bY}\mid\bQ_{\sy})\nonumber\\
&=\frac{1}{Z_{\sx}(\bQ_{\sx})Z_{\sy}(\bQ_{\sy})}\prod_{1\leq a\leq a'\leq n}\delta(NQ_{\sx,a,a'}-\tilde{\bx}_{a}^{\top}\tilde{\bx}_{a'})\prod_{1\leq a\leq a'\leq n}\delta(MQ_{\sy,a,a'}-\tilde{\by}_{a}^{\top}\tilde{\by}_{a'}),
\end{align}
where the normalization factors are
\begin{align}
Z_{\sx}(\bQ_{\sx})&\triangleq\int{\prod_{a=1}^{n}\dd\tilde{\bx}_{a}}\,\prod_{1\leq a\leq a'\leq n}\delta(NQ_{\sx,a,a'}-\tilde{\bx}_{a}^{\top}\tilde{\bx}_{a'}),\\
Z_{\sy}(\bQ_{\sy})&\triangleq\int{\prod_{a=1}^{n}\dd\tilde{\by}_{a}}\,\prod_{1\leq a\leq a'\leq n}\delta(MQ_{\sy,a,a'}-\tilde{\by}_{a}^{\top}\tilde{\by}_{a'}).
\end{align}
This density representation turns \eqref{Eq:HCIZ_Definition_1} into
\begin{align}
g_{\sh,n}&=\lim_{N\to\infty}\frac{1}{N}\log\int{\prod_{a=1}^{n}\dd\tilde{\bx}_{a}\prod_{a=1}^{n}\dd\tilde{\by}_{a}}\,\exp\left(\sum_{1\leq a,a'\leq n}\tilde{\bx}_{a}^{\top}\bA_{a,a'}\tilde{\bx}_{a'}+\tilde{\by}_{a}^{\top}\bB_{a,a'}\tilde{\by}_{a'}+\tilde{\bx}_{a}^{\top}[\bC_{a,a'},\bm{0}]\tilde{\by}_{a'}\right)\nonumber\\
&\quad\times\prod_{1\leq a\leq a'\leq n}\delta(NQ_{\sx,a,a'}-\tilde{\bx}_{a}^{\top}\tilde{\bx}_{a'})\prod_{1\leq a\leq a'\leq n}\delta(MQ_{\sy,a,a'}-\tilde{\by}_{a}^{\top}\tilde{\by}_{a'})-c_{\sx,n}-\alpha c_{\sy,n},\label{Eq:gh_n}
\end{align}
with
\begin{align}
c_{\sx,n}\triangleq\lim_{N\to\infty}\frac{1}{N}\log Z_{\sx}(\bQ_{\sx}),\quad c_{\sy,n}\triangleq\lim_{M\to\infty}\frac{1}{M}\log Z_{\sy}(\bQ_{\sy}).
\end{align}

The diagonal structure of the matrices makes the exponent in \eqref{Eq:gh_n} separable over the spin indices. For example, the $\bA$ term satisfies
\begin{align}
\sum_{1\leq a,a'\leq n}\tilde{\bx}_{a}^{\top}\bA_{a,a'}\tilde{\bx}_{a'}&=\sum_{1\leq a,a'\leq n}\sum_{i=1}^{N}A_{i,a,a'}\tilde{x}_{a,i}\tilde{x}_{a',i}=\sum_{i=1}^{N}\sum_{1\leq a,a'\leq n}A_{i,a,a'}\tilde{x}_{a,i}\tilde{x}_{a',i}=\sum_{i=1}^{N}\vec{\tilde{\bx}}_{i}\bA_{i}\vec{\tilde{\bx}}^{\top}_{i}.
\end{align}
The same rearrangement applies to the $\bB$ and $\bC$ terms. For notational convenience, define
\begin{align}
\tilde{\bA}_{i}&\triangleq\tilde{\bQ}_{\sx}-2\bA_{i},\quad\tilde{\bB}_{i}\triangleq\tilde{\bQ}_{\sy}-2\bB_{i}.
\end{align}
Equation \eqref{Eq:gh_n} consequently becomes
\begin{align}
g_{\sh,n}&=\lim_{N\to\infty}\frac{1}{N}\log\int{\prod_{a=1}^{n}\dd\tilde{\bx}_{a}\prod_{a=1}^{n}\dd\tilde{\by}_{a}}\,\exp\left(\sum_{i=1}^{N}\vec{\tilde{\bx}}_{i}\bA_{i}\vec{\tilde{\bx}}_{i}^{\top}+\sum_{i=1}^{M}\vec{\tilde{\by}}_{i}\bB_{i}\vec{\tilde{\by}}_{i}^{\top}+\sum_{i=1}^{N}\vec{\tilde{\bx}}_{i}\bC_{i}\vec{\tilde{\by}}_{i}^{\top}\right)\nonumber\\
&\quad\times\prod_{1\leq a\leq a'\leq n}\delta(NQ_{\sx,a,a'}-\tilde{\bx}_{a}^{\top}\tilde{\bx}_{a'})\prod_{1\leq a\leq a'\leq n}\delta(MQ_{\sy,a,a'}-\tilde{\by}_{a}^{\top}\tilde{\by}_{a'})-c_{\sx,n}-\alpha c_{\sy,n}\nonumber\\
&\overset{\text{(A)}}{=}\lim_{N\to\infty}\frac{1}{N}\log\int{\dd\tilde{\bQ}_{\sx}\dd\tilde{\bQ}_{\sy}}\,\exp\left(\frac{N}{2}\Tr\,\bQ_{\sx}\tilde{\bQ}_{\sx}\right)\exp\left(\frac{M}{2}\Tr\,\bQ_{\sy}\tilde{\bQ}_{\sy}\right)\prod_{i=N+1}^{M}\int{\dd\vec{\tilde{\by}}_{i}}\,\exp\left(-\frac{1}{2}\vec{\tilde{\by}}_{i}\tilde{\bB}_{i}\vec{\tilde{\by}}_{i}^{\top}\right)\nonumber\\
&\quad\times\prod_{i=1}^{N}\int{\dd\vec{\tilde{\by}}_{i}}\,\exp\left(-\frac{1}{2}\vec{\tilde{\by}}_{i}\left(\tilde{\bB}_{i}-\bC_{i}^{\top}\tilde{\bA}_{i}^{-1}\bC_{i}\right)\vec{\tilde{\by}}_{i}^{\top}\right)\prod_{i=1}^{N}\int{\dd\vec{\tilde{\bx}}_{i}}\,\exp\left(-\frac{1}{2}\left(\vec{\tilde{\bx}}_{i}-\vec{\tilde{\by}}_{i}\bC_{i}^{\top}\tilde{\bA}_{i}^{-1}\right)\tilde{\bA}_{i}\left(\vec{\tilde{\bx}}_{i}-\vec{\tilde{\by}}_{i}\bC_{i}^{\top}\tilde{\bA}_{i}^{-1}\right)^{\top}\right)\nonumber\\
&\quad-c_{\sx,n}-\alpha c_{\sy,n}+C\nonumber\\
&\overset{\text{(B)}}{=}\lim_{N\to\infty}\frac{1}{N}\log\int{\dd\tilde{\bQ}_{\sx}\dd\tilde{\bQ}_{\sy}}\exp\left(\frac{N}{2}\Tr\,\bQ_{\sx}\tilde{\bQ}_{\sx}\right)\exp\left(\frac{M}{2}\Tr\,\bQ_{\sy}\tilde{\bQ}_{\sy}\right)\prod_{i=1}^{N}|\tilde{\bQ}_{\sx}-2\bA_{i}|^{-\frac{1}{2}}\nonumber\\
&\quad\times\prod_{i=N+1}^{M}|\tilde{\bQ}_{\sy}-2\bB_{i}|^{-\frac{1}{2}}\prod_{i=1}^{N}|\tilde{\bQ}_{\sy}-2\bB_{i}-\bC_{i}^{\top}(\tilde{\bQ}_{\sx}-2\bA_{i})^{-1}\bC_{i}|^{-\frac{1}{2}}-c_{\sx,n}-\alpha c_{\sy,n}+C.
\end{align}
In step (A), we introduce the conjugate replica matrices through Fourier representations of the overlap constraints. For the $\sx$-channel,
\begin{align}
\prod_{1\leq a\leq a'\leq n}\delta(NQ_{\sx,a,a'}-\tilde{\bx}_{a}^{\top}\tilde{\bx}_{a'})&=C\int{\dd\tilde{\bQ}_{\sx}}\exp\left(\frac{N}{2}\Tr\,\bQ_{\sx}\tilde{\bQ}_{\sx}\right)\prod_{i=1}^{N}\exp\left(-\frac{1}{2}\vec{\tilde{\bx}}_{i}\tilde{\bQ}_{\sx}\vec{\tilde{\bx}}_{i}^{\top}\right),
\end{align}
where $\vec{\tilde{\bx}}_{i}\triangleq[\tilde{x}_{1,i},\cdots,\tilde{x}_{n,i}]\in\R^{1\times n}$, and $\dd\tilde{\bQ}_{\sx}$ runs over the independent entries of $\tilde{\bQ}_{\sx}$. The $\sy$-channel constraint is treated analogously. In step (B), we evaluate the resulting Gaussian integrals and substitute the definitions of $\tilde{\bA}_{i}$ and $\tilde{\bB}_{i}$. All normalization factors independent of the variational parameters are absorbed into $C$. Here $C$ collects the factors $(2\pi)^{n/2}$ produced by the Gaussian integrals and the Fourier normalizations of the overlap constraints. Exactly the same factors appear in $c_{\sx,n}$ and $c_{\sy,n}$ below, so they cancel in $g_{\sh,n}$.

Applying Laplace's method to the remaining integral over the conjugate replica matrices gives
\begin{align}
g_{\sh,n}&=\Extr_{\tilde{\bQ}_{\sx},\tilde{\bQ}_{\sy}}\Bigg(\frac{1}{2}\Tr\bQ_{\sx}\tilde{\bQ}_{\sx}+\frac{\alpha}{2}\Tr\bQ_{\sy}\tilde{\bQ}_{\sy}-\frac{1}{2N}\sum_{i=1}^{N}\log|\tilde{\bQ}_{\sx}-2\bA_{i}|-\frac{1}{2N}\sum_{i=N+1}^{M}\log|\tilde{\bQ}_{\sy}-2\bB_{i}|\nonumber\\
&\quad-\frac{1}{2N}\sum_{i=1}^{N}\log|(\tilde{\bQ}_{\sy}-2\bB_{i})-\bC_{i}^{\top}(\tilde{\bQ}_{\sx}-2\bA_{i})^{-1}\bC_{i}|\Bigg)-c_{\sx,n}-\alpha c_{\sy,n}+C.
\end{align}

It remains to evaluate the normalization constants. For the $\sx$-channel, Gaussian integration followed by Laplace's method gives
\begin{align}
c_{\sx,n}&=\lim_{N\to\infty}\frac{1}{N}\log\int{\dd\tilde{\bQ}_{\sx}}\,\exp\left(\frac{N}{2}\Tr\,\bQ_{\sx}\tilde{\bQ}_{\sx}\right)\prod_{i=1}^{N}\int{\dd\vec{\tilde{\bx}}_{i}}\,\exp\left(-\frac{1}{2}\vec{\tilde{\bx}}_{i}\tilde{\bQ}_{\sx}\vec{\tilde{\bx}}_{i}^{\top}\right)+C\nonumber\\
&=\lim_{N\to\infty}\frac{1}{N}\log\int{\dd\tilde{\bQ}_{\sx}}\,\exp\left(\frac{N}{2}\Tr\,\bQ_{\sx}\tilde{\bQ}_{\sx}\right)|\tilde{\bQ}_{\sx}|^{-\frac{N}{2}}+C\nonumber\\
&=\Extr_{\tilde{\bQ}_{\sx}}\left(\frac{1}{2}\Tr\,\bQ_{\sx}\tilde{\bQ}_{\sx}-\frac{1}{2}\log|\tilde{\bQ}_{\sx}|\right)+C\nonumber\\
&=\frac{n}{2}+\frac{1}{2}\log|\bQ_{\sx}|+C,
\end{align}
where we used the stationarity condition $\tilde{\bQ}_{\sx}=\bQ_{\sx}^{-1}$, so that $\frac{1}{2}\Tr\,\bQ_{\sx}\tilde{\bQ}_{\sx}=\frac n2$. Unlike $C$, this term does not cancel and must be retained: it is proportional to $n$, and the free entropy is obtained from the $O(n)$ part of $f_n$. The same calculation gives $c_{\sy,n}=\frac n2+\frac12\log|\bQ_{\sy}|+C$. Substituting these normalization constants into the preceding variational expression yields \eqref{Eq:Gen_Spherical_Int_Result_Lemma}, completing the derivation.

\section{Proofs of the Spectral Mapping Identities}\label{Appendix:Proof_Spectral_Mapping}

This appendix proves the identities used to express the saddle point equations in spectral form.

\begin{proof}[Proof of Lemma \ref{Lemma:Hilbert_Potential_Identities}]

We first evaluate the integral of the squared Stieltjes transform. For every $\varepsilon>0$,
\begin{align}
\int{\dd x}\,\cS^{2}_{\mu}(x+\imagUnit\varepsilon)
&=\int{\dd x}\,\left(\int{\dd\lambda_{1}}\,\frac{\mu(\lambda_{1})}{x+\imagUnit\varepsilon-\lambda_{1}}\right)\left(\int{\dd\lambda_{2}}\,\frac{\mu(\lambda_{2})}{x+\imagUnit\varepsilon-\lambda_{2}}\right)\nonumber\\
&=\int{\dd x\dd\lambda_{1}\dd\lambda_{2}}\,\mu(\lambda_{1})\mu(\lambda_{2})\frac{1}{\lambda_{1}-\lambda_{2}}\left(\frac{1}{x+\imagUnit\varepsilon-\lambda_{1}}-\frac{1}{x+\imagUnit\varepsilon-\lambda_{2}}\right)\nonumber\\
&=\int{\dd\lambda_{1}\dd\lambda_{2}}\,\mu(\lambda_{1})\mu(\lambda_{2})\frac{1}{\lambda_{1}-\lambda_{2}}\int{\dd x}\,\left(\frac{1}{x+\imagUnit\varepsilon-\lambda_{1}}-\frac{1}{x+\imagUnit\varepsilon-\lambda_{2}}\right)\nonumber\\
&=0.
\end{align}
For fixed $\varepsilon>0$, the corresponding triple integral is absolutely convergent because $\mu(\cdot)$ is a probability density, so the order of integration may be exchanged. To verify the last equality, fix $\lambda_{1}\neq\lambda_{2}$. The difference in parentheses decays as $O(x^{-2})$, and
\begin{align}
\int{\dd x}\,\left(\frac{1}{x+\imagUnit\varepsilon-\lambda_{1}}-\frac{1}{x+\imagUnit\varepsilon-\lambda_{2}}\right)&=\int_{\lambda_{2}}^{\lambda_{1}}{\dd t}\,\int{\dd x}\,\frac{1}{(x+\imagUnit\varepsilon-t)^{2}}\nonumber\\
&=\int_{\lambda_{2}}^{\lambda_{1}}{\dd t}\,\left.\left(-\frac{1}{x+\imagUnit\varepsilon-t}\right)\right|_{-\infty}^{+\infty}\nonumber\\
&=0.
\end{align}
At $\lambda_{1}=\lambda_{2}$, the partial-fraction expression is understood by continuity as $(x+\imagUnit\varepsilon-\lambda_{1})^{-2}$, whose integral also vanishes.

Decompose the Stieltjes transform as
\begin{align}
\cS_{\mu}(x+\imagUnit\varepsilon)
=\cR_{\mu}(x,\varepsilon)+\imagUnit\cI_{\mu}(x,\varepsilon).
\end{align}
The identity above gives
\begin{align}
0=\int{\dd x}\,\cS^{2}_{\mu}(x+\imagUnit\varepsilon)=\int{\dd x}\,\left(\cR^{2}_{\mu}(x,\varepsilon)-\cI^{2}_{\mu}(x,\varepsilon)+2\imagUnit\cR_{\mu}(x,\varepsilon)\cI_{\mu}(x,\varepsilon)\right).
\end{align}
Taking the imaginary part yields
\begin{align}
0=\int{\dd x}\,\cR_{\mu}(x,\varepsilon)\cI_{\mu}(x,\varepsilon),\label{Eq:int}
\end{align}
which holds for every $\varepsilon>0$.

Since the Hilbert transform is bounded on $L^{2}(\R)$, the assumption $\mu(\cdot)\in L^{2}(\R)$ implies $\cH_{\mu}(\cdot)\in L^{2}(\R)$. Here, $L^{2}(\R)$ denotes the space of square-integrable functions on $\R$. As $\varepsilon\to0^{+}$, the $L^{2}$ boundary-value form of the Sokhotski-Plemelj formula gives
\begin{align}
\cR_{\mu}(x,\varepsilon)&\to\pi\cH_{\mu}(x),\quad\cI_{\mu}(x,\varepsilon)\to-\pi\mu(x).
\end{align}
Both convergences are in $L^{2}(\R)$. The two sequences are therefore bounded in $L^{2}(\R)$, and the Cauchy-Schwarz inequality gives
\begin{align}
\cR_{\mu}(x,\varepsilon)\cI_{\mu}(x,\varepsilon)\to-\pi^{2}\mu(x)\cH_{\mu}(x)
\end{align}
in $L^{1}(\R)$. Here, $L^{1}(\R)$ denotes the space of absolutely integrable functions on $\R$. Taking $\varepsilon\to0^{+}$ in \eqref{Eq:int}, we obtain
\begin{align}
0&=\lim_{\varepsilon\to0^{+}}\int{\dd x}\,\cR_{\mu}(x,\varepsilon)\cI_{\mu}(x,\varepsilon)=-\pi^{2}\int{\dd x}\,\mu(x)\cH_{\mu}(x).
\end{align}
Consequently,
\begin{align}
\int{\dd x}\,\mu(x)\cH_{\mu}(x)=0,
\end{align}
which proves the lemma.

\end{proof}

\begin{proof}[Proof of Lemma \ref{Lemma:Spectral_Mapping}]

The two identities follow from the equilibrium equation for the limiting spectral measure. When $\frac{M}{N}\to\alpha$, the joint eigenvalue density in \eqref{Eq:Joint_Eigenvalue_Density} corresponds, to leading order in $N$, to a log-gas on $(0,\infty)$ with effective potential $V_{\Sf{eff}}(\gamma)\triangleq V(\gamma)-(\alpha-1)\log\gamma$, where the logarithmic term arises from the rectangular Jacobian. Under standard regularity and confinement assumptions, the empirical spectral measure converges to the equilibrium measure associated with this potential; see \cite[Section 2.6]{anderson2010introduction}. We derive the required identities below from the corresponding Coulomb-gas saddle point equation.

Let $\{\gamma_{i}\}_{i=1}^{N}$ denote the nonzero eigenvalues of $\bZ_{0}\bZ_{0}^{\top}$. For the rectangular invariant ensemble, their joint density is given by \cite{di2003rectangular,potters2020first}
\begin{align}
\label{Eq:Joint_Eigenvalue_Density}
P(\gamma_{1},\cdots,\gamma_{N})&\propto\prod_{i=1}^{N}\left[\gamma_{i}^{\frac{1}{2}(M-N-1)}\exp\left(-\frac{N}{2} V(\gamma_{i})\right)\right]\prod_{\substack{i,j=1\\i<j}}^{N}|\gamma_{i}-\gamma_{j}|.
\end{align}
The corresponding normalization constant can be written as
\begin{align}
Z\triangleq&\int{\prod_{i=1}^{N}\dd \gamma_{i}}\,\exp(-N^{2}\cL),
\end{align}
where
\begin{align}
\cL\triangleq&-\frac{1}{N^{2}}\log\left(\prod_{i=1}^{N}\left[\gamma_{i}^{\frac{1}{2}(M-N-1)}\exp\left(-\frac{N}{2}V(\gamma_{i})\right)\right]\prod_{\substack{i,j=1\\j>i}}^{N}|\gamma_{i}-\gamma_{j}|\right)\nonumber\\
&=-\frac{M-N-1}{2N^{2}}\sum_{i=1}^{N}\log \gamma_{i}+\frac{1}{2N}\sum_{i=1}^{N}V(\gamma_{i})-\frac{1}{2N^{2}}\sum_{\substack{i,j=1\\j\neq i}}^{N}\log|\gamma_{i}-\gamma_{j}|.
\end{align}
The saddle point condition $\frac{\partial\cL}{\partial \gamma_{i}}=0$ gives
\begin{align}
V'(\gamma_{i})-\frac{M-N-1}{N}\frac{1}{\gamma_{i}}&=\frac{2}{N}\sum_{\substack{j=1\\j\neq i}}^{N}\frac{1}{\gamma_{i}-\gamma_{j}}.\label{Eq:Section:Saddle_Point}
\end{align}

For $z\in\C\setminus\{\gamma_1,\cdots,\gamma_N\}$, define the finite-$N$ resolvent
\begin{align}
\omega_N(z)&\triangleq\frac{1}{N}\sum_{i=1}^{N}\frac{1}{z-\gamma_{i}}.
\label{Eq:Section:Resolvent}
\end{align}
Multiplying \eqref{Eq:Section:Saddle_Point} by $\frac{1}{z-\gamma_{i}}$ and summing over $i$ gives
\begin{align}
\sum_{i=1}^{N}\frac{V'(\gamma_{i})}{z-\gamma_{i}}-\frac{M-N-1}{N}\sum_{i=1}^{N}\frac{1}{z-\gamma_{i}}\frac{1}{\gamma_{i}}&=\frac{2}{N}\sum_{\substack{i,j=1\\j\neq i}}^{N}\frac{1}{z-\gamma_{i}}\frac{1}{\gamma_{i}-\gamma_{j}}\nonumber\\
&=\frac{1}{N}\left(\sum_{\substack{i,j=1\\j\neq i}}^{N}\frac{1}{z-\gamma_{i}}\frac{1}{\gamma_{i}-\gamma_{j}}+\sum_{\substack{j,i=1\\i\neq j}}^{N}\frac{1}{z-\gamma_{j}}\frac{1}{\gamma_{j}-\gamma_{i}}\right)\nonumber\\
&=\frac{1}{N}\sum_{\substack{i,j=1\\j\neq i}}^{N}\frac{1}{z-\gamma_{i}}\frac{1}{z-\gamma_{j}}\nonumber\\
&=\frac{1}{N}\left(\sum_{i,j=1}^{N}\frac{1}{z-\gamma_{i}}\frac{1}{z-\gamma_{j}}-\sum_{i=1}^{N}\frac{1}{(z-\gamma_{i})^{2}}\right)\nonumber\\
&=N\omega_{N}^{2}(z)+\omega_{N}'(z).\label{Eq:Section:Func_1}
\end{align}
For $z$ in the analytic domain of $V(\cdot)$ and off the spectral support, multiplying \eqref{Eq:Section:Resolvent} by $V'(z)$ and adding and subtracting $V'(\gamma_{i})$ in the numerator gives
\begin{align}
\omega_N(z)V'(z)&=\frac{1}{N}\sum_{i=1}^{N}\frac{V'(z)}{z-\gamma_{i}}=\frac{1}{N}\sum_{i=1}^{N}\frac{V'(z)-V'(\gamma_{i})}{z-\gamma_{i}}
+\frac{1}{N}\sum_{i=1}^{N}\frac{V'(\gamma_{i})}{z-\gamma_{i}}.
\label{Eq:Section:Func_2}
\end{align}
Equations \eqref{Eq:Section:Func_1}-\eqref{Eq:Section:Func_2} then give
\begin{align}
N\omega_{N}^{2}(z)+\omega_{N}'(z)&=N\omega_{N}(z)V'(z)-\sum_{i=1}^{N}\frac{V'(z)-V'(\gamma_{i})}{z-\gamma_{i}}-\frac{M-N-1}{N}\sum_{i=1}^{N}\frac{1}{z}\left(\frac{1}{z-\gamma_{i}}+\frac{1}{\gamma_{i}}\right)\nonumber\\
&=N\omega_{N}(z)V'(z)-\sum_{i=1}^{N}\frac{V'(z)-V'(\gamma_{i})}{z-\gamma_{i}}-(M-N-1)\frac{1}{z}\omega_{N}(z)+(M-N-1)\frac{1}{z}\omega_{N}(0).
\end{align}

After division by $N$, let $N\to\infty$ with $\frac{M}{N}\to\alpha$. For the saddle point configurations considered here, the empirical spectral measures converge to $\mu(\cdot)$ and their spectra remain in a common compact subset of $(0,\infty)$. Hence, for fixed $z$ off the real axis, we have
\begin{align}
\omega_N(z)&\to\cS_{\mu}(z),\quad\omega_N(0)\to-\E_{\gamma\sim\mu}\frac{1}{\gamma},\quad\frac{1}{N}\omega_N'(z)\to0,\quad\frac{1}{N}\sum_{i=1}^{N}\frac{V'(z)-V'(\gamma_{i})}{z-\gamma_{i}}\to\E_{\gamma\sim\mu}\frac{V'(z)-V'(\gamma)}{z-\gamma}.
\end{align}
Therefore,
\begin{align}
\cS_{\mu}^{2}(z)&=\underbrace{\left[V'(z)-(\alpha-1)\frac{1}{z}\right]}_{\triangleq D(z)}\cS_{\mu}(z)-\underbrace{\left[\E_{\gamma\sim\mu}\frac{V'(z)-V'(\gamma)}{z-\gamma}+(\alpha-1)\frac{1}{z}\E_{\gamma\sim\mu}\frac{1}{\gamma}\right]}_{\triangleq F(z)}.
\end{align}
The resulting quadratic equation gives
\begin{align}
\cS_{\mu}(z)&=\frac{1}{2}D(z)\pm\frac{1}{2}\sqrt{D^{2}(z)-4F(z)}.
\end{align}
The appropriate branch is the one agreeing with the Stieltjes transform in its analytic domain. If the displayed formula extends to infinity, it satisfies $\cS_{\mu}(z)\sim\frac{1}{z}$ as $|z|\to\infty$.

For $x$ in the interior of the spectral support, the Sokhotski-Plemelj formula and the upper-half-plane boundary value of this branch give
\begin{align}
\lim_{\varepsilon\to0^{+}}\cS_{\mu}(x+\imagUnit\varepsilon)&=\pi\cH_{\mu}(x)-\imagUnit\pi\mu(x)=\frac{1}{2}D(x)-\imagUnit\frac{1}{2}\sqrt{4F(x)-D^{2}(x)}.\label{Eq:SP_R}
\end{align}
The real part of \eqref{Eq:SP_R} gives
\begin{align}
\cH_{\mu}(x)&=\frac{1}{\pi}\lim_{\varepsilon\to0^{+}}\Re\left[\cS_{\mu}(x+\imagUnit\varepsilon)\right]=\frac{1}{2\pi}D(x)=\frac{1}{2\pi}\left[V'(x)-(\alpha-1)\frac{1}{x}\right],
\end{align}
and hence
\begin{align}
V'(x)&=2\pi\cH_{\mu}(x)+(\alpha-1)\frac{1}{x}.\label{Eq:V_x}
\end{align}
The imaginary part of \eqref{Eq:SP_R} gives
\begin{align}
\mu(x)&=-\frac{1}{\pi}\lim_{\varepsilon\to0^{+}}\Im\left[\cS_{\mu}(x+\imagUnit\varepsilon)\right]=\frac{1}{2\pi}\sqrt{4F(x)-D^{2}(x)}\nonumber\\
&=\frac{1}{\pi}\sqrt{\E_{\gamma\sim\mu}\frac{V'(x)-V'(\gamma)}{x-\gamma}+(\alpha-1)\frac{1}{x}\E_{\gamma\sim\mu}\frac{1}{\gamma}-\pi^{2}\cH_{\mu}^{2}(x)}.
\end{align}
This identity can be rewritten as
\begin{align}
\E_{\gamma\sim\mu}\frac{V'(x)-V'(\gamma)}{x-\gamma}&=\pi^{2}\left(\mu^{2}(x)+\cH_{\mu}^{2}(x)\right)-(\alpha-1)\frac{1}{x}\E_{\gamma\sim\mu}\frac{1}{\gamma}.\label{Eq:E_s}
\end{align}
Equations \eqref{Eq:V_x} and \eqref{Eq:E_s} give the two identities stated in the lemma, completing the derivation.

\end{proof}

\clearpage

\end{document}